\documentclass[11pt]{article}

\usepackage[margin=1in]{geometry}

\usepackage{rotating}
\usepackage{makecell}
\usepackage{booktabs}
\usepackage{multirow}

\usepackage{tcolorbox}

\newcolumntype{C}[1]{>{\centering\arraybackslash}p{#1}}

\usepackage[multiple]{footmisc}

\usepackage{float}
\usepackage{comment}
\usepackage{soul}
\usepackage{wrapfig}

\usepackage{algorithmicx}
\usepackage[noend]{algpseudocode}
\usepackage{algpseudocodex}

\usepackage{amsmath,amsthm,amssymb,mathtools}

\usepackage{graphicx}
\usepackage{dsfont}
\usepackage{tikz}
\usetikzlibrary{calc}
\usetikzlibrary{shapes.geometric}
\usepackage{paralist}
\usepackage{nicefrac}

\usepackage{wrapfig}

\usepackage{float}
\usepackage{comment}
\usepackage{algorithm}
\usepackage[noend]{algpseudocode}

\usepackage{thmtools}
\usepackage{thm-restate}

\usepackage{caption}
\usepackage[labelformat=simple]{subcaption}

\usepackage{aliascnt}

\usepackage{xcolor}

\usepackage[backref=page]{hyperref}
\hypersetup{colorlinks=true,citecolor=blue,linkcolor=blue}
\hypersetup{colorlinks=true}
\hypersetup{linkcolor=[rgb]{.7,0,0}}
\hypersetup{citecolor=[rgb]{0,.7,0}}
\hypersetup{urlcolor=[rgb]{.7,0,.7}}

\declaretheorem[numberwithin=section]{theorem}
\declaretheorem[sibling=theorem, style=definition]{definition}
\declaretheorem[sibling=theorem]{lemma}

\declaretheorem[sibling=theorem]{observation}

\declaretheorem[sibling=theorem]{corollary}
\declaretheorem[sibling=theorem, style=definition]{remark}
\declaretheorem[sibling=theorem]{proposition}
\declaretheorem[sibling=theorem, style=definition]{example}

\makeatletter
\patchcmd{\ALG@step}{\addtocounter{ALG@line}{1}}{\refstepcounter{ALG@line}}{}{}
\newcommand{\ALG@lineautorefname}{Line}
\makeatother

\newcommand{\poly}{\operatorname{poly}}

\makeatletter
\newcommand\Ps@textstyle[2]{\mathbb{P}_{#1}\left[{#2}\right]}
\newcommand\Es@textstyle[2]{\mathbb{E}_{#1}\left[{#2}\right]}
\newcommand\Ps[2]{%
  \mathchoice 
  {\underset{{#1}}{\mathbb{P}}\left[{#2}\right]}
  {\Ps@textstyle{#1}{#2}}
  {\Ps@textstyle{#1}{#2}}
  {\Ps@textstyle{#1}{#2}}
}
\newcommand\Es[2]{%
  \mathchoice 
  {\underset{{#1}}{\mathbb{E}}\left[{#2}\right]}
  {\Es@textstyle{#1}{#2}}{\Es@textstyle{#1}{#2}}{\Es@textstyle{#1}{#2}}
}
\makeatother

\hypersetup{pdftitle={Structural-Complexities-of-Matching-Mechanisms}}

\algdef{SE}[DOWHILE]{Do}{doWhile}{\algorithmicdo}[1]{\algorithmicwhile\ #1}%

\usepackage{comment}

\usepackage{sidecap}

\usepackage{csquotes}
\usepackage{colortbl}

\usepackage{sansmath}

\usepackage{tikz}
\usetikzlibrary{automata, positioning, calc, fit, shapes.misc}
\usetikzlibrary{matrix}
\usetikzlibrary{positioning}
\usetikzlibrary{arrows, decorations.pathmorphing}
\usetikzlibrary{decorations.pathreplacing,calligraphy}

\newcommand{\M}{\mathcal{M}}
\newcommand{\Menu}{\mathsf{Menu}}

\newcommand{\T}{\mathcal{T}}

\newcommand{\StabB}{\mathsf{StabB}}

\newcommand{\Appls}{\mathcal{A}}
\newcommand{\Insts}{\mathcal{I}}

\newcommand{\APDA}{\mathsf{APDA}}
\newcommand{\DA}{\mathsf{DA}}
\newcommand{\IPDA}{\mathsf{IPDA}}
\newcommand{\TTC}{{{\mathsf{TTC}}}}
\newcommand{\SD}{{{\mathsf{SD}}}}
\newcommand{\SDrot}{{{\mathsf{SD}^{\mathsf{rot}}}}}

\newcommand{\match}{{\mathsf{match}}}
\renewcommand{\check}{{\mathsf{check}}}
\newcommand{\Certs}{{\mathcal{C}}}

\newcommand{\chain}{\mathsf{chain}}
\newcommand{\stab}{\mathsf{stab}}
\newcommand{\stabmatch}{\mathsf{unrej}}
\newcommand{\prop}{\mathsf{nodes}}
\newcommand{\UnrejGr}{\mathsf{UnrejGr}}

\newcommand{\ImprGrA}{\mathsf{ImprGr}^{\Appls}}
\newcommand{\ImprGrI}{\mathsf{ImprGr}^{\Insts}}

\newcommand{\hunrej}{h^{\mathsf{unrej}}}
\newcommand{\hmatch}{h^{\mathsf{match}}}

\newcommand{\defeq}{\stackrel{\mathrm{def}}{=}}

\usepackage{enumitem}

\begin{document}

\title{
Structural Complexities of Matching Mechanisms\thanks{We thank Nick Arnosti, Uma Girish, Aram Grigoryan, Andreas Haupt, Ori Heffetz, Nicole Immorlica, Jacob Leshno, Irene Lo, Kunal Mittal, Madhu Sudan, S.\ Matthew Weinberg, and seminar participants at Princeton and Harvard for illuminating discussions and helpful feedback.}
}
\author{
  Yannai A. Gonczarowski\thanks{Department of Economics and Department of Computer Science, Harvard University | \emph{E-mail}: \href{mailto:yannai@gonch.name}{yannai@gonch.name}.}
  \and
  Clayton Thomas\thanks{Microsoft Research | \emph{E-mail}: \href{mailto:clathomas@microsoft.com}{clathomas@microsoft.com}.}
}
\date{March 30, 2024}

\newcommand{\questionTextTypeToMatching}{How can one applicant affect the outcome matching?}
\newcommand{\questionTextTypeToMenu}{How can one applicant affect another applicant's set of options?}
\newcommand{\questionTextRepresentation}{How can the outcome matching be represented / communicated?}
\newcommand{\questionTextVerification}{How can the outcome matching be verified?}

\begin{titlepage}
\maketitle
\begin{abstract}
  
We study various novel complexity measures for two-sided matching mechanisms, applied to the two canonical strategyproof matching mechanisms, Deferred Acceptance ($\DA$) and Top Trading Cycles ($\TTC$). 
Our metrics are designed to capture the complexity of various structural (rather than computational) concerns, in particular ones of recent interest within economics.
We consider a unified, flexible approach to formalizing our questions: Define a protocol or data structure performing some task, and bound the number of bits that it requires.
Our main results apply this approach to four questions of general interest; for mechanisms matching applicants to institutions, our questions are:
\begin{enumerate}[label=(\arabic*),ref=(\arabic*)]
  \item {\questionTextTypeToMatching}
  \item {\questionTextTypeToMenu}
  \item {\questionTextRepresentation}
  \item {\questionTextVerification}
\end{enumerate}

Holistically, our results show that $\TTC$ is more complex than $\DA$, formalizing previous intuitions that $\DA$ has a simpler structure than $\TTC$.
For question~(2), our result gives a new combinatorial characterization of which institutions are removed from each applicant's set of options when a new applicant is added in $\DA$; this characterization may be of independent interest.
For question~(3), our result gives new tight lower bounds proving that 
the relationship between the matching and the priorities is more complex in $\TTC$ than in $\DA$.
We nonetheless showcase that this higher complexity of $\TTC$ is nuanced: By constructing new tight lower-bound instances and new verification protocols, we prove that $\DA$ and $\TTC$ are comparable in complexity under questions~(1) and~(4). 
This more precisely delineates the ways in which $\TTC$ is more complex than $\DA$, and emphasizes that diverse considerations must factor into gauging the complexity of matching mechanisms.

\end{abstract}
\thispagestyle{empty}
\end{titlepage}

\maketitle

\tableofcontents
\thispagestyle{empty}

\clearpage
\pagenumbering{arabic}

\section{Introduction}
\label{sec:intro}

School districts in many cities employ school-choice mechanisms, i.e., algorithms that produce a matching of students to schools on the basis of students' reported preferences and schools' priorities.\footnote{
  As is customary in many papers, we use the word \emph{priorities} to refer to schools' ``preferences'' over the students (or more generally, the institutions' preferences over the applicants). 
  This reflects the fact that we assume the priorities are constant and fixed ahead of time (e.g., based on factors such as the distance from the school to a student's residence, or whether the student already has a sibling attending the school).
}
Two such mechanisms are widely studied and deployed: Top Trading Cycles \cite{ShapleyS74} (henceforth $\TTC$, the canonical mechanism finding an optimal matching for students), and Deferred Acceptance \cite{GaleS62} (henceforth $\DA$, the canonical mechanism finding a fair and stable matching).
In contrast to most of the algorithms people encounter throughout daily life---from search engines to ride hailing platforms---school districts give detailed explanations of how these algorithms work, and 
expect students and parents to understand them, participate correctly, and trust the outcomes. 
In reality, however, many participants find these mechanisms confusing, unpredictable, and hard to trust \cite{Kasman19opportunities,RobertsonNS21}.

Given the above, it is clear that \emph{non-computational} desiderata such as simplicity, transparency, and explainability are first-order concerns in matching mechanisms \cite{Li17, AkbarpourL20,GonczarowskiHT22}. 
The economics literature has recently asked a number of questions on this theme, for instance: 
In what ways are these mechanism unpredictable 
\cite{Arnosti20BlogPredictable}? 
How can a student's priorities at different schools be used to directly explain their match \cite{AzevedoL16, LeshnoL21}?
To what extent can players detect whether school districts deviate from the promised mechanisms \cite{Moller22, MollerG23, HakimovR23}?
In fact, such concerns are seemingly instrumental in real-world policy decisions; for instance, when Boston Public School compared $\DA$ and $\TTC$, they wrote \cite{BPS05}:
\begin{displayquote}
  ``The Deferred Acceptance Algorithm will \emph{best serve Boston families} as a centralized procedure\ldots\ 
  {[}In $\DA$,] students will receive their highest choice among their school choices for which they have \emph{high enough priority} to be assigned.''

  ``{[}In $\TTC$,] the behind the scenes mechanized trading makes the student assignment process less transparent [than in $\DA$].''
\end{displayquote}
In the end, this flagship district chose $\DA$ in favor of $\TTC$, seemingly largely for reasons of transparency and explainability, and many other school districts have followed suit.


Comparisons and questions such as these call out for precise, quantitative measures of the various complexities of these mechanisms.
Nonetheless, suitable complexity results capturing the relevant distinction between $\TTC$ and $\DA$ have remained elusive. 
In fact, known formal results have often produced ways in which $\TTC$ might be considered \emph{less complex} than $\DA$, missing the concern raised by \cite{BPS05} quoted above.\footnote{
    For instance, when preferences and priorities must be communicated between the students and schools, the communication complexity of running $\TTC$ is less than that of running $\DA$ \cite{GonczarowskiNOR19}.
    Additionally, \cite{GonczarowskiHT22} provide a complexity-theoretic sense in which $\TTC$'s strategyproofness may be more apparent than $\DA$'s.
    Moreover, in $\TTC$ (but not $\DA$), an applicant can only change the matching if her own match changes (a property called \emph{nonbossiness}).
}

In this paper, we consider several \emph{structural} complexity measures addressing the above concerns.
We draw upon the techniques of theoretical computer science, traditionally studying topics such as verification (i.e., nondeterministic computation) or cryptographic security, and extend these tools to address questions of predictability, explainability, and transparency.
We consider a unified, flexible framework for formalizing and answering the questions at hand: 
First, we define a protocol or data structure that captures the task at hand; second, we bound the number of bits required
by this protocol or data structure.
We consider mechanisms matching applicants to institutions, and apply our framework to four main questions:
\begin{enumerate}[label=(\arabic*),ref=(\arabic*)]
  \itemsep0em
  \item \nameref{ppg:type-to-match} Question: \questionTextTypeToMatching
  \item \nameref{ppg:type-to-menu} Question: \questionTextTypeToMenu 
  \item \nameref{ppg:representation} Question: \questionTextRepresentation 
  \item \nameref{ppg:verification} Question: \questionTextVerification 
\end{enumerate}

The answers to these questions both help to illuminate the mathematical structure of matching mechanisms, and are of concrete economic interest.
First, our \nameref{ppg:type-to-match} and \nameref{ppg:type-to-menu} Questions provide a stylized theoretical lens into the predictability question of \cite{Arnosti20BlogPredictable}, since they address how ``small'' changes to the inputs affect the mechanism.\footnote{
    Our \nameref{ppg:type-to-match} and \nameref{ppg:type-to-menu} Questions also provide quantitative complexity results that complement the traditional economics literature, which studies qualitative results showing that different matching mechanisms satisfy or do not satisfy various binary properties (such as strategyproofness, nonbossiness, and population monotonicity; see \autoref{sec:related}).}
Second, our \nameref{ppg:representation} and \nameref{ppg:verification} Questions provide
practically-relevant---but theoretically principled---insights towards questions of how matches can be explained in terms of priorities \cite{LeshnoL21, AzevedoL16}, and how the outcome of the mechanism can be trusted and made transparent \cite{Moller22, MollerG23, HakimovR23}.

Taken as a whole, our results show that $\TTC$ is more complex than $\DA$, corroborating some of the concerns of \cite{BPS05} and making them precise.
In particular, under our \nameref{ppg:type-to-menu} and \nameref{ppg:representation} Questions, we give as-large-as-possible separations between the complexity of $\TTC$ and $\DA$.
Along the way, we provide several new tight lower-bound constructions that extend or complement prior results of \cite{GonczarowskiHT22,AzevedoL16,LeshnoL21}, and provide novel new characterizations of how one applicant can affect others in certain mechanisms.
The higher complexity of $\TTC$ is, however, nuanced: Under our \nameref{ppg:type-to-match} and \nameref{ppg:verification} Questions, we find that the mechanisms are comparably complex. Overall, we use our framework to precisely delineate all of these complexities of the two mechanisms in a cohesive model.

\paragraph{Outcome-Effect}\label{ppg:type-to-match}\textbf{\hspace{-0.8em} Question\ \ (\questionTextTypeToMatching{} \big/ \autoref{sec:gtc})\ } 
To give exposition into our framework, consider the different ways one could formalize our \nameref{ppg:type-to-match} Question.
A na{\"i}ve approach might be to bound the ``magnitude'' of one applicant's effect on the matching. However, as observed in \cite{Arnosti20BlogPredictable}, it is not hard to see that even in very simple mechanisms, changing one applicant's preference report can change the match of \emph{every} applicant.
Indeed, consider serial dictatorship (henceforth, $\SD$), a special case of both $\TTC$ and $\DA$ in which all institutions have the same priority list.
In this mechanism, each of the applicants $d_1,d_2,d_3,d_4$ (say, in this order) is permanently matched to her top-ranked not-yet-matched institution. Suppose their reported preference lists are, respectively:
\[P_1 : h_1\succ h_2 \succ \emptyset
\qquad P_2 : h_2 \succ h_3 \succ \emptyset
\qquad P_3 : h_3 \succ h_4 \succ \emptyset
\qquad P_4 : h_4 \succ h_1 \succ \emptyset.
\]
$\SD$ matches $d_i$ to $h_i$ for every $i$.\footnote{
To avoid using the common generic variables $a$ and $i$ for applicants and institutions, we denote applicants with the letter $d$ (mneumonic: doctor) and institutions with the letter $h$ (mneumonic: hospital).
}
Now, observe that if $d_1$ changes her preferences from $P_1$ to $P_1' : h_2 \succ h_1$, then every applicant's match changes: the new matching matches $d_i$ to $h_{i+1}$ (with indices modulo $4$) for each $i=1,\ldots,4$. So a seemingly small change to one applicant's report can have a large effect on the outcome, even for the quite simple mechanism $\SD$.

We therefore look not at the magnitude, but rather at the \emph{structure}, of one applicant's effect on the matching. 
As a warm-up, in \autoref{thrm:gtc-SD-UB} we characterize the function $\SD(\cdot,P_{-d})$ from one applicant $d$'s report to the outcome matching of $\SD$.\footnote{As is customary, $P_{-d}$ denotes a profile of preferences for all applicants except $d$.}
In particular, we show that for all $P_{-d}$, there exists a data structure representing this function in only $\widetilde O(n)$ bits; note that since this matches the number of bits needed to record a single matching, this bit-complexity is as low as possible. 
This shows that, while changing one applicant's report can change many other applicants' matches, it can only do so in a structured way.

In contrast, for each $f \in \{ \TTC, \DA \}$, we prove that any representation of the function $f(\cdot, P_{-d})$ \emph{requires $\Omega(n^2)$ bits}, matching (up to logarithmic factors) the trivial solution of storing the entire preference profile $P_{-d}$, and therefore as high as possible.
This gives one precise, natural, and generic sense in which the complexity of how one applicant can affect the matching is high:
the functions $\{ f(\cdot, P_{-d}) \}_{P_{-d}}$ representing this affect cannot (up to lower-order terms) be described in any more-compact fashion than ``the function that results from preference profile $P_{-d}$.'' 

\begin{theorem}[Informal; Combination of \autoref{thrm:gtc-TTC-LB}, \autoref{thrm:gtc-DA-LB}, and \autoref{thrm:gtc-IPDA}]
  For both $\TTC$ and $\DA$, the complexity of how one applicant can affect the outcome matching is (nearly) as high as possible.\footnote{
    $\DA$ comes in two variants: applicant-proposing $\DA$ ($\APDA$), and institution-proposing $\DA$ ($\IPDA$). 
    This result holds for both, however establishing this requires a very different proof for each (\autoref{thrm:gtc-DA-LB}, which follows directly from \cite{GonczarowskiHT22}, for $\APDA$; and \autoref{thrm:gtc-IPDA}, which requires an involved new construction, for $\IPDA$).
  }
\end{theorem}

These results are quite novel among complexity results for matching mechanisms; we are not aware of any similar formal complexity bounds.
We prove these results directly by constructing a large set of preference profiles for all applicants other than $d$ (possibly using intermediate mechanisms, e.g. in \autoref{thrm:gtc-SDrot-LB}), and showing that each distinct preference profile defines a distinct function. Each such construction must carefully exploit the properties of the relevant mechanism. 

Our results establish a new complexity gap between the simple mechanism $\SD$ and the more complex $\TTC$ and $\DA$.
This gap is (nearly) as high as possible, since (as discussed above) this complexity measure is always at least $\widetilde\Omega(n)$ and at most $\widetilde O(n^2)$.
This is representative of the strong separations we achieve in our framework: we think of  $\widetilde O(n)$ as low complexity / structured, and $\Omega(n^2)$ as high complexity / non-structured.
Note additionally that ``on average'', i.e. \emph{per-applicant}, the gap between these two complexities is exponential ($O(\log(n))$ vs.\ $\Omega(n)$). Moreover, this is an economically meaningful gap in real-world markets: explaining something to students using ``just a few bits per school'' may be far more tractable than an explanation requiring $\Omega(n)$ bits per school.

\paragraph{Options-Effect}\label{ppg:type-to-menu}\textbf{\hspace{-0.4em}Question\ \ (\questionTextTypeToMenu{} \big/\ \autoref{sec:type-to-menu})\ } 
While our results in \autoref{sec:gtc} separate the simple mechanism $\SD$ from $\TTC$ and $\DA$, they do not distinguish between the complexity of $\TTC$ and $\DA$ themselves.
To get such a separation, and to provide a deeper look into how applicants affect each other in these more involved mechanisms, we investigate how one applicant's report can affect another's set of obtainable options.
The natural definition of an applicant $d$'s set of obtainable options is \textbf{$\boldsymbol{d}$'s menu given $\boldsymbol{P_{-d}}$}, which is defined as the set of institutions that $d$ could be matched to under some report, when holding the reports of all the other applicants fixed at $P_{-d}$. Formally, for a mechanism $f$, we define this set as $\Menu^f_d(P_{-d}) = \bigl\{ f_d(P_d', P_{-d}) ~\big|~ P_d' \in \T_d\bigr\}$.\footnote{
We let $f_d(\cdot)$ denote the match of $d$ in the mechanism (i.e., if $\mu = f(P)$, then $f_d(P) = \mu(d)$), and we let $\T_d$ denote the set of all possible preference lists of $d$.
}
We consider the function $\Menu^{f}_{d_\dagger}\bigl(\cdot, P_{-\{d_*, d_\dagger\}}\bigr)$, which maps from the preference list of ${d_*}$ to the menu of $d_\dagger$ (while the preferences $P_{-\{d_*, d_\dagger\}}$ of all other applicants are held fixed).
We measure the complexity of this function as before: we ask whether any data structure can represent this function in a more compact way than storing all of $P_{-\{d_*, d_\dagger\}}$.
Our results under this complexity measure give a separation between $\DA$ and $\TTC$:

\begin{theorem}[Informal; Combination of \autoref{thrm:type-to-menu-TTC} and \autoref{thrm:type-to-menu-DA}]
  For $\TTC$, the complexity of how one applicant can affect another's menu is (nearly) as \textbf{high} as possible.
  For $\DA$, this complexity is (nearly) as \textbf{low} as possible.
\end{theorem}

For $\TTC$, this result requires another novel lower bound construction.
To prove this result for $\DA$, we uncover a novel combinatorial characterization of the set of institutions that are removed from applicant $d_\dagger$'s menu when applicant $d_*$ is added to the market.\footnote{
  The fact that such a set of institutions exists (without any characterization of the set) follows from the classical economic property of \emph{population monotonicity}, which $\DA$ satisfies, but $\TTC$ does not; see \autoref{sec:related}.
}
This characterization may be of independent interest;
we also exploit it to give an additional new characterization of the ``pairwise menu'' of applicants $d_*$ and $d_\dagger$ (formally, the set of pairs $(h_*, h_\dagger)$ such that there exists some $(P_{d_*}, P_{d_\dagger})$ such that $d_*$ and $d_\dagger$ respectively match to $h_*$ and $h_\dagger$, holding $P_{-\{d_*, d_\dagger\}}$ fixed).

\paragraph{Representation}\label{ppg:representation}\textbf{\hspace{-0.4em}Question\ \ (\questionTextRepresentation{} \big/\ \autoref{sec:compression-complexity})\ } 
Next, we switch gears and investigate a complexity question concerning the outcome matching (instead of the complexity of one applicant's effect on the mechanism).
We start by asking how the outcome can be communicated, assuming all preferences are already known by the mechanism designer.
Under a standard model with $n$ applicants and $n$ institutions, this question is trivial: $\widetilde O(n)$ bits are both necessary and sufficient to communicate the matching.

Given the above, to capture the insights we seek, we make two natural assumptions.
First, we consider sets of applicants~$\Appls$ and institutions $\Insts$ with $|\Appls|\gg|\Insts|$ (since, e.g., in student-school matchings many students will attend the same school), say with $|\Appls|=\poly(|\Insts|)$ for concreteness. 
Second, we assume that all applicants start off knowing---in addition to their own preferences---all their priorities at all institutions (quite plausible in school matching where priorities are often determined by polices).\footnote{
  All our results are the same regardless of whether we assume each applicant knows only her own priority at each institution (say, because it is communicated privately to the applicant ahead of time) or knows all other applicants' priorities at all institutions (say, because the priorities are determined by set and publicly known policies).
}

Under these assumptions, it turns out that for some mechanisms, one can simultaneously communicate to each applicant her own match using a single global message containing \emph{less than one bit per applicant}.
Indeed, for $\DA$, a representation of the matching of the above form follows readily from well-known properties of stable matchings. 
Specifically, as observed by \cite{AzevedoL16}, for any stable matching $\mu$ that could possibly result from $\DA$, there exist \emph{cutoffs} $c_h \in \Appls \cup \{\emptyset\}$, one for each institution $h$, such that each applicant is matched in $\mu$ to her highest-ranked institution $h$ whose cutoff is below her priority score; formally, $\mu(d) = \max_{P_d} \bigl\{ h ~\big|~ d \succ_h c_h \bigr\}$.\footnote{
   Here, $\max_{P_d} H$ for a set of institutions $H$ denotes the highest-ranked $h\in H$ according to preference list $P_d$. This representation of the matching is a simple reformulation of the definition of a stable matching; see \autoref{thrm:representing-DA-cutoff-lemma}.
}
Assuming each applicant knows the priorities, this profile of cutoffs $(c_h )_{h \in \Insts}$ represents to each applicant her own match using altogether $\widetilde O\bigl(|\Insts|\bigr)$ bits (and it is not hard to show that this is tight).

\cite{LeshnoL21} have previously studied a variant of this question for $\TTC$, and showed that there exists a cutoff-like representation of the matching, but one requiring a cutoff for each \emph{pair} of institutions. This gives an $\widetilde O\bigl(|\Insts|^2\bigr)$-bit upper bound for $\TTC$. Our contribution in \autoref{sec:compression-complexity} is a nearly matching lower bound: when $|\Appls|$ is sufficiently large relative to $|\Insts|$ (namely, whenever $|\Appls| > |\Insts|^2$), representing the matching \emph{requires} $\Omega\bigl(|\Insts|^2\bigr)$ bits. Note, however, that if $|\Appls| < |\Insts|^2$, then it is always more communication-efficient to simply write down the match of each applicant separately; this gets our final tight bound of $\widetilde\Theta\bigl(\min\{|\Appls|,|\Insts|^2\}\bigr)$ on the complexity of representing the matching in $\TTC$. 
To prove this result, we use another carefully crafted lower-bound construction quite akin to our constructions for our \nameref{ppg:type-to-match} and \nameref{ppg:type-to-menu} Questions.
Qualitatively, our results in \autoref{sec:compression-complexity} are:

\begin{theorem}[Informal; Combination of \autoref{thrm:representing-DA} and \autoref{thrm:representing-TTC}]
  Consider a market with many more applicants than institutions, where all applicants initially know their priorities.  For $\DA$, representing to each applicant her own match with a blackboard protocol requires only a few bits per institution. For $\TTC$, the same complexity requires at least one bit for each \textbf{pair} of institutions.
\end{theorem}

This result gives perhaps the most application-oriented distinction between $\DA$ and $\TTC$ in our paper. With a single complexity measure, it gives a precise sense in which priorities relate to the outcome matching in a more complex manner in $\TTC$ than in $\DA$, corroborating and clarifying past intuitions from both practitioners \cite{BPS05} and theorists \cite{LeshnoL21}.

\paragraph{Verification}\label{ppg:verification}\textbf{\hspace{-0.4em}Question\ \ (\questionTextVerification{} \big/\ \autoref{sec:verif-complexity})\ } 
Our final main question builds upon \autoref{sec:compression-complexity}, and asks: how many bits are required to additionally \emph{verify} that the matching was calculated correctly? Here, we mean verification in the traditional theoretical computer science sense of nondeterministic computation / communication: no incorrect matching can be described without some applicant being able to detect the fact that the matching is incorrect.
To begin, we show in \autoref{thrm:verif-lb-both-mechs} that for both mechanisms, a verification protocol requires $\Omega\bigl(|\Appls|\bigr)$ bits, so (unlike in the \nameref{ppg:representation} Question) we cannot hope to verify the matching using a blackboard certificate containing less than one bit per applicant, showing that there is a real cost to verification.

Despite the fact that $\TTC$ is harder to represent than $\DA$, we prove that these two mechanisms are equally difficult to verify, each requiring $\widetilde O\bigl(|\Appls|\bigr)$ bits. Possibly surprisingly, this upper bound turns out to be perversely harder to show for the easier-to-represent mechanism $\DA$. 
Indeed, while $\TTC$ has a simple \emph{deterministic} $\widetilde O\bigl(|\Appls|\bigr)$-bit communication protocol, we achieve the same bound with a \emph{nondeterministic} protocol for $\DA$ that delicately exploits classical properties of the extremal elements of the set of stable matchings \cite{GusfieldStableStructureAlgs89}. We prove:

\begin{theorem}[Informal; Combination of \autoref{thrm:verif-complexity-ttc} and \autoref{thrm:verif-complexity-da}]
  For both $\TTC$ and $\DA$, the outcome matching can be verified using a blackboard certificate containing a few bits per applicant, and this is tight.
\end{theorem}

Our protocol for $\DA$ crucially uses the fact that the priorities are prior knowledge.
Indeed, \cite{GonczarowskiNOR19} show that (with $n$ applicants and $n$ institutions) if the priorities must be communicated, then any protocol verifying $\DA$ requires $\Omega(n^2)$ bits; in contrast, our protocol uses $\widetilde O(n)$ bits in the known-priorities model.
This gives a strong separation for $\DA$ between these two natural models.
It also shows that the positive result for $\DA$ under the \nameref{ppg:representation} Question (where $\DA$ required $\widetilde O\bigl(|\Insts|\bigr)$ bits, but $\TTC$ required $\Omega\bigl(|\Insts|^2\bigr)$ bits) is somewhat limited. 
Namely, it shows this approach towards representation fundamentally does not suffice to perform verification (under which $\DA$ and $\TTC$ become equally complex, requiring $\Omega(|\Appls|)$ bits).

\paragraph{Additional Questions.}
To explore additional applications of our high-level framework, in \autoref{sec:additional-results} we address several supplementary complexity measures that arise as follow-ups to our main questions. 
Perhaps most interestingly, we observe that while the protocols for representing (the outcomes of) $\TTC$ and $\DA$ in \autoref{sec:compression-complexity} communicate each applicant $d$'s match in terms of her top pick from some set, this set is \emph{not} $d$'s menu.
Thus, we consider a harder version of the \nameref{ppg:representation} Question, namely, the complexity of simultaneously representing all applicant's menus (\autoref{sec:bit-complexity-menus}), and show that it is $\Omega(n^2)$ (as high as possible), even for the simple-to-represent $\DA$.
We also study an easier version of the \nameref{ppg:type-to-match} Question: we ask how one applicant can affect a single other applicant's match (\autoref{sec:type-to-anothers-match}).
We find both $\TTC$ and $\DA$ have low complexity according to this measure.
Finally, we study a topic in the intersection of our \nameref{ppg:type-to-match} and \nameref{ppg:representation} Questions.
We consider the complexity of representing, for each applicant simultaneously, the effect that unilaterally changing her type can have on a single (fixed) applicant's match (\autoref{sec:all-type-to-one-match}).
We show that this complexity is high, even for $\SD$. To sum up, none of these supplementary complexity measures separates $\TTC$ and $\DA$.

\paragraph{Discussion.}
Holistically, our results formalize ways in which $\TTC$ is more complex than $\DA$: one agent can have a more complex effect on another's menu, and the relationship between the priorities and the outcome matching are harder to describe.
However, we also find that such distinctions are somewhat limited: in both mechanisms, one applicant can have an equally complex effect on the matching overall; the two mechanisms are equally complex to verify; and it is equally hard to describe all applicants' menus together. 
Our results could be seen as complementing and contrasting recent results under different models \cite{GonczarowskiHT22, GonczarowskiNOR19}, which point out ways in which $\DA$ is more complex than $\TTC$, emphasizing that many different concerns factor into the complexity of matching mechanisms.
Nevertheless, we show that the complexity concerns of \cite[etc]{BPS05} can be made precise, and that $\DA$ has a formal advantage in simplicity in multiple regards.
See \autoref{tab:all-main-results} for a summary of our results comparing $\TTC$ and $\DA$. 

\newcommand{\ResultEntry}[2]{\vspace{0in} \makecell{ \makecell{#1} \\ {\footnotesize \makecell{#2} }}}
\newcommand{\ResultDivider}{\ $\Big|$\ \:\!}

\begin{table}[thb]
  \caption[Results]{Summary of all our results comparing $\TTC$ and $\DA$.}
  \vspace{-0.3in}
  \label{tab:all-main-results}
  \begin{center}
    \begin{tabular}{C{15.3em}cC{12.3em}C{10.8em}}
  \toprule
    && $\TTC$ & $\DA$ \\
    \cmidrule{3-4}
    \\[-1.5em]
    \vspace{-0.3em}
    \makecell{Describing one's effect\\ on the full outcome matching\\[-0.2em] {\footnotesize
    (\nameref{ppg:type-to-match} Question / \autoref{sec:gtc})}}
    && \ResultEntry{ $\widetilde\Theta\bigl(n^2\bigr)$ 
      }{ By \autoref{thrm:gtc-TTC-LB}. }
    & \ResultEntry{ $\widetilde\Theta\bigl(n^2\bigr)$ 
    }{ For $\APDA$, by \cite{GonczarowskiHT22}. \\
      For $\IPDA$, by \autoref{thrm:gtc-IPDA}. }
    \\ \\
    \vspace{-0.3em}
    \makecell{Describing one's effect\\ on another's menu (set of options)\\[-0.2em] {\footnotesize (\nameref{ppg:type-to-menu} Question / \autoref{sec:type-to-menu})}}
    && \ResultEntry{ $\widetilde\Theta\bigl(n^2\bigr)$
      }{ By \autoref{thrm:type-to-menu-TTC}. }
    & 
    \ResultEntry{${\widetilde\Theta\bigl(n\bigr)}$}{By \autoref{thrm:type-to-menu-DA}.}
    \\  \\ 
    \vspace{-0.3em}
    \makecell{Concurrently representing to\\ every applicant her own match
    \\[-0.2em] {\hspace{-0.5em}\footnotesize (\nameref{ppg:representation} Question / \autoref{sec:compression-complexity})} }
    &&  \ResultEntry{ $\widetilde\Theta\bigl(\min\{|\Appls|, |\Insts|^2\}\bigr)${\ResultDivider$\widetilde\Theta(n)$ }
      }{ By \cite{LeshnoL21} and \autoref{thrm:representing-TTC}.  }
    &  \ResultEntry{ $\widetilde\Theta\bigl(|\Insts|\bigr)$\ResultDivider$\widetilde\Theta(n)$
      }{For any stable matching,\\by \cite{AzevedoL16} / \autoref{thrm:representing-DA}. }
    \\ \\
    \vspace{-0.3em}
    \makecell{\nameref{ppg:representation} as above, as well \\as jointly verifying the matching\\[-0.2em] {\footnotesize (\nameref{ppg:verification} Question / \autoref{sec:verif-complexity})} }
    && \ResultEntry{ {$\widetilde\Theta\bigl(|\Appls|\bigr)$}\ResultDivider{$\widetilde\Theta(n)$}
    }{ By deterministic complexity, \\ see \autoref{thrm:verif-complexity-ttc}. }
    &  \ResultEntry{ {$\widetilde\Theta\bigl(|\Appls|\bigr)$}\ResultDivider{$\widetilde\Theta(n)$}
    }{ For both $\APDA$ and $\IPDA$, \\ by \autoref{thrm:verif-complexity-da}. }
    \\[-0.5em]
    \vspace{-0.3em}
    \\[-0.5em] \bottomrule
  \end{tabular}
  \end{center}

  {\footnotesize \textbf{Notes:}
  Each result bounds the number of bits required to represent the function from one (fixed) applicant's report onto some piece of data (\nameref{ppg:type-to-match} and \nameref{ppg:type-to-menu} Questions), or the number of bits required to simultaneously perform a task with all applicants simultaneously via a blackboard protocol (\nameref{ppg:representation} and \nameref{ppg:verification} Questions).
  We consider markets matching $n$ applicants and $n$ institutions; 
  for our \nameref{ppg:representation} and \nameref{ppg:verification} Questions, we also consider markets with general sets of applicants $\Appls$ and institutions $\Insts$ with $|\Appls|\gg|\Insts|$ and $|\Appls|=\poly(|\Insts|)$; see \autoref{sec:protocol-models} for a discussion of the $|\Appls|\gg|\Insts|$ model, and see \autoref{remark:unbalanced-bounds} for a discussion of the Effect Questions when $|\Appls|\ne|\Insts|$.
  \par } 
\end{table}

\paragraph{Paper Organization.}
After discussing related work and providing preliminaries, we begin our paper with our more purely theoretical results.
These are Sections~\ref{sec:gtc} and~\ref{sec:type-to-menu}, discussing outcome-effect and options-effect complexity.
Next, our more policy-relevant results on representation complexity are in \autoref{sec:compression-complexity}, and our results on verification complexity are in \autoref{sec:verif-complexity}.
\autoref{sec:additional-results} gives additional complexity bounds that extend our main results; see \autoref{tab:appendix-results} on Page~\pageref{tab:appendix-results} for a summary. \autoref{sec:additional-prelims} gives additional preliminaries. \autoref{sec:related-simplicity} gives a technical comparison between our results and certain related prior works.

\subsection{Related work}
\label{sec:related}

The traditional economic approach to studying the structure of matching mechanisms often considers qualitative, binary properties that are satisfied by some mechanisms but not by others.
These properties include strategyproofness, nonbossiness, and population monotonicity.
Strategyproofness says that applicant~$d$'s report $P_d$ can only possibly affect $d$'s own match in a very controlled way: by matching $d$ to her highest-ranked institution on her menu.
Similarly, nonbossiness and population montonicity restrict when and how one applicant can affect other applicants' matches.
Nonbossiness says that an applicant can only change the matching if her own match changes.
Population monotonicity says that adding a new applicant to the market can only make other applicants worse off. Both $\TTC$ and $\APDA$ (the applicant-proposing variant of $\DA$) are strategyproof; $\TTC$ and $\IPDA$ (the institution-proposing variant) are nonbossy; $\APDA$ and $\IPDA$ are population monotonic.

Each of these properties has an interpretation similar to our \nameref{ppg:type-to-match} and \nameref{ppg:type-to-menu} Questions: they ask ``how can one applicant's report affect the mechanism?''.
Our complexity results complement these classical, qualitative binary properties in a quantitative way.
For instance, the fact that the outcome-effect complexity of $\TTC$ (a nonbossy mechanism) is high (\autoref{thrm:gtc-TTC-LB}) shows that an applicant's report can affect the outcome matching in a complex way (despite the fact that an applicant's report can only affect the matching if it affects her own match).
Additionally, our characterization showing that the options-effect complexity of $\DA$ is low (\autoref{thrm:type-to-menu-DA}) makes $\DA$'s population monotonicity more precise: it characterizes exactly which institutions will be removed from some applicant's menu when a new applicant is added.

Many of our results in Sections~\ref{sec:gtc} and~\ref{sec:type-to-menu} are most directly inspired by the recent mechanism design paper \cite{GonczarowskiHT22}. Briefly and informally, \cite{GonczarowskiHT22} are interested in describing mechanisms to participants in terms of their menus, as a way to better convey strategyproofness. 
In this direction, our \autoref{sec:gtc} provides a new lens into the relationship between one applicant's menu and the full matching (see \autoref{sec:gtc-apda} for details), and our \autoref{sec:type-to-menu} provides a new lens into the relationship between different applicants' menus (see \autoref{sec:application-SEDA-1-to-1} for details).

Our results in Sections~\ref{sec:gtc} and~\ref{sec:type-to-menu} are also inspired by the broader literature within algorithmic mechanism studying menus of selling mechanisms.
Traditionally, this literature studies single-buyer mechanisms \cite{HartN19,DaskalakisDT17,BabaioffGN22,SaxenaSW18,Gonczarowski18}. More recently, the scope of mechanisms under consideration has expanded to multi-buyer ones \cite{Dobzinski16b, DobzinskiR21}. Compared to these works, the study of the complexity of menus for matching mechanisms (initiated by \cite{GonczarowskiHT22}) requires asking different questions. To illustrate why, we recall that \cite{Dobzinski16b} formulates the \emph{taxation complexity} of strategyproof mechanisms, which measures the number of distinct menus a bidder can have. 
The taxation complexity is trivial to bound for strategyproof matching rules, since the set of possible menus in these mechanisms is (typically) simply the family of all subsets of institutions, requiring $\Theta(n)$ bits to represent. Our outcome-effect complexity (\autoref{sec:gtc}), as well as certain complexity measures in \autoref{sec:additional-results},
can be seen as significant generalizations of taxation complexity.
Investigating our questions for auction-like mechanisms remains an intriguing future direction.

Our \autoref{sec:compression-complexity} is inspired by the literature on the cutoff structures of $\DA$ \cite{AzevedoL16} and $\TTC$ \cite{LeshnoL21}. To our knowledge, we are the first paper to formalize lower bounds on the complexity of this structure, and prove that the upper bounds of \cite{AzevedoL16,LeshnoL21} are tight.\footnote{
  \cite{LeshnoL21} give a brief argument as to how $\TTC$ is formally more complex than $\DA$, but (unlike our paper) they do not prove any type of $\Omega(n^2)$ lower bounds. 
  We provide a full discussion and comparison to their observation in \autoref{sec:related-cutoff-TTC}.
} However, in \autoref{sec:verif-complexity} we also show that verification of these mechanisms is more delicate than may have been previously believed.
We give an extensive discussion of how our definitions and results compare to previous ones in \autoref{sec:related-simplicity}.

There has also been prior work on verification of $\DA$ in more traditional models where the priorities of the institutions also need to communicated. For example, \cite{Segal07, GonczarowskiNOR19} prove lower bounds of $\Omega(n^2)$ for computing or even verifying a stable matching in this model. Since our protocol in \autoref{sec:verif-complexity} for $\DA$ constructs an $\widetilde O(n)$ upper bound for verification, our models are provably quite different.

In various contexts, different works have studied the prospect of changing one applicant's preferences in stable matching markets. \cite{MaiV18, GangamMRV18} study matchings that are stable both before and after one applicant changes her list; our upper bound characterization in \autoref{sec:type-to-menu-DA} may have some conceptual relation to these works (though little technical resemblance). \cite{Kupfer20} studies the effect of one applicant changing her reported preference to her strategically optimal one under~$\IPDA$.

Much more work has gone into understanding matching mechanisms under more traditional computer-science questions (such as computational complexity) or economic questions (such as incentives).
From a computer science perspective, \cite{IrvingCountingStable86, SabanS15} show that several problems related to stable matching are $\#\mathsf{P}$-hard, and \cite{KarlinGW18, PalmerP21} give upper bounds on the maximum number of stable matchings corresponding to a given set of preferences and priorities.
\cite{Subramanian94, CookFL14} show that stable matchings are connected to comparator circuits and certain novel complexity classes between $\mathsf{NL}$ and $\mathsf{P}$.
\cite{PittelAverageStable89, ImmorlicaM05, AshlagiKL17, KanoriaMQ21, CaiT22} study the running time of $\DA$ under random preferences.
\cite{BadeG16,AshlagiG18, Troyan19, Thomas21, mandal2020obviously,PyciaT21,GolowichL21} study the strategic simplicity of matching mechanisms through the lens of obvious strategyproofness \cite{Li17}.
Our work also relates to a recent push in the algorithmic mechanism design literature to understand incentive-constrained computation in more detailed and refined ways. 
This push includes papers studying restricted solutions concepts like dominant-strategy implementations \cite{RubinsteinSTWZ21, DobzinskiRV22} or the power of different types of simple mechanisms \cite[and many others]{BulowK96,HartlineR09,HartN17,EdenFFTW17b,EdenFFTW17a,AkbarpourL20,BabaioffILW20,BabaioffGG20}.

\section{Preliminaries}
\label{sec:prelims}

This paper studies rules for matching a set of \emph{applicants} $\Appls$ and a set of \emph{institutions} $\Insts$. A matching rule is a function $f : \T_1\times\ldots\times\T_{|\Appls|} \to \M$, where each set $\T_d$ (mnemonic: the possible \emph{types} of applicant $d$) is the set of all rank-order preference lists over $\Insts$, and $\M$ is the set of matchings $\mu : \Appls \to \Insts \cup \{\emptyset\}$ (where we write $\mu(d)=\emptyset$ if $d\in\Appls$ is unmatched). Preference lists may be partial: when an applicant does not rank an institution, this indicates that the applicant finds the institution unacceptable. We typically write $P_d \in \T_d$, and let $\succ_d^{P_d}$ denote the relation over $\Insts$ such that $h \succ_d^{P_d} h'$ if and only if $h$ is ranked above $h'$ according to $P_d$.
We also write $h \succ_d h'$ where no confusion can arise.
For each $d\in\Appls$, we let $f_d : \T_1\times\ldots\times\T_n \to \Insts$ denote the function giving the match of applicant $d$ according to $f$.%
\footnote{
We also follow standard notations such as writing $h \succeq_d^{P_d} h'$ when $h \succ_d^{P_d} h'$ or $h = h'$, writing $P\in\T$ to denote $(P_1,\ldots,P_n)\in\T_1\times\ldots\times\T_n$, writing $P_{-i} \in \T_{-i} = \T_1\times\ldots\times\T_{i-1}\times\T_{i+1}\times\ldots\times\T_n$ to denote $(P_1,\ldots,P_{i-1},P_{i+1},\ldots,P_n)$,
and writing $(P'_i, P_{-i})$ to denote $(P_1,\ldots,P_{i-1},P_i',P_{i+1},\ldots,P_n)$. 
For a set of applicants $S\subseteq\Appls$, we write $P_{-S}$ for a profile of preferences of applicants not in $S$.
For a matching $\mu$ and institution $h$, we abuse notation and let $\mu(h)\in\Appls\cup\{\emptyset\}$ denote the applicant $d$ (if there is any) with $\mu(d)=h$.
To avoid common variables $a$ and $i$, we typically use $d$ (mnemonic: doctor) to denote an element of $\Appls$ and $h$ (mnemonic: hospital) to denote an element of~$\Insts$.
}
The ``preferences'' of the institutions are known priorities, which we denote by $Q = (Q_h)_{h\in\Insts}$. Similarly to the applicants, we write $d \succ_h^{Q_h} d'$ or simply $d \succ_h d'$ when $d$ has higher priority at $h$ according to $Q$.

We study three canonical priority-based matching mechanisms. 
The first is $\TTC$.

\begin{definition}
  \label{def:TTC}
  The Top Trading Cycles ($\TTC = \TTC_Q(\cdot)$) mechanism is defined with respect to a profile of priority orders $Q = \{Q_h\}_h$.
  The matching is produced by repeating the following until every applicant is matched (or has exhausted her preference list): each remaining (i.e., not-yet-matched) applicant points to her favorite remaining institution, and each remaining institution points to its highest-priority remaining applicant. There must be some cycle in this directed graph (as the graph is finite). Pick any such cycle and permanently match each applicant in this cycle to the institution to which she is pointing. These applicants and institutions do not participate in future iterations.
\end{definition}

$\TTC$ produces a Pareto-optimal matching under the applicants' preferences, i.e. a matching $\mu$ such that no $\mu'\ne\mu$ exists such that $\mu'(d)\succeq_d \mu(d)$ for each $d\in\Appls$.
Despite the fact that the $\TTC$ procedure does not specify the order in which cycles are matched, the matching produced by $\TTC$ is unique (\autoref{thrm:TTC-order-independent}).%
\footnote{
$\TTC$ is also commonly studied under a model where each applicant starts out ``owning'' a single institution (a so-called housing market, which is equivalent to each institution having a distinct top-priority applicant in our model). Our lower bounds hold for this alternative model as well (outside of the questions of \autoref{sec:compression-complexity}, which require having $|\Appls|\gg|\Insts|$).
}
In certain common market structures, $\TTC$ is the unique Pareto-optimal mechanism that is strategyproof (see below) and satisfied individual rationality (i.e., no applicant is matched below any institution where she has top priority) \cite{Ma94,Sethuraman16}.
Thus, $\TTC$ is the nearly-unique strategyproof (and individually rational) matching mechanism that produces applicant-optimal matchings.

The second and third mechanisms that we study are the two canonical variants of $\DA$, namely, $\APDA$ and $\IPDA$:

\begin{definition}
  \label{def:DA}
  Applicant-Proposing Deferred Acceptance ($\APDA = \APDA_Q(\cdot)$) is defined with respect to a profile of priority orders $Q = \{Q_h\}_h$, one for each institution $h$, over applicants. 
  The matching is produced by repeating the following until every applicant is matched (or has exhausted her preference list): A currently unmatched applicant is chosen to \emph{propose} to her favorite institution that has not yet \emph{rejected} her. The institution then rejects every proposal except for the \emph{top-priority applicant} who has proposed to it thus far.
  Rejected applicants become (currently) unmatched, while that top-priority applicant is tentatively matched to the institution. This process continues until no more proposals can be made, at which time the tentative allocations become final.

  The mechanism Institution-Proposing Deferred Acceptance $\IPDA = \IPDA_Q(\cdot)$ is defined in one-to-one markets identically to $\APDA$, except interchanging the roles of the applicants and the institutions. In other words, the matching of $\IPDA_Q(P)$ coincides with $\APDA_P(Q)$, treating the preferences $P$ as priorities and priorities $Q$ as preferences.
\end{definition}

We use $\DA$ to denote either $\APDA$ or $\IPDA$ when the distinction is not important (i.e., when the result holds for both mechanisms by the same proof).
Both $\APDA$ and $\IPDA$ produce stable matchings. A matching $\mu$ is stable if there is no unmatched pair $d,h$ such that $d \succ_h \mu(h)$ and $h \succ_d \mu(d)$; this property is often taken as the canonical notion of fairness in matching mechanisms~\cite{AbdulkadirougluS03}.
Despite the fact that these procedures do not specify the order in which proposals are made, the matching produced by each of them is unique (\autoref{thrm:da-indep-execution}). In fact, $\APDA$ is the unique stable mechanism that is strategyproof for the applicants \cite{GaleS85,Roth82-DA}.
Thus, $\APDA$ is the unique strategyproof (for applicants) matching mechanism that produces stable matchings (and naturally, the same holds for institutions in $\IPDA$).

We also consider the very simple mechanism of \emph{serial dictatorship}, which corresponds to both $\TTC$ and to $\DA$ in the case where all institutions' priority lists are $\succ$. 
\begin{definition}
  Serial Dictatorship ($\SD = \SD_\succ(\cdot)$) is defined with respect to a single priority ordering $\succ$ over applicants, say with $d_1\succ\ldots\succ d_n$. 
  The matching is produced by matching applicants, in order $d_1,d_2,\ldots,d_n$, to their favorite remaining institution according to their submitted list $P_{d_1},\ldots,P_{d_n}$.
  We also denote this as $\SD_{d_1,\ldots,d_n}(\cdot)$.
\end{definition}

\paragraph{Strategyproofness and menus}
The matching rules of $\TTC$ and $\APDA$ (though not $\IPDA$) are \emph{strategyproof} for the applicants, i.e., for all $d\in\Appls$, all $P_d, P_d' \in \T_d$ and all $P_{-d}\in\T_{-d}$, we have $f_d(P_d, P_{-d}) \succeq_d^{P_d} f_d(P_d', P_{-d})$.
This property is tightly connected to the classical notion of the \emph{menu} of a player in a mechanisms, which is a common notion throughout our paper. The menu is the natural definition of the set of obtainable options that an applicant has in the mechanism:
\begin{definition}
  \label{def:menu}
  For any matching rule $f$ and applicant $d$, the \emph{menu $\Menu_d^f(P_{-d})$ of $d$ given $P_{-d}\in\T_{-d}$} is the subset of all institutions $h \in \Insts$ such that there exists some $P_d \in \T_d$ such that $f_d(P_d, P_{-d}) = h$.  That is,
  \[ \Menu_d(P_{-d}) = \Menu^f_d(P_{-d}) = 
    \left\{\ f_d(P_d, P_{-d}) \ \big|\ P_d \in \T_d\ \right\} 
    \subseteq \Insts. \]
\end{definition}

The menu is a lens through which one can view or understand strategyproofness, as captured by the following equivalence:
\begin{theorem}[\cite{Hammond79}]
  \label{thrm:TaxationPrinciple}
  A matching rule $f$ is strategyproof if and only if each applicant $d$ is always matched to her favorite institution from her menu (that is, for each $P_{-d}\in\T_{-d}$ and $P_d \in \T_d$, we have $f_d(P_d, P_{-d}) \succeq_d^{P_d} h$ for any $h \in \Menu_d^f(P_{-d})$).
\end{theorem}

\section{Outcome-Effect Complexity}
\label{sec:gtc}

Recall our \nameref{ppg:type-to-match} Question: \questionTextTypeToMatching{} As discussed in \autoref{sec:intro}, one canonical way to measure the complexity of how one applicant can affect the matching is through the number of bits it takes to represent the function $f(\cdot, P_{-d})$. Formally:

\begin{definition}
  \label{def:type-to-matching}
  The \emph{outcome-effect complexity} of a matching mechanism $f$ is
  \[ \log_2\ \max_{d\in\Appls} \left|\left\{ 
      f(\cdot, P_{-d})\ \big|\ P_{-d}\in\T_{-d}
    \right\}\right|, \]
  where $f(\cdot, P_{-d}) : \T_d \to \M$ is the function mapping each $P_d\in\T_d$ to the matching $f(P_d, P_{-d})$.
\end{definition}

This provides a generic formal measure into the complexity of how one applicant can affect the mechanism, which we explore throughout this section.

\subsection{Warmup: \texorpdfstring{$\SD$}{SD}}

To begin the discussion of outcome-effect complexity, we consider $\SD$. 
Despite the extreme simplicity of this mechanism, and despite the fact that it is strategyproof and nonbossy, it is not immediately clear whether there is any way to represent the function $\SD(\cdot, P_{-1}) : \T_1 \to \M$ efficiently. One could write down a separate matching $\bigl(\SD(\{h_1\}, P_{-1}),\ldots,\SD(\{h_n\}, P_{-1})\bigr)$ for each possible institution that applicant $1$ might pick, but this representation takes $\widetilde\Omega(n^2)$ bits (matching the $\widetilde\Omega(n^2)$ solution that simply records the entirety of $P_{-1}$). 
Nevertheless, we now show that this complexity is $\widetilde O(n)$ for $\SD$, using a novel data structure representing all possible matchings as a function of $P_1$. 

\begin{proposition}[restate=restateThrmGtcSdUb, name=] 
  \label{thrm:gtc-SD-UB}
  The outcome-effect complexity of $\SD$ is $\widetilde\Theta(n)$.
\end{proposition}
\begin{proof}
  Suppose $\Appls = \{1,2,\ldots,n\}$ and $\succ$ ranks applicants in order $1 \succ 2 \succ \ldots n$. Our goal is to bound $\max_{d_*} \log_2 \bigl|\bigl\{ \SD(\cdot, P_{-{d_*}}) ~\big|~ P_{-d_*}\in\T_{-d_*} \bigr\}\bigr|$; observe that it is without loss of generality to take $d_* = 1$.
  Observe also that an $\widetilde\Omega(n)$ lower bound follows from the need to write down the matching that results if applicant $1$ selects no items.

  Now, fix preferences $P_{-1}$ of applicants other than $d_*=1$.
  Our proof will proceed by finding another profile of preferences $P_{\mathsf{small}}$, such that (1) each applicant's preference list in $P_{\mathsf{small}}$ has length at most two, and (2) for all $P_1 \in \T_1$, we have $\SD(P_1,P_{\mathsf{small}}) = \SD(P_1,P_{-1})$.

  To begin, we define a set of ``filtered'' preference profiles $P_{\mathsf{filt}}^1,\ldots,P_{\mathsf{filt}}^n$. Each $P_{\mathsf{filt}}^i$ contains a preference list for all applicants except applicant $1$. Define $P_{\mathsf{filt}}^1 = P_{-1} \in \T_{-1}$. Now, for each $i=2,3,\ldots,n-1$ in order, define $P_{\mathsf{filt}}^i$ by modifying $P_{\mathsf{filt}}^{i-1}$ as follows:
  if $h \in \Insts$ is the top-ranked institution on $i$'s list in $P_{\mathsf{filt}}^{i-1}$, then remove $h$ from the preference list of each applicant $d$ with $d > i$.
  Define $P_{\mathsf{filt}}$ as $P_{\mathsf{filt}}^{n-1}$.

  First, we show that $P_{\mathsf{filt}}$ always produces the same matching as $P_{-1}$: 
  \begin{lemma}
    \label{thrm:gtc-SD-UB-key-lemma-necessary}
    For any $P_1\in\T_1$, we have $\SD(P_1, P_{-1}) = \SD(P_1, P_{\mathsf{filt}})$.
  \end{lemma}

  To prove this, consider each step of the above recursive process, where some applicant $i$ ranked $h$ first on $P^{i-1}_{\mathsf{filt}}$, and we constructed $P^i_{\mathsf{filt}}$ by removing $h$ from the list of all $d>i$. 
  Now, observe that for any possible $P_1$, when $d$ picks an institution in $\SD(P_1,P_{\mathsf{filt}}^i)$, institution $h$ must already be matched (either to applicant $i$, or possibly an earlier applicant). Thus, removing $h$ from the list of $d$ cannot make a difference under any $P_1$, and for each $P_1 \in \T_1$, we have $\SD(P_1, P_{-1}) = \SD(P_1, P_{\mathsf{filt}}^i)$, by induction. 
  This proves \autoref{thrm:gtc-SD-UB-key-lemma-necessary}.

  Next, we show that only the first two institutions on each preference list in $P_{\mathsf{filt}}$ can matter:
  \begin{lemma}
    \label{thrm:gtc-SD-UB-lemma-complete}
    For any $P_1 \in \T_1$, 
    each applicant will be matched to either her first or second institution in her preference list in $P_{\mathsf{filt}}$.
  \end{lemma}
  To prove this, suppose player $1$'s top choice according to $P_1$ is institution $h_1\in\Insts$. 
  The mechanism will run initially with each applicant $i > 1$ taking her first choice (all of these are distinct in $P_{\mathsf{filt}}$), until some applicant $i_2$ that ranks $h_1$ as her first choice.
  This $i_2$ will take her second choice $h_2$ (since $h_2$ cannot yet be matched, because each prior applicant took her first choice). 
  But, applying this same argument for applicants $d > i_2$, we see that each such applicant will be matched to her first choice until some $i_3 > i_2$ that ranks $h_2$ first. This $i_3$ will get her second choice $h_3$ (which cannot be $h_1$, since $h_1$ was ranked first by $i_2$).
  This will continue again until some $i_4$ whose first choice is $h_3$ and whose second choice is $h_4 \notin \{ h_1, h_2, h_3 \}$; this applicant $i_4$ will match to $h_4$.
  This argument applies recursively until the mechanism is finished, proving \autoref{thrm:gtc-SD-UB-lemma-complete}.

  Now, consider preferences $P_{\mathsf{small}}$, which contain only the first two institutions on the list of each applicant in $P_{\mathsf{filt}}$.
  Then, by both lemmas above, for each $P_1\in\T_1$ we have $\SD(P_1,P_{-1}) = \SD(P_1,P_{\mathsf{small}})$.
  Because each list in $P_{\mathsf{small}}$ is of size at most $2$, it takes $\widetilde O(n)$ bits to write down $P_{\mathsf{small}}$, and thus the outcome-effect complexity of serial dictatorship is $\widetilde O(n)$.
\end{proof}

In particular, the proof of this result gives a concrete and concise representation of the function $\SD(\cdot,P_{-1})$. Namely, there is some ``default matching'' $\mu_0$ of applicants $i > 1$, which corresponds to the matching when applicant $d_1$ submits an empty preference list and matches to no institutions, and there is a DAG of ``displaced matches'' that might occur based on the choice of $d_1$. 
To construct this DAG, create a vertex $(d_i,\mu_0(d_i))$ for each applicant $d_i$ (and a vertex $(\emptyset, h)$ for each $h$ unmatched in $\mu_0$), and for each preference list in $P_{\mathsf{small}}$, say with $d_i : h_j \succ h_k$, create an edge from $(d_i,h_j)$ to the vertex $(d',h_k)$, where $d'$ satisfies $\mu_0(d')=h_k$.
The result of $\SD(P_1, P_{-1})$ can be calculated by matching $d_1$ to their top choice $h_*$, then considering the unique maximal path $(d_2',h_*), (d_3',h_3') \ldots, (d_K',h_K')$ in the DAG, and matching each $d_i'$ with $h_{i+1}'$ for $i\in \{2,\ldots,K-1\}$.
See \autoref{fig:type-to-outcome-sd} for an example and illustration.
\begin{figure}[tbp]
  \begin{minipage}[c]{0.3\textwidth}
    \centering
    \includegraphics[width=0.9\textwidth]{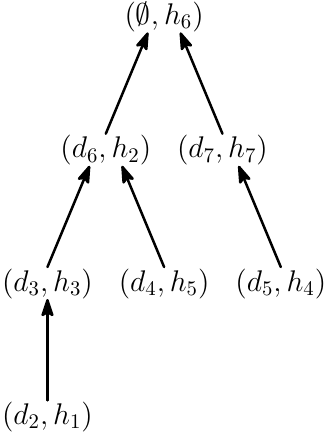}
  \end{minipage}
  \qquad
  \begin{minipage}[c]{0.65\textwidth}
  \caption[Type-to-menu DA Example]{
  Illustration of $\SD(\cdot, P_{-1})$, as given by \autoref{thrm:gtc-SD-UB}, on the preferences profile below.
  }
  \label{fig:type-to-outcome-sd}
    \begin{minipage}[c]{0.35\textwidth}
      Preferences:
      \begin{align*}
           & d_1 : \text{(Unknown)}
        \\ & d_2 : h_1 \succ h_3 \succ h_4 \succ h_5
        \\ & d_3 : h_3 \succ h_2 \succ h_1
        \\ & d_4 : h_3 \succ h_5 \succ h_1 \succ h_2
        \\ & d_5 : h_4 \succ h_3 \succ h_7
        \\ & d_6 : h_2 \succ h_4 \succ h_5 \succ h_6
        \\ & d_7 : h_2 \succ h_4 \succ h_7 \succ h_6
      \end{align*}
    \end{minipage}
    \qquad
    \begin{minipage}[c]{0.28\textwidth}
      $P_{\mathsf{small}}$:
      \begin{align*}
          & d_1 : \text{(Unknown)}
        \\ & d_2 : h_1 \succ h_3
        \\ & d_3 : h_3 \succ h_2
        \\ & d_4 : h_5 \succ h_2
        \\ & d_5 : h_4 \succ h_7
        \\ & d_6 : h_2 \succ h_6
        \\ & d_7 : h_7 \succ h_6
      \end{align*}
    \end{minipage}
    \qquad
    \begin{minipage}[c]{0.15\textwidth}
      $\mu_{0}$:
      \begin{align*}
           & d_1 : \emptyset
        \\ & d_2 : h_1
        \\ & d_3 : h_3
        \\ & d_4 : h_5
        \\ & d_5 : h_4
        \\ & d_6 : h_2
        \\ & d_7 : h_7
      \end{align*}
    \end{minipage}
  \end{minipage}
\end{figure}

For our main mechanisms of interest, the outcome-effect complexity will turn out to be high, as we see next.

\subsection{\texorpdfstring{$\TTC$}{TTC}}
\label{sec:gtc-ttc}

We now bound the outcome-effect complexity of $\TTC$. 

\paragraph{Building up to our construction.}
To begin, we note that it is not initially clear which ``$\TTC$-pointing graphs'' (i.e., which data structures used by the $\TTC$ algorithm in \autoref{def:TTC}) could possibly suffice to lower bound the outcome-effect complexity. 
Fix some applicant $d_*$, and consider delaying matching the cycle containing $d_*$ as long as possible.\footnote{This is inspired by the first step of a description of $\TTC$ given in \cite{GonczarowskiHT22}.}
One can show that if this initial $\TTC$-pointing graph consists of \emph{only} a long path, then the outcome-effect complexity is low (since for each applicant, only their favorite choice among institutions further up the path can matter for the remainder of the run of $\TTC$). Moreover, if this initial graph consists only of isolated vertices, then matching the cycle that $d_*$ completes can only affect the pointing graph in a small way. Since these two extremes fail to provide a high-complexity construction, it is initially unclear how high this complexity might be.

\paragraph{Intermediate mechanism: $\SDrot$.}
Informally, we are able to show this complexity measure is high because an applicant in $\TTC$ can dramatically affect the \emph{order} in which future cycles will be matched in $\TTC$, and thus dramatically affect the entirety of the matching.
To present our formal proof, we first define an intermediate mechanism that we call $\SDrot$, a variant of $\SD$ where the first applicant can affect the order in which other applicants choose institutions. We show that, unlike $\SD$, the mechanism $\SDrot$ has high outcome-effect complexity. Then, we show that $\TTC$ can ``simulate'' $\SDrot$, and thus $\TTC$ has outcome-effect complexity at least as high as $\SDrot$.

\begin{definition}
  \label{def:SDrot}
  Consider a matching market with $n+1$ applicants $\{d_*,d_1,\ldots,d_n\}$ and $2n$ institutions $\{h_1,\ldots,h_n,h_1^{\mathsf{rot}},\ldots,h_n^{\mathsf{rot}}\}$.
  Define a mechanism $\SDrot$ as follows: first, $d_*$ is permanently matched to her top-ranked institution $h_j^{\mathsf{rot}}$ from $\{h_1^{\mathsf{rot}},\ldots,h_n^{\mathsf{rot}}\}$. 
  Then, in order, each of the applicants { $d_j, d_{j+1}, \allowbreak\ldots,\allowbreak d_{n-1}, d_n$} is permanently matched to her top-ranked remaining institution from $\{h_1,\ldots,h_n\}$ (and all other applicants go unmatched). In other words, applicants are allocated to $\{h_1,\ldots,h_n\}$ according to $\SD_{d_j,d_{j+1},\ldots,d_n}(\cdot)$.
\end{definition}

Informally, this mechanism has high outcome-effect complexity because (under different reports of applicant $d_*$) each applicant $d_i$ with $i>0$ may be matched \emph{anywhere} in the ordering of the remaining agents, and thus \emph{any} part of $d_i$'s preference list might matter for determining $d_i$'s matching. 

\begin{lemma}
  \label{thrm:gtc-SDrot-LB}
  The outcome-effect complexity of $\SDrot$ is $\Omega(n^2)$.
\end{lemma}
\begin{proof}
  Fix $k$, where we will take $n = \Theta(k)$. For notational convenience, we relabel the applicants $d_1,d_2,\ldots,d_n$ as $d^L_1,d^R_1,d^L_2,d^R_2,\ldots,d^L_k,d^R_k$ in order, and relabel the institutions $h_1,h_2,\ldots,h_n$ as $h^0_1,h^1_1,h^0_2,h^1_2,\ldots,h^0_k,h^1_k$, and relabel the institutions $h_1^{\mathsf{rot}},h_2^{\mathsf{rot}},\ldots,h_n^{\mathsf{rot}}$ as $h^{0,\mathsf{rot}}_1,h^{1,\mathsf{rot}}_1,h^{0,\mathsf{rot}}_2,h^{1,\mathsf{rot}}_2,\ldots,\allowbreak h^{0,\mathsf{rot}}_k,h^{1,\mathsf{rot}}_k$. We define a collection of preference profiles of all applicants other than $d_*$. This collection is defined with respect to a set of $k(k+1)/2 = \Omega(n^2)$ bits $b_{i,j}\in\{0,1\}$, one bit for each $i,j\in[k]$ with $j \le i$. For such a bit vector, consider preferences such that for each $i\in[k]$, we have:
  \begin{align*}
    d^L_i : & \quad h_1^{b_{i,1}} \succ h_1^{1 - b_{i,1}} 
      \succ h_2^{b_{i,2}} \succ h_2^{1 - b_{i,2}} 
      \succ \ldots \succ h_i^{b_{i,i}} \succ h_i^{1 - b_{i,i}}
    \\
    d^R_i: & \quad h_1^{1 - b_{i,1}} \succ h_1^{b_{i,1}} 
      \succ h_2^{1 - b_{i,2}} \succ h_2^{b_{i,2}} 
      \succ \ldots \succ h_i^{1 - b_{i,i}} \succ h_i^{b_{i,i}}
  \end{align*}
  In words, each such list agrees that $h_1^0$ and $h_1^1$ are most preferred, followed by $h_2^0$ and $h_2^1$, etc., and the lists of $d_i^L$ and $d^R_i$ rank all such institutions up to $h_i^0$ and $h_i^1$. But $d^L_i$ and $d^R_i$ may flip their ordering over each $h_j^0$ and $h_j^1$ for $j\le i$, as determined by the bit $b_{i,j}$. The key lemma is the following:
  \begin{lemma}
    \label{thrm:gtc-SDrot-LB-main-lemma}
    Consider any $i,j\in[k]$ with $j \le i$.
    If $d_*$ submits a preference list containing only $\{ h^{0,\mathsf{rot}}_{i-j+1} \}$, then applicant $d^L_{i}$ matches to $h_j^{b_{i,j}}$.
  \end{lemma}
  To prove this lemma, note that $d_*$ will match to $\{ h^{0,\mathsf{rot}}_{i-j+1} \}$, and thus by the definition of $\SDrot$, we know $d^L_{i-j+1}$ will choose her match next (followed by the rest of the applicants, in order).
  Thus, consider the execution of $\SD(d^L_{i-j+1},d^R_{i-j+1},\ldots,d^L_{k},d^R_{k})$.
  Initially, applicants $d^L_{i-j+1},d^R_{i-j+1}$ match to $h_1^0$ and $h_1^1$ (in some order),
  then applicants $d^L_{i-j+2},d^R_{i-j+2}$ match $h_2^0$ and $h_2^1$ (in some order), and so on, until applicants $d^L_i,d^R_i$ pick among $h_j^0$ and $h_j^1$.
  When this happens, $d_i^L$ picks $h_j^{b_{i,j}}$, proving \autoref{thrm:gtc-SDrot-LB-main-lemma}.
  For an illustration, see \autoref{fig:gtc-sdrot-lb}.
\begin{figure}[tbp]
  \begin{center}
  \includegraphics[width=0.7\textwidth]{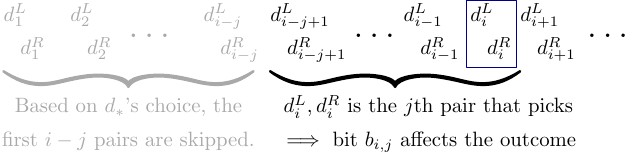}
  \end{center}
  \caption{
  Illustration of the proof of \autoref{thrm:gtc-SDrot-LB}.}
  \label{fig:gtc-sdrot-lb}
\end{figure}

  This shows that for every pair of distinct preference profiles $P$ and $P'$ of the above form, there exists a preference list $P_{d_*}$ of $d_*$ such that $\SDrot(P_{d_*}, P) \ne \SDrot(P_{d_*}, P')$.
  Thus, there are at least $2^{\Omega(k^2)}$ distinct possible functions $\SDrot(\cdot,P_{-d_*})$, and the outcome-effect complexity of $\SDrot$ is at least $\Omega(k^2) = \Omega(n^2)$, as claimed.
\end{proof}

\paragraph{$\TTC$.}
We now proceed to our main mechanisms of interest, beginning with $\TTC$. We show that $\SDrot$ can be ``embedded in'' $\TTC$, and thus show that the outcome-effect complexity of $\TTC$ is high.

\begin{theorem}[restate=restateThrmGtcTtcLb,name=]
  \label{thrm:gtc-TTC-LB}
  The outcome-effect complexity of $\TTC$ is $\Omega(n^2)$.
\end{theorem}
\begin{proof}
  For some $k = \Theta(n)$, we consider applicants $\Appls = \{ d_*, d^R_1,\ldots,d^R_k, d^P_1,\ldots,d^P_k, d^A_1,\ldots,d^A_k \}$. 
  For this construction, we consider markets in which there are exactly $|\Appls|$ institutions, and each institution ranks a distinct, unique applicant with highest priority. In this case, it is easy to see that only the top-priority ranking of each institution can matter for determining the outcome, so this defines the priorities of the institutions. For notational convenience, we identify each institution with the applicant that they rank highest, e.g., we use $d^R_1$ to denote both an applicant and the institution which ranks $d^R_1$ highest. (This is equivalent to considering this construction in a housing market, i.e., where each applicants starts by ``owning'' the institution where they have highest priority.)

  We describe collection of preference profiles which each induce a distinct function $\TTC(\cdot, P_{-d_*}) : \T_{d_*}\to\M$. 
  For applicants outside of $\{d_*, d^P_1,\ldots,d^P_k\}$,
   the preferences are fixed, and defined as follows:
  \begin{align*}
    & d^R_1 : \quad d_* \succ d^P_1
    \\
    & d^R_i : \quad d^R_{i-1} \succ d^P_i
    && \forall i\in\{2,\ldots,k\}
    \\
    & d^A_i : \quad d^R_1 \succ d^R_2 \succ \ldots \succ d^R_k
    && \forall i\in\{1,\ldots,k\}
  \end{align*}
  This construction allows us to embed the workings of $\SDrot$ into $\TTC$ as follows:

  \begin{lemma}
    \label{lem:sdrot-gtc-lb}
  Suppose applicants in $\{d^P_1,\ldots,d^P_k\}$ only rank institutions in $\{d^A_1,\ldots,d^A_k\}$. Furthermore, suppose that applicant $d_*$ submits list $\{ d^R_j \}$ for some $j \in \{0,1,\ldots,k-1\}$, where we denote institution $d_*$ as $d^R_0$.
  Then, applicants in  $\{d^P_1,\ldots,d^P_k\}$ will be matched to institutions in $\{d^A_1,\ldots,d^A_k\}$ according to $\SD_{d^P_{j+1},d^P_{j+2},\ldots,d^P_k}(\cdot)$.
  \end{lemma}
    To prove this lemma,
    consider the run of $TTC$ after $d_*$ points to institution $d^R_j$.  First, $d_*$ and each $d^R_1, \ldots, d^R_{j}$ is matched to their top-ranked institution in one cycle.
    Now, each applicant in $\{d^A_1,\ldots,d^A_k\}$ has the same preference list, and her top-ranked remaining institution is $d^R_{j+1}$. Additionally, $d^R_{j+1}$ now points to $d^P_{j+1}$, so $d^P_{j+1}$ will be matched to her top-ranked institution in $\{d^A_1,\ldots,d^A_k\}$. Next, all of the remaining institutions in $\{d^A_1,\ldots,d^A_k\}$ point (transitively through $d^R_{j+2}$) to $d^P_{j+2}$, who can match to her top-ranked remaining institution. This continues for all additional applicants in $\{d^P_{j+1},\ldots,d^P_{k}\}$.
    This proves \autoref{lem:sdrot-gtc-lb}; see \autoref{fig:gtc-ttc-LB} for an illustration.
  \begin{figure}[tbhp]
    \begin{minipage}[c]{0.55\textwidth}
      \includegraphics[width=\textwidth]{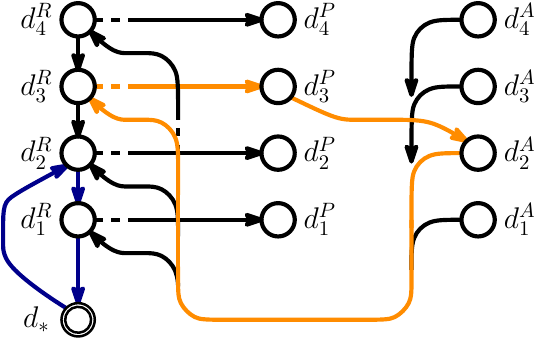}
    \end{minipage}
    \hspace{0.04\textwidth}
    \begin{minipage}[c]{0.4\textwidth}
      \caption{
      Illustration of the set of preferences showing that the outcome-effect complexity of $TTC$ is $\Omega(n^2)$ (\autoref{thrm:gtc-TTC-LB}). 
      }
      {\footnotesize \textbf{Notes:} 
      Two example cycles are highlighted, the first with $d_*$ matching to $h^R_2$, and the second with $d^P_3$ matching to $d_2^A$, her top choice from $\{d^A_1,\ldots,d^A_4\}$.
      \par}
      \label{fig:gtc-ttc-LB}
    \end{minipage}
  \end{figure}

  Thus, applicants in $\{d^P_1,\ldots,d^P_k\}$ are matched to $\{d^A_1,\ldots,d^A_k\}$ according to $\SDrot$, and the outcome-effect complexity of $\TTC$ is at least as high as $\SDrot$, which by \autoref{lem:sdrot-gtc-lb}, is $\Omega(k^2) = \Omega(n^2)$.
\end{proof}

\subsection{\texorpdfstring{$\APDA$}{APDA} and Relation to \texorpdfstring{\cite{GonczarowskiHT22}}{GHT}}
\label{sec:gtc-apda}

We now study the outcome-effect complexity of the stable matching mechanisms $\APDA$ and $\IPDA$. For $\APDA$, this result is an immediate corollary of results from \cite{GonczarowskiHT22}.
We now explain their result and how it implies a bound on the outcome-effect complexity of $\APDA$, as well as how it originally motivated our \autoref{thrm:gtc-TTC-LB} regarding $\TTC$.

\cite{GonczarowskiHT22} look for algorithms for computing matching rules while making the strategyproofness of these matching rules clear. To do this for some strategyproof $f$ and some applicant $d$, they look for algorithms similar to the traditional algorithms used to describe matching rules, but with the following three steps:
\begin{enumerate}[label=(\arabic*)]
  \itemsep0em
  \item Using only $P_{-d}$, the preferences of applicants other than $d$, calculate $d$'s menu $\Menu^f_d(P_{-d})$.
  \item Using $d$'s preferences $P_{d}$, match $d$ to her favorite institution from her menu.
  \item Using $P_{d}$ and $P_{-d}$, calculate the rest of the matching $f(P_d,P_{-d})$.
\end{enumerate}
\cite{GonczarowskiHT22} posits that algorithms written this way are one way to expose strategyproofness. Indeed, to observe strategyproofness, applicant $d$ only has to notice that her report cannot affect her menu, and reporting her true type will always match her to her favorite institution on her menu.

\cite{GonczarowskiHT22} prove that for $\TTC$ and each applicant $d$, there is an algorithm meeting the above three-step outline that is very similar to the traditional algorithm. In fact, this algorithm follows directly from the fact that $\TTC$ is independent of the order in which cycles are matched (\autoref{thrm:TTC-order-independent})---the cycle involving $d$ can simply be matched as late as possible (for details, see \cite{GonczarowskiHT22}).
In contrast, \cite{GonczarowskiHT22} prove that for any algorithm for $\APDA$ meeting the above three-step outline, if this algorithm reads each applicant's preference in favorite-to-least-favorite order (as the traditional algorithm for $\APDA$ does), then the algorithm \emph{requires $\Omega(n^2)$ memory} (much more than the $\widetilde O(n)$ bits of memory used by the traditional algorithm).
This gives one sense in which the traditional $\APDA$ algorithm obscures the menu, and hence strategyproofness;
for additional discussion, see \cite{GonczarowskiHT22}.

To compare our outcome-effect complexity with the results of \cite{GonczarowskiHT22}, observe that the outcome-effect complexity of any mechanism exactly equals the memory requirements of an algorithm with the following \emph{two} steps:
\begin{enumerate}[label=(\arabic*)]
    \itemsep0em
    \item Perform any calculation whatsoever using only $P_{-d}$.
    \item Calculate the entire matching $f(P_d,P_{-d})$ using only $P_d$.
\end{enumerate}
Thus, outcome-effect complexity is a coarsening of the complexity results proven in \cite{GonczarowskiHT22}, and our lower bound
for the outcome-effect complexity of $\TTC$ (\autoref{thrm:gtc-TTC-LB}) shows one way in which the three-step outline for $\TTC$ is tight---the last step~(3) is crucial
(and in particular, the ``pointing graph'' in the $\TTC$ algorithm at the end of Step (1) above does not contain \emph{nearly} enough information to calculate the rest of the matching).
In some sence, \autoref{thrm:gtc-TTC-LB} shows that from a ``global'' perspective, $\TTC$ is just as complex as $\APDA$, and that the three-step outline very precisely captures the sense in which $\TTC$ is ``strategically simpler'' than $\APDA$ under the framework of \cite{GonczarowskiHT22}.

For our direct purposes, we observe the following corollary of the construction in \cite{GonczarowskiHT22}.
\begin{corollary}[Follows from {\cite{GonczarowskiHT22}}]
  \label{thrm:gtc-DA-LB}
  The outcome-effect complexity of $\APDA$ is $\Omega(n^2)$.
\end{corollary}
\begin{proof}
  The main impossibility theorem of \cite{GonczarowskiHT22} directly constructs a set of $2^{\Omega(n^2)}$ preferences for applicants other than $d_*$ such that the function $\APDA(\cdot,P_{-d_*}) : \T_{d_*} \to \M$ is distinct for each $P_{-d_*}$ in this class.
\end{proof}

In the next section, we lower bound the outcome-effect complexity for $\IPDA$.\footnote{
While $\IPDA$ is not strategyproof, and hence the above three-step outline does not apply to $\IPDA$, there is a different sense in which a outcome-effect lower bound for $\IPDA$ shows that an algorithm from \cite{GonczarowskiHT22} is tight. 
Namely, \cite[Appendix D.3]{GonczarowskiHT22} constructs a delicate algorithm using the outcome of $\IPDA$ as a building block, and if there were an algorithm $A$ that were able to calculate and store the function $\IPDA(\cdot,P_{-d_*})$ in $\widetilde O(n)$ bits, then the delicate algorithm of \cite{GonczarowskiHT22} could have been easily implemented using calls to $A$.
}

\subsection{\texorpdfstring{$\IPDA$}{IPDA}}
\label{sec:gtc-ipda}

We next turn our attention to the outcome-effect complexity of $\IPDA$.
While $\APDA$ and $\IPDA$ are both stable mechanisms, they operate in a different way and can produce very different matchings.
Additionally, it turns out $\IPDA$ is much harder to reason about from the perspective of outcome-effect complexity than $\APDA$ is, as we discuss next.

\paragraph{Building up to our construction.}
Informally, in $\APDA$, an applicant $d_*$ can ``trigger any subset of effect'' by simply making some subset of proposals.
In contrast, in $\IPDA$, each applicant (including $d_*$) rejects all proposals \emph{except for} (at most) one, so she must ``trigger all effects \emph{except for} one''. 
Thus, the class of effects $d_*$ could trigger may seem far more limited in $\IPDA$ than in $\APDA$.
Nonetheless, our construction overcomes this difficulty by making $d_*$ receive $\Omega(n)$ proposals in $\IPDA$ all sequentially.
Between these proposals to $d_*$, the match of \emph{every} other applicant changes.
Informally, this allows us to use $d_*$'s choice of which proposal to accept to ``determine the rest of the matching''. In particular, $\Omega(n)$ bits will be needed for each possible choice from $d_*$, giving our final outcome-effect complexity bound of $\Omega(n^2)$.
In the end, we show the outcome-effect complexity of $\IPDA$ is high:

\begin{theorem}
  \label{thrm:gtc-IPDA}
  The outcome-effect complexity of $\IPDA$ is $\Omega(n^2)$.
\end{theorem}
\begin{proof}
  Fix $k$, where we will have $n = |\Insts| = |\Appls| = \Theta(k)$. 
  The institutions are $h^0_i, h^1_i$ for $i=1,\ldots,k$ and $h^R_i$ for $i=0,\ldots,k$, and the applicants are $d_i, d_i'$ for $i=1,\ldots,k$ and $d^R_i$ for $i=1,\ldots,k$, as well as $d_*$. First we define the fixed priorities $Q$ of the institutions (where, for the entirety of this construction, we take indices mod $k$):
  \begingroup
  \allowdisplaybreaks
\begin{align*}
  & h_i^b :
    d_i \succ d_i' \succ d_{i+1} \succ d_{i+1}' \succ \ldots \succ d_{i-1} \succ d_{i-1}'
    \qquad
    && \text{For each $i=1,\ldots,k$ and $b \in \{0,1\}$}
    \\
  & h^R_0 : d_* \succ d_1^R 
  \\
  & h^R_1 : d_* \succ d_1 \succ d_1' \succ d^R_2
  \\
  & h^R_i : d^R_i \succ d_* \succ d_1 \succ d_1' \succ d^R_{i+1}
    && \text{For each $i=2,\ldots,k-1$}
  \\
  & h^R_k : d^R_k \succ d_*
\end{align*}
  \endgroup
The preferences of the applicants $\{d_1',\ldots,d_n',d^R_0,d^R_1,\ldots,d^R_k\}$ are fixed.  The preferences of applicants $\{d_1,\ldots,d_k\}$ depend on bits $(b_{i,j})_{i,j\in\{1,\ldots,k\}}$ where each $b_{i,j}\in\{0,1\}$. Since the run of $\IPDA$ will involve applicants receiving proposals in (loosely) their reverse order of preference, we display the preferences of applicants in worst-to-best order for readability. 
The preferences are as follows:
First, for applicants in $\{d^R_1,\ldots,d^R_k\}$: 
\begin{align*}
  & d^R_1 : h^R_0
  &\qquad\qquad&
  & d^R_i : h^R_{i} \prec h^R_{i-1}
    && \text{For each $i=2,\ldots,k$}
\end{align*}
Next, for $i=2,3,\ldots,k$, we have:
\begin{align*}
  &
  d_i : & h_{i}^{1 - b_{1,1}} & \prec h_{i}^{b_{1,1}} 
    & \prec && h_{i-1}^{1-b_{1,2}} & \prec h_{i-1}^{b_{1,2}} 
    & \prec && h_{i-2}^{1-b_{1,3}} & \prec h_{i-2}^{b_{1,3}}
    & \prec \ldots \prec &&
    & & h_{i+1}^{1-b_{1,k}} & \prec h_{i+1}^{b_{1,k}}
  \\ &
  d_i' : & h_{i}^{0} & \prec h_{i}^{1}
    & \prec && h_{i-1}^{0} & \prec h_{i-1}^{1} 
    & \prec && h_{i-2}^{0} & \prec h_{i-2}^{1}
    & \prec \ldots \prec && 
    & & h_{i+1}^{0} & \prec h_{i+1}^{1}
\end{align*}
In words, applicants of the form $d_i'$ always prefer institutions in the cyclic order, starting with institutions of the form $h_i^b$ as their least-favorites. Applicants of the form $d_i$ also rank institutions like this, but they flip adjacent places in this preference based on the bits $b_{i,j}$. 

Finally, we define the preferences of $d_1$ and $d_1'$, which are like the other $d_i$ and $d_i'$, except that these applicants will also accept proposals from institutions of the form $h^R_j$. Specifically:
\begin{align*}
  &
  d_1 : & h_{1}^{1 - b_{1,1}} & \prec h_{1}^{b_{1,1}} 
    & \prec {h^R_1} \prec
    &&  h_{k}^{1-b_{1,2}} & \prec h_{k}^{b_{1,2}}
    & \prec {h^R_2} \prec
    && h_{k-1}^{1-b_{1,3}} & \prec h_{k-1}^{b_{1,3}}
    & \prec \ldots \prec {h^R_{k-1}}\prec
    & & h_{2}^{1-b_{1,k}} & \prec h_{2}^{b_{1,k}}
  \\ &
  d_1' : & h_{1}^{0} & \prec h_{1}^{1} 
    & \prec {h^R_1} \prec
    && h_{k}^{0} & \prec h_{k}^{1}
    & \prec {h^R_2} \prec
    && h_{k-1}^{0} & \prec h_{k-1}^{1}
    & \prec \ldots \prec {h^R_{k-1}} \prec
    & & h_{2}^{0} & \prec h_{2}^{1}
\end{align*}

These preferences are illustrated in \autoref{fig:gtc-ipda-LB}, along with an informal description of how the preferences operate. Formally, our key claim is the following:
\begin{figure}[tbp]
  \begin{center}
  \includegraphics[width=\textwidth]{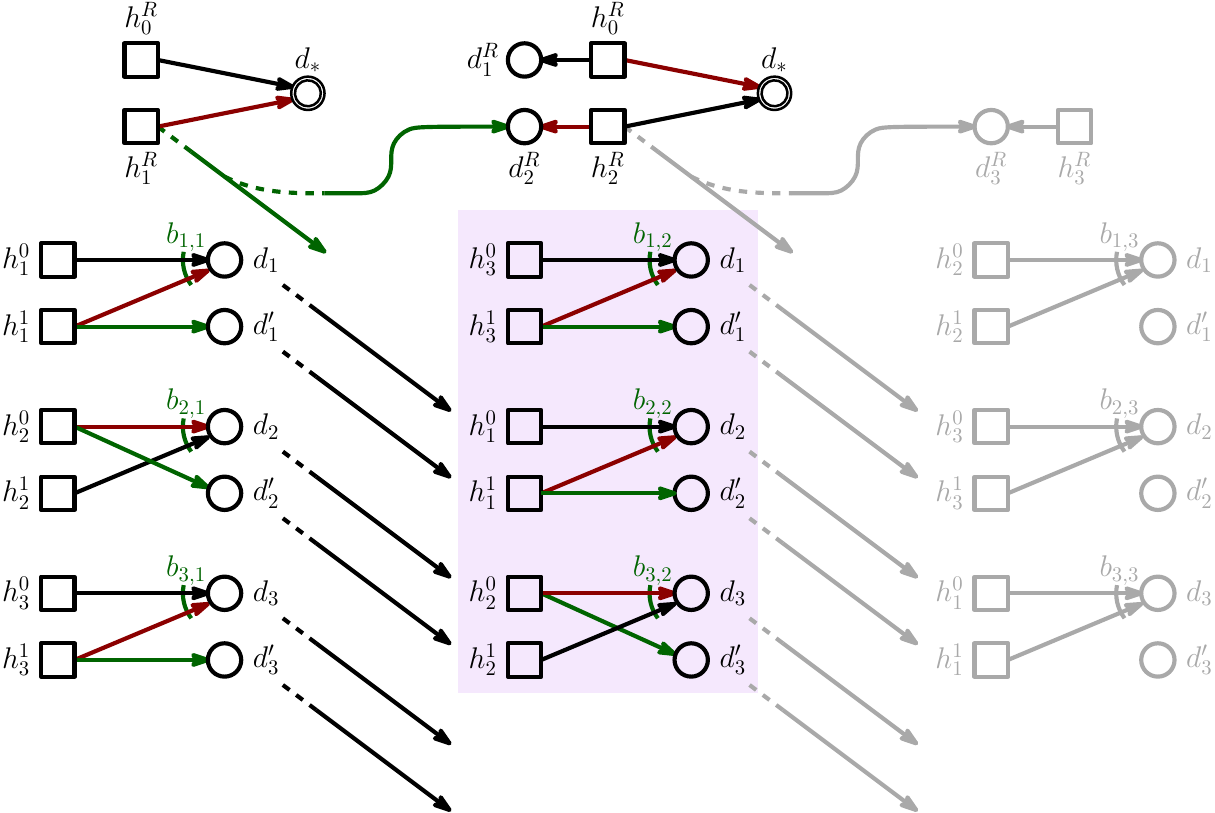}
  \end{center}
  \caption{
  Illustration of the set of preferences showing that the outcome-effect complexity of $\IPDA$ is $\Omega(n^2)$ (\autoref{thrm:gtc-IPDA}).   }
  \vspace{0.1in}
  {\footnotesize \textbf{Notes:}
  When $d_*$ is matched to $h^R_j$, all the institutions in $h^R_1, h^R_2, \ldots, h^R_{j-1}$ propose to $d_1$ and $d_1'$. This successively ``rotates'' all of the institutions and applicants once for each $h^R_{\ell}$ that is rejected by $d_*$, so the institutions and applicants are rotated $j-1$ times. When this happens, the bits $b_{i, j}$ determine the match of each applicant $d_i$, which shows that all of the bits $b_{i,j}$ can affect the matching for some possible report of $d_*$. In the illustration, $d_*$ matches to $h^R_2$, so the applicants are rotated once; the selected matching is in the highlighted box, and proposals made in hypothetical future rotations are shown in grey.
  \par }
  \label{fig:gtc-ipda-LB}
\end{figure}

\begin{lemma}
  \label{thrm:gtc-ipda-main-lemma}
  Consider any $j \in \{1,\ldots,k\}$. If $d_*$ submits a preference list containing only $\{ h^R_j \}$, then for each $i \in \{1,\ldots,k\}$, it holds that $d_i$ is matched to $h_{i - j + 1}^{b_{i, j}}$ in $\IPDA$.\footnote{
  This construction would also work with the full-length list
  $h^R_k \succ \ldots \succ h^R_{j+1} \succ h^R_j \succ h^R_0 \succ h^R_{j-1} \succ \ldots \succ h^R_1 \succ \ldots$.
}
\end{lemma}
To prove this lemma, start by considering $j=1$. 
When $j = 1$, each $h_i^0$ and $h_i^1$ for $i=1,\ldots,k$ simply proposes to $d_i$, and $d_i$ prefers and tentatively accept  $h_i^{b_{i,1}}$.
Moreover, $d_*$ rejects $h^R_0$, who then proposes to $d^R_1$, and no further proposals are made. 

Now consider $j \ge 2$, and suppose we choose an order of proposals such that $d_*$ has already rejected $h^R_{j-2}$, but has not yet rejected $h^R_{j-1}$. 
This is equivalent to tentatively considering the matching where $d_*$ submits only list $h^R_{j-1}$, so by induction, $d_i$ and $d_i'$ are matched to $h_{i-j+2}^0$ and $h_{i-j+2}^1$, respectively, for each $i$. In particular, $d_1$ and $d_1'$ are matched to $h^0_{3-j}$ and $h^1_{3-j}$ (equivalently, $h^0_{k - j + 3}$ and $h^1_{k-j+3}$).
Now consider what happens when $d_*$ rejects $h^R_{j-1}$.
First, $h^R_{j-1}$ proposes to $d_1$, and the proposal is tentatively accepted.
This causes rejections among $d_1$ and $d_1'$ which lead $h^0_{3-j}$ to propose to $d_2$.
This leads analogously to  $h^0_{4-j}$ proposing to $d_3$.
This continues similarly, with each $h^0_{i - j +1}$ proposing to $d_i$ for each $i$, until $h^0_{2-j}$ proposes to $d_1$. 
This proposal is accepted, causing $h^R_{j-1}$ to propose to $d_1'$, and $h^1_{3-j}$ proposes to $d_2$. Regardless of which proposal among $h^0_{3-j}$ and $h^1_{3-j}$ is accepted by $d_2$, we next have $h^1_{4-j}$ proposing to $d_3$.
This continues similarly, with each $h^1_{i - j +1}$ proposing to $d_i$ for each $i$, until $h^1_{2-j}$ proposes to $d_1$.
Regardless of whether $d_1$ favors $h^0_{2-j}$ or $h^1_{2-j}$, this leads one of them to propose to $d_1'$ and thus $h^R_{j-1}$ is rejected by $d_1'$. Finally, $h^R_{j-1}$ proposes to $d^R_j$, and $d_*$ next receives a proposal from $h^R_j$, which she accepted, ending the run of $\IPDA$.

All told, when $h_*$ submits list $\{h^R_j\}$, each $d_{i}$ receives a proposal from $h_{i-j+1}^0$ and $h_{i-j+1}^1$, and picks and is finally matched to whichever of the two she prefers according to $b_{i,j}$.
This proves \autoref{thrm:gtc-ipda-main-lemma}.

Thus, for each possible profile of bits $b = ( b_{i,j} )_{i,j\in \{1,\ldots k\}}$, there is a distinct function $\IPDA_Q(\cdot, P_{-d_*})$.
This proves that the outcome-effect complexity of $\IPDA$ is at least $k^2 = \Omega(n^2)$.
\end{proof}

\section{Options-Effect Complexity}
\label{sec:type-to-menu}
We now turn to the \nameref{ppg:type-to-menu} Question: \questionTextTypeToMenu{}
In this section, we study the function from one applicant's report to another applicant's menu, and quantify the complexity of this function.
Like the \nameref{ppg:type-to-match} Question from \autoref{sec:gtc}, this gives a lens into the complexity of how one applicant can affect the mechanism.\footnote{
There is also a technical connections between the outcome-effect and options-effect complexity: both measure the complexity of some function from one applicant's type, a $\widetilde O(n)$-bit piece of data, onto another $\widetilde O(n)$-bit piece of data (the outcome matching, or another applicant's menu, respectively). In some sense, many of our supplementary results in \autoref{sec:additional-results} show that related mappings from $\widetilde O(n)$ bits to a small number of bits (say, $O(\log n)$) will not suffice to capture the relevant complexities.
} However, the options-effect complexity captures very different phenomena, and unlike for the outcome-effects complexity, we will show that the options-effect complexity separates $\TTC$ and $\DA$.
Our main definition for this section is:

\begin{definition}
  The \emph{options-effect complexity} of a matching mechanism $f$ is
  \[ \log_2 \max_{d_*,d_\dagger\in\Appls} \left|\left\{ 
    \Menu_{d_\dagger}^f(\cdot, P_{-\{d_*,d_\dagger\}})
    \ \middle|\ P_{-\{d_*,d_\dagger\}}\in\T_{-\{d_*,d_\dagger\}}
    \right\}\right|, 
  \]
  where $\Menu_{d_\dagger}^f(\cdot, P_{-\{d_*,d_\dagger\}}) : \T_{d_*} \to 2^{\Insts}$ is the function mapping each $P_{d_*}\in\T_{d_*}$ to the menu $\Menu_{d_\dagger}^f(P_{d_*}, P_{-\{d_*,d_\dagger\}})\subseteq \Insts$ of $d_\dagger$.
\end{definition}

We will show that this complexity is high for $\TTC$, providing another novel way in which $\TTC$ is complex. In contrast, for $\DA$ we give a new structural characterization which shows that this complexity is low in stable matching mechanisms. In \autoref{sec:application-SEDA-1-to-1}, we use this result to give additional characterizations and connections to \cite{GonczarowskiHT22}, illustrating how this result may be of independent interest.

To begin, we note that the same construction used in the proof of \autoref{thrm:gtc-SD-UB} suffices to show that the options-effect complexity of $\SD$ is low:
\begin{corollary}
  \label{thrm:type-to-menu-SD}
  The options-effect complexity of $\SD$ is $\widetilde O(n)$.
\end{corollary}
\begin{proof}
Without loss of generality, assume $d_*$ is first in the priority order, and $d_\dagger$ is last. Then, using the construction from the proof \autoref{thrm:gtc-SD-UB}, one can represent the function from $P_{d_*}$ to the outcome matching before $d_\dagger$ selects their match. Now, once this matching is known, observe that $d_\dagger$'s menu is simply the set of remaining institutions. So the same $\widetilde O(n)$-bit construction as in the proof of \autoref{thrm:gtc-SD-UB} suffices to represent the mapping $\Menu_{d_\dagger}^f(\cdot, P_{-\{d_*,d_\dagger\}})$, completing the proof.
\end{proof}

\subsection{\texorpdfstring{$\TTC$}{TTC}}
\label{sec:type-to-menu-TTC}

We now prove an $\Omega(n^2)$ lower bound on the options-effect complexity of $\TTC$.

\paragraph{The need for new ideas.}
One might hope that the ideas behind \autoref{thrm:gtc-TTC-LB}, which reduces the question of the outcome-effect complexity of $\TTC$ to the complexity of $\SDrot$, might gain traction towards bounding the options-effect complexity of $\TTC$ as well. However, it turns out that $\SDrot$ has \emph{low} options-effect complexity ($\widetilde O(n)$), as we prove in \autoref{thrm:type-to-menu-SDrot} for completeness. Thus, embedding $\SDrot$ into $\TTC$ cannot establish a lower bound on the options-effect complexity, so new ideas are needed.

\paragraph{Building up to our construction.}
To illustrate the key ideas behind our construction, start by considering how the menu of an applicant $d_\dagger$ could be changed based on the reports of other applicants. Suppose that (as discussed in \autoref{sec:gtc-ttc}) we delay matching the cycle involving $d_\dagger$ as long as possible. Suppose that when we do this, some set $S$ of institutions all point to some applicant $d_S$, and a completely different set $T$ point to $d_T$. Then, if $d_S$ points to $d_\dagger$ but $d_T$ does not, then $d_\dagger$'s menu may contain only $S$. Likewise, if only $d_T$ points to $d_\dagger$, then $d_\dagger$'s menu is $T$. The key to our construction is to allow one applicant $d_*$ to ``select'' a single applicant $d_{\mathsf{sel}}$ from within a large gadget. Every applicant in the gadget \emph{except} $d_{\mathsf{sel}}$ will then be matched, and then $d_{\mathsf{sel}}$ will point to $d_\dagger$. In the end, $d_\dagger$'s menu will consist of exactly those institutions that point to $d_{\mathsf{sel}}$.

\begin{theorem}
  \label{thrm:type-to-menu-TTC}
  The options-effect complexity of $\TTC$ is $\Omega(n^2)$.
\end{theorem}
\begin{proof}
  As in the proof of \autoref{thrm:gtc-TTC-LB}, we prove this bound using a construction with $n = |\Appls| = |\Insts|$, and each institution ranking a distinct, unique applicant with highest priority (equivalent to the case of a housing market), and for convenience, we identify each institution with the applicant that it prioritizes highest, e.g., we use $d^X_1$ to denote both an applicant and the institution that prioritizes $d^X_1$ highest.

  For some $k = \Theta(n)$, we consider applicants $\Appls = \{ d_*, d_\dagger, d^X_0, d^Y_0, d^X_1,d^Y_1,\ldots,d^X_k, d^Y_k, \allowbreak d^T_1, d^T_2,\allowbreak\ldots, d^T_k \}$. 
  The key to our construction is a ``selection gadget'' created from applicants of the form $d^X_j$ and $d^Y_j$. Their preferences are fixed, as follows:
  \begin{align*}
    d^X_0 
    & : d^Y_0 \succ d^X_0
    \\
    d^X_j 
    & : d^Y_j \succ d^X_0 \succ d_\dagger \succ d^X_j
      && \text{For each $j \in \{1,\ldots,k\}$}
    \\
     d^Y_j 
     & : d_* \succ d^X_{j+1}
       && \text{For each $j \in \{0, 1,\ldots,k-1\}$}
    \\
     d^Y_k 
     & : d^Y_1 \succ d^Y_2 \succ \ldots \succ d^Y_{k-1} \succ d^T_k
  \end{align*}
  
  Our key claim shows that when $d_*$ points to $d^Y_{j-1}$, the selection gadget causes only $d^X_j$ to remain unmatched.
  \begin{lemma}
    \label{thrm:ttm-ttc-key-lemma}
    Suppose $d_*$ submits preference list $\{ d^Y_{j-1} \}$, for $j \in \{1,\ldots,k\}$. Then every applicant in the selection gadget except for $d^X_j$ is matched to other agents in the selection gadget. 
  \end{lemma}
  To prove this lemma, observe that when $d_*$ points to $d^Y_{j-1}$, she matches to $d^Y_{j-1}$. Next, cycle $d^X_0, d^Y_0, d^X_1, \ldots, d^X_{j-1}$ matches. Finally, cycle $d^Y_j, d^Y_{j+1},\ldots,d^Y_k$ matches.
  This proves \autoref{thrm:ttm-ttc-key-lemma}; see \autoref{fig:TTM-TTC} for an illustration.

\begin{figure}[tbp]
 \includegraphics[width=\textwidth]{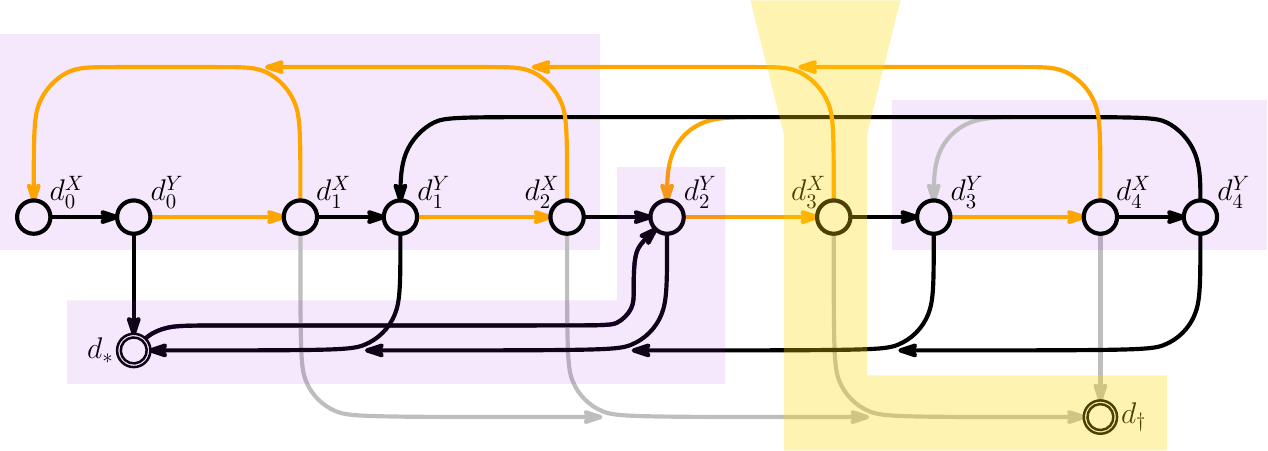}
 \caption[Type-to-menu TTC]{
 Illustration of the selection gadget used to bound the options-effect complexity of $\TTC$ (\autoref{thrm:type-to-menu-TTC}). 
 }
 \vspace{0.1in}
 {\footnotesize \textbf{Notes:}
   Applicants' first choices are denoted by black colored arrows, and their second and third choices are colored orange and gray respectively.
   The illustration shows an example where $d_*$ matches to $h^Y_2$.
   After this, $\{d_0^X, d_0^Y, d_1^X, d_1^Y, d_2^X\}$ complete a cycle, then $\{d^Y_3, d^X_4, d^Y_4\}$ complete a cycle. 
  Finally, $d^X_3$ points to $d_\dagger$, and $d_\dagger$'s menu is determined by precisely which applicants in $\{d^T_1,d^T_2,d^T_3,d^T_4\}$ point to $d^X_3$.
  \par }
  \label{fig:TTM-TTC}
\end{figure}

  Now, for each $i = 1,\ldots,k$, we consider applicant $d^T_i$ whose preferences are defined by an arbitrary subset $T_i \subseteq \{ d^X_1, d^X_2, \ldots, d^X_k \}$. Namely,
  \[ d^T_i : T_i \succ d^T_i, \]
  where the element of $T_i$ may be placed on this list in an arbitrary fixed order.

  Finally, observe that $d_\dagger$'s menu consists precisely of those institutions that remain unmatched after all possible cycles not involving $d_\dagger$ are eliminated.\footnote{
     This follows by the fact that $\TTC$ is independent of the order in which cycles are eliminated (\autoref{thrm:TTC-order-independent}), and the fact that after all cycles not including $d_\dagger$ have been matched, $d_\dagger$ must complete a cycle regardless of which institution they point to \cite{GonczarowskiHT22}.
   } 
  Consider this process, with $d_*$ submitting list $\{ d^Y_{j-1} \}$.
  \autoref{thrm:ttm-ttc-key-lemma} shows that exactly those $d^T_i$ such that $j \in T_i$ transitively point to $d_\dagger$ through applicant $d^X_j$, and for all other $i$, $d^T_i$ will be matched.
  Thus, $d_\dagger$'s menu is exactly $\{ d^X_j \} \cup \{ d^T_i\ |\ j \in T_i \}$.
  Thus, there is a distinct function $\Menu_{d_\dagger}^{\TTC}(\cdot, P_{-\{d_*, d_\dagger\}})$ for each distinct profile of subsets $T_1,\ldots,T_n$. There are $2^{k^2}$ such sets, showing that the options-effect complexity of $\TTC$ is $\Omega(k^2) = \Omega(n^2)$.
\end{proof}

\subsection{DA}
\label{sec:type-to-menu-DA}

Thus far, our study of quantifying the complexities of matching mechanisms has largely yielded negative results and as-high-as-possible lower bounds. 
However, perhaps surprisingly, we will next present a positive result for the options-effect complexity of $\APDA$ and $\IPDA$, showing that one applicant can only affect another applicant's menu in a combinatorially simple way. 
This result will take the form of a novel ``effects DAG'' which shows how one applicant's preference inputs can affect another. 
This is somewhat reminiscent of how \autoref{thrm:gtc-SD-UB} and \autoref{fig:type-to-outcome-sd} bound the outcome-effect complexity of $\SD$, except tailored to representing the effect on menus instead of matchings; this reminiscence may be surprising, since by the negative results of \autoref{sec:gtc} for $\DA$, the outcome-effect complexity of $\DA$ is far higher than that of $\SD$.

To begin, we remark that by \autoref{thrm:same-menu-all-stable}, the menu in $\APDA$ is identical to the menu in $\IPDA$, so proving this bound for the two mechanisms is identical. For the remainder of this section, we thus simply refer to the menu in $\DA$.

\begin{theorem}
  \label{thrm:type-to-menu-DA}
  The options-effect complexity of $\DA$ is $\widetilde \Theta(n)$.
\end{theorem}

The remainder of this subsection is dedicated to proving this theorem.
Consider any pair of applicant $S = \{ d_*, d_\dagger \}$, priorities $Q$ and preferences $P_{-\{d_*,d_\dagger\}}$.
Our goal is to represent the function 
$\Menu_{d_\dagger}^{\APDA_Q}(\cdot, P_{-S}) : \T_{d_*} \to 2^{\Insts}$ using an $\widetilde O(n)$-bit data structure.
The starting point of our representation will be the following fact from \cite{GonczarowskiHT22}:
\begin{lemma}[Follows from {\cite[Description 1]{GonczarowskiHT22}}]
  \label{thrm:ght-main-positive-DA}
  The menu of $d_\dagger$ in $\DA$ under priorities $Q$ and preferences $P_{-d_\dagger}$ is exactly the set of proposals $d_\dagger$ receives in $\IPDA_Q(d_\dagger:\emptyset, P_{-d_\dagger})$, i.e., the set of proposals $d_\dagger$ receives if $\IPDA$ is run with $d_\dagger$ rejecting all proposals.
\end{lemma}

Note that, while this lemma characterizes the menu in both $\APDA$ and $\IPDA$, the mechanism $\IPDA$ specifically must be used to achieve this characterization (see \cite{GonczarowskiHT22} for a discussion).
In every run of $\IPDA$ for the remainder of this subsection, the priorities $Q$ and preferences $P_{-S}$ will be fixed. Thus, we suppress this part of the notation, and write $\IPDA(d_S : P_S)$ in place of $\IPDA_Q(d_S : P_S, d_{-S} : P_{-S})$.

Remarkably, beyond \autoref{thrm:ght-main-positive-DA}, the \emph{only} property of $\DA$ that we will use in this proof (beyond the definition of how $\DA$ is calculated) is the fact that $\IPDA$ is independent of the order in which proposals are chosen (\autoref{thrm:da-indep-execution}).
We start by defining a graph representing collective data about multiple different runs of $\IPDA_Q$ under related preference lists in a cohesive way. This data structure is parametrized by a general set $S$ in order to avoid placing unnecessary assumptions and in hope that it will be of independent interest, but we will always instantiate it with $S = \{d_*, d_\dagger\}$.

\begin{definition}
  Fix a set $S \subseteq \Appls$ of applicants, and a profile of priorities $Q$ and preferences $P_{-S}$ of all applicants other than those in $S$. For a $d \in S$ and $h \in \Insts$, define an ordered list of pairs in $S\times\Insts$ called $\chain(d,h)$ as follows: First, calculate $\mu = \IPDA(d : \{ h \}, d_{S\setminus\{d\}} : \emptyset)$. If $h$ never proposes to $d$, set $\chain(d,h)=\emptyset$. Otherwise, starting from tentative matching $\mu$, let $d$ reject $h$ and have $h$ continue proposing, following the rest of the execution of $\IPDA$ with $d$ rejecting all proposals. Note that this constitutes a valid run of $\IPDA(d_S : \emptyset)$, and that during this ``continuation'' there is a unique proposal order because only a single element of $\Insts$ is proposing at any point in time. Now, define $\chain(d,h)$ as the ordered list of pairs
  \[ (d = d_0, h = h_0) \longrightarrow (d_1, h_1) \longrightarrow (d_2, h_2) \longrightarrow \ldots \longrightarrow (d_k, h_k),
  \]
  where $d_i \in S$ is each applicant in $S$ receiving a proposal from $h_i \in \Insts$ during the continued run of $\IPDA$ after $d=d_0$ rejects $h=h_0$, written in the order in which the proposals occur.

  Now, define the \emph{un-rejection graph} $\UnrejGr = \UnrejGr(Q,P_{-S})$ as the union of all possible consecutive pairs in $\chain(d,h)$ for all $d \in S$ and $h \in \Insts$. In other words, $\UnrejGr$ is a directed graph defined on the subset of pairs $S \times \Insts$ which occur in some $\chain(d,h)$, where the edges are all pairs $(d_i, h_i) \longrightarrow (d_{i+1}, h_{i+1})$ which are consecutive elements of some $\chain(d,h)$.
\end{definition}

Note that $\UnrejGr \setminus \chain(d,h)$ is exactly the set of rejections which applicants in $S$ make in $\IPDA(d : \{ h\}, d_{S\setminus{d}} : \emptyset)$, i.e., the set of pairs $(d',h')\in S\times\Insts$ for which $d$ rejects $h$ in this run of $\IPDA$.
We start by establishing general properties of $\UnrejGr$.

\begin{lemma}
  \label{lem:d-h-chain-tail}
  Consider any $\chain(d,h) \ne \emptyset$, and node $(d', h')$ contained in $\chain(d,h)$. It must be the case that $\chain(d',h')$ equals the tail of $\chain(d,h)$, including and after $(d',h')$.
\end{lemma}
\begin{proof}
  Consider any such $(d,h)$ and $(d', h')$. To start, note that by the definition of $\chain(d,h)$, we cannot have $h'$ propose to $d'$ during $\IPDA(d: \{h\}, d_{S\setminus \{d\}}:\emptyset)$. Thus, this run of $\IPDA$ will produce exactly the same matching as $\IPDA(d: \{h\}, d' : \{h'\}, d_{S\setminus \{d,d'\}} : \emptyset)$, so letting $d$ reject $h$ on top of this must produce a valid run of $\IPDA(d' : \{h'\}, d_{S\setminus \{d'\}})$ by the fact that $\DA$ is independent of execution order (\autoref{thrm:da-indep-execution}). After this run, we can now let $d'$ reject $h'$ in order to calculate the $\chain(d',h')$. But this will also correspond exactly to the part of the $\chain(d,h)$ after $h'$ proposes to $d'$, i.e., the tail of the $\chain(d,h)$ including and after $(d',h')$. This finishes the proof.
\end{proof}

\begin{lemma}
  $\UnrejGr$ is a DAG with out-degree at most $1$ at each node.
\end{lemma}
\begin{proof}
  Consider any vertex $(d,h)$ in $\UnrejGr$, and consider $\chain(d,h)$. There is at most one edge outgoing from $(d,h)$ in the $\chain(d,h)$ by definition. By \autoref{lem:d-h-chain-tail}, if $(d,h)$ appears in any other possible $\chain(d',h')$, then the edge outgoing from $(d,h)$ in this chain must be the same as in $\chain(d,h)$. Thus, the vertices in $\UnrejGr$ have out-degree at most one.

  Since $\UnrejGr$ has out-degree at most $1$, every possible max-length path in $\UnrejGr$ must equal $\chain(d,h)$ for some value of $(d,h)$ (namely, the first $(d,h)$ along the path).
  Now, this implies that $\UnrejGr$ must be acyclic, because each possible $\chain(d,h)$ is acyclic by definition.
\end{proof}

\begin{figure}[tbp]
  \begin{minipage}[c]{0.35\textwidth}
    \centering
    \includegraphics[width=\textwidth]{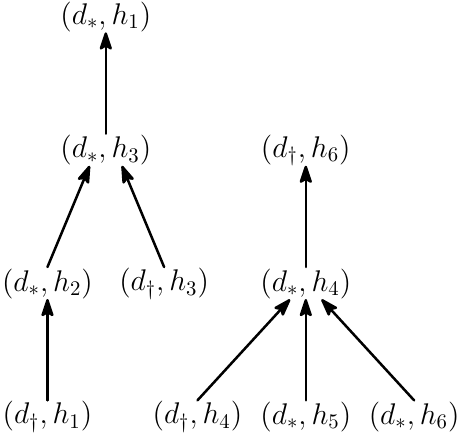}
  \end{minipage}
  \qquad
  \begin{minipage}[c]{0.6\textwidth}
  \caption[Type-to-menu DA Example]{
  Illustration of $\UnrejGr$ with the following priorities and preferences:
  }
  \label{fig:TTM-DA-primer}
  \vspace{-0.2in}
  \begin{align*}
       & h_1 : d_\dagger \succ d_2 \succ d_*
          &\qquad&
    \\ & h_2 : d_2 \succ d_* \succ d_3
          && d_2 : h_3 \succ h_1 \succ h_2
    \\ & h_3 : d_\dagger \succ d_3 \succ d_* \succ d_2
          && d_3 : h_2 \succ h_3
    \\[-0.15in]
    \\ & h_4 : d_4 \succ d_\dagger \succ d_* \succ d_5
          && d_4 : h_5 \succ h_6 \succ h_4
    \\ & h_5 : d_* \succ d_4 \succ d_5
          && d_5 : h_4 \succ h_6
    \\ & h_6 : d_* \succ d_4 \succ d_5 \succ d_\dagger
          &&
  \end{align*}
  \end{minipage}
\end{figure}

Thus, $\UnrejGr$ is a forest. 
See \autoref{fig:TTM-DA-primer} for an example and illustration.

The following notation will be highly convenient: 
\begin{definition}
  For two nodes $v, w$ in $\UnrejGr$, we write $v\trianglelefteq w$ if there exists a path in $\UnrejGr$ from $v$ to $w$.
  For a subset $T$ of vertices in $\UnrejGr$, we write $\prop_d(T)\subseteq\Insts$ to denote the set of all $h\in\Insts$ such that $(d,h) \in T$; we also refer to nodes of the form $(d,h)$ as $d$-nodes.
\end{definition}
Note that $\trianglelefteq$ this defines a partial order on the nodes of $\UnrejGr$, and that $\chain(d,h)$ equals the set of all $(d',h') \in \UnrejGr$ such that $(d,h) \trianglelefteq (d',h')$.

Since $\UnrejGr$ is defined in terms of runs of $\IPDA$ when some $d \in S$ submits a list of the form $\{h\}$, it is not clear how this relates to what will happen when $d$ submits an arbitrary list $P_d$. The following definition will end up providing the connection we need between $\UnrejGr$ and $\IPDA(d : P_d, P_{S\setminus\{d\}} : \emptyset)$ via a combinatorial characterization of the match of $d$ under list $P_d$.

\begin{definition}
  Fix $S$, priorities $Q$, and preferences $P_{-S}$. For any $d\in S$ and $P_d \in \T_d$, define $\stab_d(P_d)$ as the set of $d$-nodes $v = (d,h) \in \UnrejGr$ such that $h \succ_d^{P_d} h'$ for all $h' \in \prop_d(\UnrejGr \setminus \chain(v))$. In words, $\stab_d(P_d)$ is the set of all $d$-vertices $v$ in $\UnrejGr$ such that $P_d$ prefers $v$ to all vertices $v'$ which do not come after $v$ according to $\trianglelefteq$ (in particular, $P_d$ must rank each $h$ in $\prop_d(\stab_d(P_d))$ as acceptable).\footnote{
    One can show that $\stab_d(P_d)$ is exactly the set of stable partners of $d$ under priorities $Q$ and preferences $(d: P_d, d_{S\setminus\{d\}}: \emptyset, d_{-S}: P_{-S})$ (for instance, using the techniques of \cite{CaiT19}); this fact is not needed for our arguments, but it is what inspired the name $\stab_d(P_d)$.
  }

  Define $\stabmatch_d(P_d)$ as the $\trianglelefteq$-minimal element of $\stab_d(P_d)$ (or, if $\stab_d(P_d) = \emptyset$, then set $\stabmatch_d(P_d) = \emptyset$). 
\end{definition}

Note that $\stab_d(P_d)$ is defined solely in terms of $\UnrejGr$ and $P_d$, and does not depend in any other way on the input priorities $Q$ or preferences $P_{-S}$. Note also that $\stab_d(P_d)$ must be contained in some path in $\UnrejGr$, as otherwise we would have some $(d,h), (d,h') \in \stab_d(P_d)$ which are incomparable under $\trianglelefteq$ with both $h\succ_d^{P_d} h'$ and $h'\succ_d^{P_d} h$. This means that $\stabmatch_d(P_d)$ is uniquely defined.

\begin{lemma}
  \label{lem:stabmatch-characterizes-ipda}
  Let $\stabmatch_d(P_d) = (d,h)$ for some $h\in\Insts$.
  For any $d\in S$ and $P_d\in\T_d$, the set of proposals which an applicant $d'\in S\setminus \{d\}$ receives in $\IPDA(d:P_d, d_{S\setminus \{d\}}: \emptyset)$ is exactly 
  $\prop_{d'}\big(\UnrejGr \setminus \chain(\stabmatch_d(P_d))\big)$.
\end{lemma}
\begin{proof}
  Let $m_d(P_d)$ denote $d$'s match in $\IPDA(d : P_d, d_{S\setminus\{d\}} : \emptyset)$ (and note that this definition depends on $Q$ and $P_{-S}$, not just on $\UnrejGr$). First, observe that this run of $\IPDA$ also constitutes one valid run of $\IPDA(d : \{ m_d(P_d) \}, d_{S\setminus\{d\}} : \emptyset)$, by the fact that $\IPDA$ is independent of the order in which proposals are chosen (\autoref{thrm:da-indep-execution}), and by the fact that every proposal to $d$ except for $m_d(P_d)$ is eventually rejected in $\IPDA(d : P_d, d_{S\setminus\{d\}} : \emptyset)$. (In other words, whatever ordering of proposals you like under $(d : P_d, d_{S\setminus\{d\}} : \emptyset)$ also corresponds to some ordering of proposals under $(d : \emptyset, d_{S\setminus\{d\}} : \emptyset)$, possibly delaying future proposals from certain $h$ when they propose to $d$. The same argument is used to prove $\IPDA$ is nonbossy, \autoref{thrm:nonbossy-ipda}.)
  Thus, it suffices to show that $m_d(P_d) = \stabmatch_d(P_d)$, and going forward we can consider $\IPDA(d : \{ m_d(P_d) \}, d_{S\setminus\{d\}} : \emptyset)$.

  By definition of $\UnrejGr$, the set of institutions which $d$ rejects in $\IPDA(d : \{ m_d(P_d) \}, d_{S\setminus\{d\}} : \emptyset)$ is exactly $\prop_d(\UnrejGr \setminus \chain(d,m_d(P_d)))$. By the definition of $\IPDA$, we have that $P_d$ must prefer $m_d(P_d)$ to every institution that they reject. Thus, by the definition of $\stab_d(P_d)$, we must have $m_d(P_d) \in \prop_d(\stab_d(P_d))$. 

  To finish the proof, it suffices to show that no $v = (d,h') \vartriangleleft (d, m_d(P_d))$ satisfies $v \in \stab_d(P_d)$. Suppose for contradiction that this were the case. Then, by the definition of $\chain(v)$, we have that $m_d(P_d)$ will only ever propose to $d$ after $d$ rejects $h'$ in some run of $\IPDA(d : P', d_{S\setminus\{d\}} : \emptyset)$. But, since $P_d$ prefers $h'$ to every institution in $\prop_d(\UnrejGr \setminus \chain(d,h'))$, and institutions in $\chain(d,h')$ can only propose to $d$ after $d$ rejects $h'$, we have that $d$ will never reject $h'$ in $\IPDA$. Thus, $m_d(P_d)$ cannot possibly propose to $d$ in $\IPDA(d:P_d, d_{S\setminus\{d\}})$, a contradiction. Thus, we have $(d,m_d(P_d)) = \stabmatch_d(P_d)$, as desired.
\end{proof}

We are now ready to prove the theorem.

\begin{proof}[Proof (of \autoref{thrm:type-to-menu-DA})] 
  Take $S = \{d_*,d_\dagger\}$ in the definition of $\UnrejGr$. Since $\UnrejGr$ has out-degree at most one and at most one vertex for every pair in $S\times\Insts$, it takes $\widetilde O(n)$ bits to represent. Moreover, by \autoref{lem:stabmatch-characterizes-ipda} and \autoref{thrm:ght-main-positive-DA}, for all $P_{d_*} \in \T_{d_*}$ we have 
  \[ \Menu_{d_\dagger}^{\APDA_Q}(P_{d_*}, P_{-S}) = \prop_{d_\dagger}( \UnrejGr \setminus \chain(\stabmatch_{d_*}(P_{d_*}))).
  \] 
  Thus, the map from $P_{d_*}$ to the menu of $d_\dagger$ can be represented using only $\UnrejGr$, which takes at most $\widetilde O(n)$ bits, as desired.

  Note also that $\Omega(n)$ bits are certainly required, since any possible subset of $\Insts$ could be in $d_\dagger$'s menu (even without taking into account $d_*$'s list).
\end{proof}

\subsection{Applications and Relation to \texorpdfstring{\cite{GonczarowskiHT22}}{GHT}}
\label{sec:application-SEDA-1-to-1}

We now explore an application of the combinatorial structure we uncovered in \autoref{sec:type-to-menu-DA}, namely, the un-rejection graph $\UnrejGr$. 
We consider the notion of a ``pairwise menu'', the natural extension of the notion of the menu to a pair of applicants. 
Specifically, we characterize the set of institutions $(h_*,h_\dagger)$ to which a pair of applicants $(d_*,d_\dagger)$ might match to, holding $P_{-\{d_*,d_\dagger\}}$ fixed.
A simple corollary of \autoref{thrm:type-to-menu-TTC} shows that this set of pairs requires $\Omega(n^2)$ bits to represent for $\TTC$, but \autoref{thrm:type-to-menu-DA} implies this requires $\widetilde O(n)$ bits for $\DA$. 
This next theorem shows something even stronger: that this set of pairs has a simple and natural characterization in terms of $\UnrejGr$.

\begin{theorem}
  \label{thrm:pairwise-menu}
  Fix any priorities $Q$, applicants $S = \{d_*, d_\dagger\}$, and preferences and $P_{-S} = P_{- \{d_*,d_\dagger\}}$. Let $\UnrejGr$ be defined with respect to $S$ and $P_{-S}$ as in \autoref{sec:type-to-menu}. Then, for any pair $(h_*, h_\dagger)\in\Insts\times\Insts$, the following are equivalent:
  \begin{enumerate}[label=(\roman*)]
    \itemsep0em
    \item There exists $P_*, P_\dagger$ such that $\mu(d_*)=h_*$ and $\mu(d_\dagger)=h_\dagger$, where $\mu = \APDA_Q(P_*, P_\dagger, P_{-S})$.
    \item $(d_*, h_*)$ and  $(d_\dagger,h_\dagger)$ are nodes in $\UnrejGr$ and are not comparable under $\trianglelefteq$.
  \end{enumerate}
\end{theorem}
\begin{proof}
  ($(ii)\implies(i)$) First, suppose that $(d_*, h_*), (d_\dagger,h_\dagger) \in \UnrejGr$, but are not comparable under $\trianglelefteq$. Consider $P_{d_*} = \{ h_* \}$ and $P_{d_\dagger} = \{ h_\dagger \}$. The results of \autoref{sec:type-to-menu-DA} directly show that for each $i \in \{*,\dagger\}$, we have $\stabmatch_{d_i}(P_i) = h_i$, hence $h_i \in \Menu_{d_i}(P_{S - \{d_i\}}, P_{-S})$ by the fact that $(d_*,h_*)$ and $(d_\dagger,h_\dagger)$ are incomparable, hence $\mu(d_i) = h_i$.

  ($(i)\implies(ii)$) For the second and harder direction, we consider an arbitrary pair $P_{d_*}, P_{d_\dagger}$, and show that whatever $d_*$ and $d_\dagger$ match to in $\APDA$, the corresponding vertices in $\UnrejGr$ cannot be comparable under $\trianglelefteq$.
  To this end, consider any $P_{d_*}, P_{d_\dagger}$ and let $\mu = \APDA(d_* : P_{d_*}, d_\dagger : P_{d_\dagger})$.
  We make the following definitions (where the two rightmost equalities will follow from the strategyproofness of $\APDA$, and the results of \autoref{sec:type-to-menu-DA}):
  \begin{align*}
    h^{\mathsf{unrej}}_{*} 
     &\ \defeq\ %
    \stabmatch_{d_*}(P_{d_*})
    &&\quad\quad\qquad&
    \hmatch_*
      &\ \defeq \ \mu(d_*) = \max_{P_{d_*}} \Big(\prop_{d_*}
    \big(\UnrejGr \setminus \chain(d_\dagger, h^{\mathsf{unrej}}_\dagger) \big)\Big)
    \\
    h^{\mathsf{unrej}}_{\dagger}
     &\ \defeq\ %
    \stabmatch_{d_\dagger}(P_{d_\dagger})
    &&&
    \hmatch_\dagger
      &\ \defeq \ \mu(d_\dagger) = \max_{P_{d_\dagger}} \Big(\prop_{d_\dagger}
    \big(\UnrejGr \setminus \chain(d_*, h^{\mathsf{unrej}}_*) \big)\Big)
  \end{align*}
  For nodes $v,w \in \UnrejGr$, we say that $w$ is \emph{above} $v$ if $v \trianglelefteq w$ (and $w$ is \emph{below} $v$ if $w \trianglelefteq v$).
  In words, $\hmatch_*$ is determined by first removing every node from $\UnrejGr$ which is above $(d_\dagger,\hunrej_\dagger)$, then taking the maximum-ranked $d_*$ node according to $P_{d_*}$; vice-versa holds for $\hmatch_\dagger$.  
  The remainder of the proof proceeds in two cases based on $\hunrej_*$ and $\hunrej_\dagger$.

  \textbf{(Case 1: $\hunrej_*$ and $\hunrej_\dagger$ are incomparable.)} Suppose that neither $\hunrej_* \trianglelefteq \hunrej_\dagger$ nor $\hunrej_\dagger \trianglelefteq \hunrej_*$. By the definition of $\stabmatch_d$, we have that $\hunrej_* \succeq_{d_*}^{P_{d_*}} h$ for all $h \in \prop_{d_*}(\UnrejGr)$ where we do not have $(d_*, \hunrej_*)\trianglelefteq (d_*, h)$.
  On the other hand, compared to $\prop_{d_*}(\UnrejGr)$, the set $\Menu_{d_*}(P_\dagger,P_{-S})$ only removes some subset $\prop_{d_*}(S)$, where $S$ is some (possibly empty) subset of nodes strictly above $(d_*, h)$ according to $\trianglelefteq$. So $(d_*, \hmatch_*)$ will be weakly above $(d_*, \hunrej_*)$ and strictly below any common upper bound of $(d_*,\hunrej_*)$ and $(d_\dagger,\hunrej_\dagger)$, if such an upper bound exists. Dually, $(d_\dagger, \hmatch_\dagger)$ will be weakly above $(d_*, \hunrej_*)$ and strictly below any common upper bound of $(d_*,\hunrej_*)$ and $(d_\dagger,\hunrej_\dagger)$. Thus, $(d_*,\hmatch_*)$ and $(d_\dagger,\hmatch_\dagger)$ cannot be comparable in $\UnrejGr$, as desired. This case is illustrated in \autoref{fig:one-to-one-SEDA-incomparable.pdf}.

\begin{figure}[tbph]
    \centering
  \begin{minipage}[c]{\textwidth}
  \begin{minipage}[c]{0.45\textwidth}
    \centering
    \includegraphics[width=\textwidth]{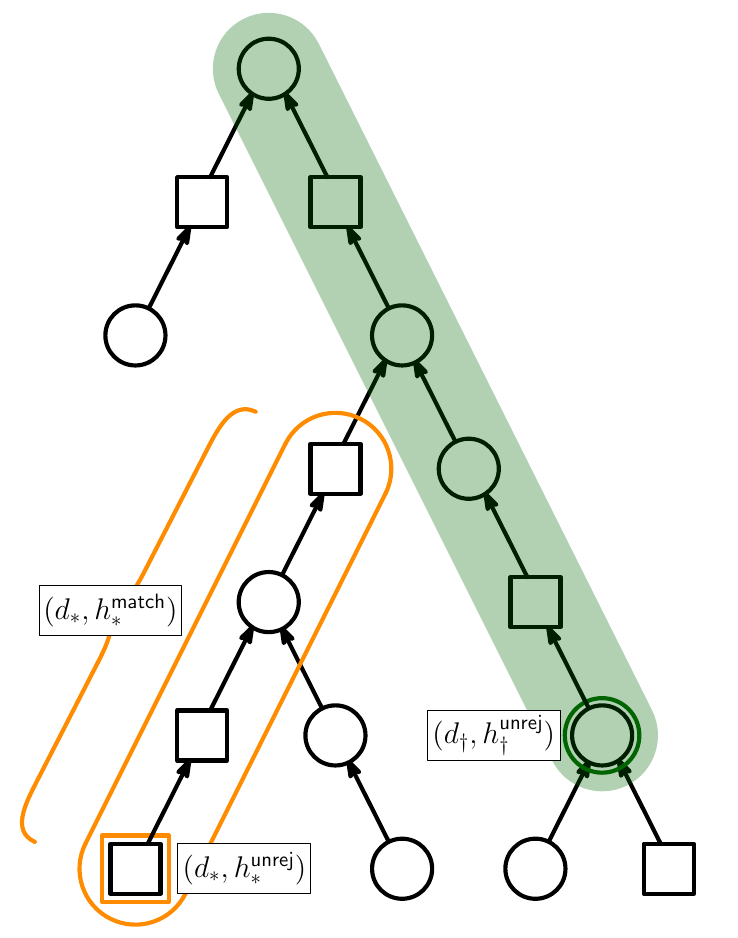}
  \end{minipage}
  \qquad
  \begin{minipage}[c]{0.45\textwidth}
    \centering
    \includegraphics[width=\textwidth]{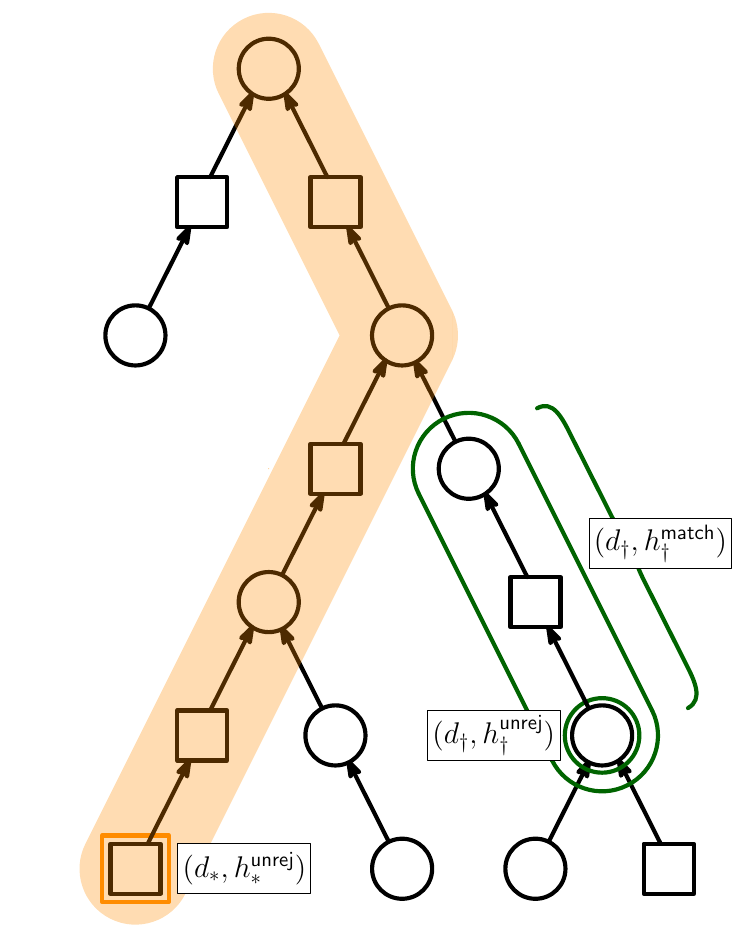}
  \end{minipage}
  \end{minipage}
  \caption{Case 1 of the proof of \autoref{thrm:pairwise-menu}.}
    \label{fig:one-to-one-SEDA-incomparable.pdf}
\end{figure}
\begin{figure}[tbph]
  \begin{minipage}[c]{\textwidth}
  \begin{minipage}[c]{0.45\textwidth}
    \centering
    \includegraphics[width=\textwidth]{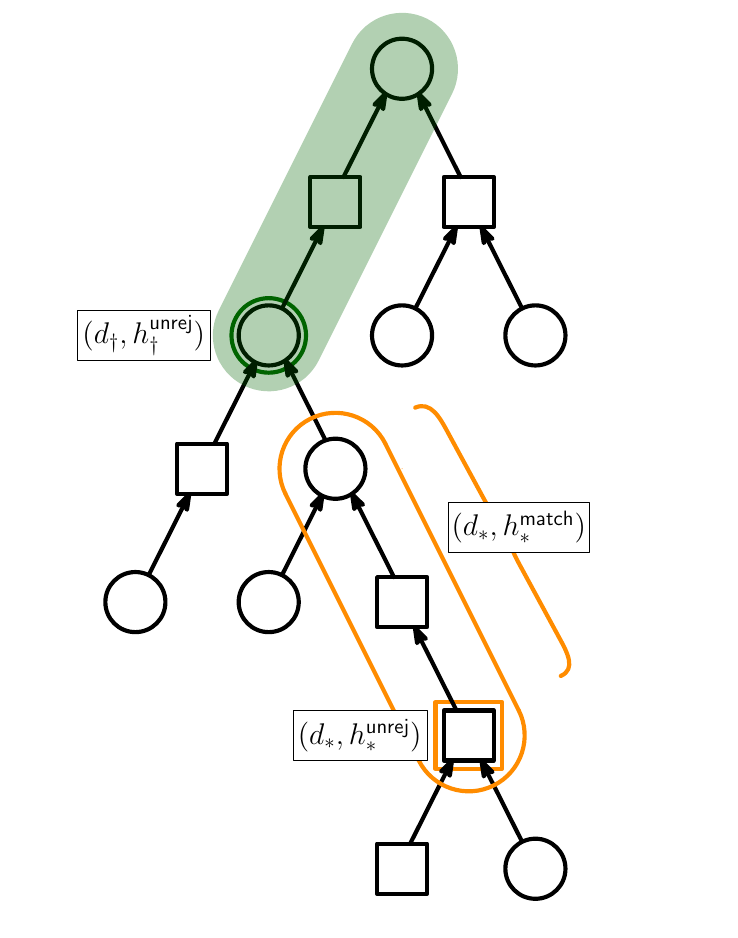}
  \end{minipage}
  \qquad
  \begin{minipage}[c]{0.45\textwidth}
    \centering
    \includegraphics[width=\textwidth]{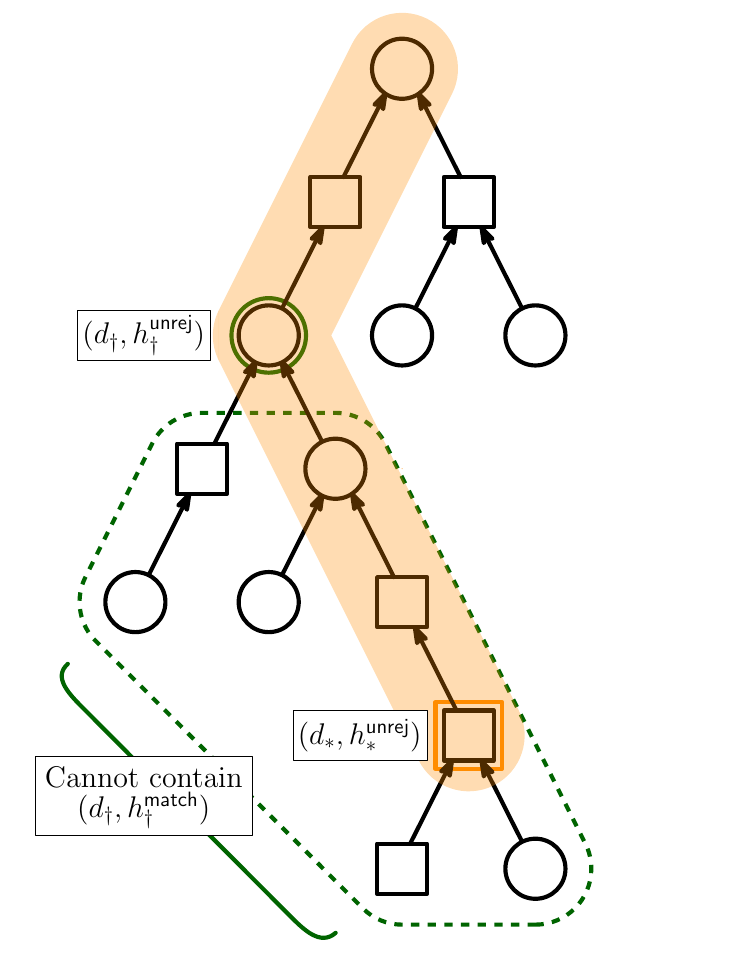}
  \end{minipage}
  \end{minipage}
  \caption{Case 2 of the proof of \autoref{thrm:pairwise-menu}.}
    \label{fig:one-to-one-SEDA-comparable.pdf}
\end{figure}

  \textbf{(Case 2: $\hunrej_*$ and $\hunrej_\dagger$ are comparable.)} Suppose without loss of generality that $\hunrej_* \trianglelefteq \hunrej_\dagger$. By the same logic as in the previous case, we have that $(d_*,\hmatch_*)$ will be weakly above $(d_*,\hunrej_*)$ and strictly below $(d_\dagger,\hunrej_\dagger)$. However, in this case, neither $\hunrej_\dagger$ nor any institution in $\prop_{d_\dagger}(\chain(d_\dagger,\hunrej_\dagger))$ will be in $\Menu_{d_\dagger}(P_*,P_{-S})$.
  Thus, consider the favorite $d_\dagger$ node according to $P_\dagger$ which is outside of $\chain(d_\dagger,\hunrej_\dagger)$, and call the corresponding institution $h^{\mathsf{second}}$. That is:
  \[ h^{\mathsf{second}} = \max_{P_{d_\dagger}}
    \left(\prop_{d_\dagger}
    \big(\UnrejGr \setminus \chain(d_\dagger, \hunrej_\dagger) \big)\right). \]
  Observe that we cannot have $(d_\dagger, h^{\mathsf{second}}) \trianglelefteq (d_\dagger, \hunrej_\dagger)$, as then we would have $\hunrej_\dagger \in \stab_{d_\dagger}(P_\dagger)$, contradicting the definition of $\stabmatch_{d_\dagger}(P_\dagger) = \hunrej_\dagger$. Since every node in $\chain(d_*, \hunrej_*)$ is comparable to $(d_\dagger, \hunrej_\dagger)$ in this case, $(d_\dagger, h^{\mathsf{second}}) \in \UnrejGr \setminus \chain(d_*, \hunrej_*)$, and thus $h^{\mathsf{second}} \in \Menu_{d_\dagger}(P_*, P_{-S})$. Thus, in fact we have $\hmatch_\dagger = h^{\mathsf{second}}$. All told, $(d_*, \hmatch_*)$ will be below $(d_\dagger, \hunrej_\dagger)$, but that $(d_\dagger, \hmatch_\dagger)$ can neither be below nor above $(d_\dagger, \hunrej_\dagger)$. Thus, $(d_*,\hmatch_*)$ and $(d_\dagger,\hmatch_\dagger)$ cannot be comparable in $\UnrejGr$, as desired. This case is illustrated in \autoref{fig:one-to-one-SEDA-comparable.pdf}.
\end{proof}

Now, having proven this characterization, we discuss a connection between the results in  \autoref{sec:type-to-menu-DA} and the work of \cite{GonczarowskiHT22}.
We first briefly recall the relevant results of \cite{GonczarowskiHT22} and their motivation.
First, we rephrase \autoref{thrm:ght-main-positive-DA} in terms of a way to describe (to each applicant separately) their match in $\APDA$.

\begin{definition}[Equivalent to {\cite[Description 1]{GonczarowskiHT22}}]
\label{def:seda}
  For any $d \in \Appls$, define $D_d : \T \to \Insts \cup \{\emptyset\}$ as follows:\footnote{
    \cite{GonczarowskiHT22} phrase their ``Description 1'' in terms of running $\IPDA$ in the market not including $d$ at all. However, it is immediate to see that their definition is equivalent to this one.
  }
  \[ D^Q_d(P) = \max_{P_d} \left\{ h \in \Appls
    \ \big|\ \text{$d$ receives a proposal from $h$ in $\IPDA_Q(d : \emptyset, P_{-d})$} \right\}.
  \]
\end{definition}

Then, by \autoref{thrm:ght-main-positive-DA} and the fact that $\APDA$ is strategyproof, for each $d\in\Appls$, priorities $Q$, and preferences $P$, we have $\APDA_d^Q(P) = D^Q_d(P)$.

\cite{GonczarowskiHT22} are interested in $D^Q_d(\cdot)$ as an alternative way to describe $\APDA$ that might make the strategyproofness of $\APDA$ more clear. 
Indeed, to see that $D^Q_d(\cdot)$ is strategyproof, all $d$ has to observe is that her own report cannot effect the set of proposals she receives in $\IPDA_Q(d : \emptyset, P_{-d})$, and that submitting her true preference ranking $P_d$ always matches her to her highest-ranked obtainable institution.
As discussed in \autoref{sec:gtc-apda}, \cite{GonczarowskiHT22} also prove that (in a perhaps surprising contrast to $\TTC$) traditional descriptions of $\APDA$ cannot suffice to obtain a description whose strategyproofness is easy to see in this way.

However, note that in contrast to traditional descriptions of mechanisms, $D^Q_d$ does not describe the matching of agents other than $d$. In particular, the matching $\IPDA^Q(d:\emptyset, P_{-d})$ does not give the match of agents other than $d$, and the match of another applicant $d'$ is described by $D^Q_{d'}$. This may raise concerns with the description $D^Q_d$: while an agent can easily observe strategyproofness, she cannot easily see that, for example, the matching that results from this description is a feasible (one-to-one) matching, something that is clear under the traditional description.
In fact, \cite{GonczarowskiHT22} go on to formalize a tradeoff between conveying the menu (and hence strategyproofness) or conveying the matching (and hence one-to-one); we defer to \cite[Section 6]{GonczarowskiHT22} for details.

Interestingly, our Theorems~\ref{thrm:type-to-menu-DA} and~\ref{thrm:pairwise-menu} can serve to address this tradeoff to some theoretical degree. Namely, we show next that our theorems provide a direct way for applicants to observe that, if $D^Q_d(\cdot)$ is run separately for each applicant $d$, then no pair of applicants will match to the same institution. In particular, observe that every aspect of the proof of Theorems~\ref{thrm:type-to-menu-DA} and~\ref{thrm:pairwise-menu} could have been directly phrased in terms of $D^Q_{d_*}, D^Q_{d_\dagger}$, since these results exclusively argue about runs of $\IPDA$ of the form given by $D^Q_{d_*}, D^Q_{d_\dagger}$. Thus, the characterization of the menu (and ``pairwise menu'') provided by Theorems~\ref{thrm:type-to-menu-DA} and~\ref{thrm:pairwise-menu} also characterizes the set of institutions appearing in \autoref{def:seda}. This allows us to prove the following:

\begin{corollary}
  For any priorities $Q$ and preferences $P$, and any two applicants $d_*$ and $d_\dagger$, we have $D^Q_{d_*}(P) \ne D^Q_{d_\dagger}(P)$.
\end{corollary}
\begin{proof}
  Consider $\UnrejGr$ defined with respect to $S = \{d_*,d_\dagger\}$ and $P_{-\{d_*,d_\dagger\}}$.
  Observe that for any fixed $h \in \prop_{d_*}\UnrejGr\ \bigcap\ \prop_{d_\dagger}\UnrejGr$, if we have $d_* \succ_h d_\dagger$, then we must have $(d_*,h)\trianglelefteq (d_\dagger,h)$ by the definition of $\trianglelefteq$ and $\chain(\cdot)$.
  \autoref{thrm:pairwise-menu} then shows that $d_*$ and $d_\dagger$ will not simultaneously match to $h$ under any possible $P_{d_*}$ and $P_{d_\dagger}$.
  Thus, $d_*$ and $d_\dagger$ cannot simultaneously match to the same institution according to $D^Q_{d_*}, D^Q_{d_\dagger}$, as desired.
\end{proof}

While these results are complicated, when phrased in terms of $D^Q_{d_*}$ and $D^Q_{d_\dagger}$, they \emph{only} rely on the fact that $\IPDA$ is independent of the order in which proposals are chosen (\autoref{thrm:da-indep-execution}), since the usage of the strategyproofness of $\APDA$ in the proof of \autoref{thrm:pairwise-menu} is replaced with the use of $\max_{P_d}$ directly in the definition of $D^Q_d$.
Thus, the fact that $\{ D^Q_d \}_{d\in\Appls}$ gives a one-to-one matching can be directly proven from the way the description $D^Q_d$ itself operates.
That being said, we cannot imagine how such a proof could be relayed to laypeople in its current form, so this does not diminish the message of \cite{GonczarowskiHT22} in terms of the tradeoff in explainability to real-world participants.

\section{Concurrent Representation Complexity}
\label{sec:compression-complexity}

\subsection{Model}
\label{sec:protocol-models}

We now switch gears, and investigate the structural complexity of communicating the matching (after all reports are known). 
Recall our \nameref{ppg:representation} Question: \questionTextRepresentation{} 
Investigating this questions requires a novel type of protocol in a novel model, which we now discuss.
Recall that the set of applicants is $\Appls$ and the institutions is $\Insts$.
Our model has two major components, which we phrase as two assumptions:
\begin{enumerate}[label=(Assumption \arabic*),ref=Assumption \arabic*,leftmargin=7.5em]
  \itemsep0em
  \item We assume there are many more applicants than institutions, specifically, $|\Appls| \ge |\Insts|^C$ for some fixed, arbitrary $C > 1$.
    \label{item:assumption-imbalanced}
  \item We assume the priority lists of the institutions are a fixed piece of data that defines the mechanism being used, and are thus known to all applicants.
    \label{item:assumption-known-priorities}
\end{enumerate}
See \autoref{fig:model-sizes-etc} for an illustration.
To begin, we discuss why these assumptions are necessary to capture the insights we seek.

\textbf{(\ref{item:assumption-imbalanced})} 
First, we consider a many-to-one matching market with $|\Appls|\gg |\Insts|$. 
This means that each $h \in \Insts$ can match with multiple distinct applicants, up to some positive integer \emph{capacity} $q_h$. See \autoref{sec:additional-prelims} for the standard extension of the definitions of all the matching mechanisms we consider to many-to-one markets (which simply consider $q_h$ identical copies of $h$ for each $h \in \Insts$).

\begin{figure}[tbp]
  \begin{center}
  \includegraphics[width=\textwidth]{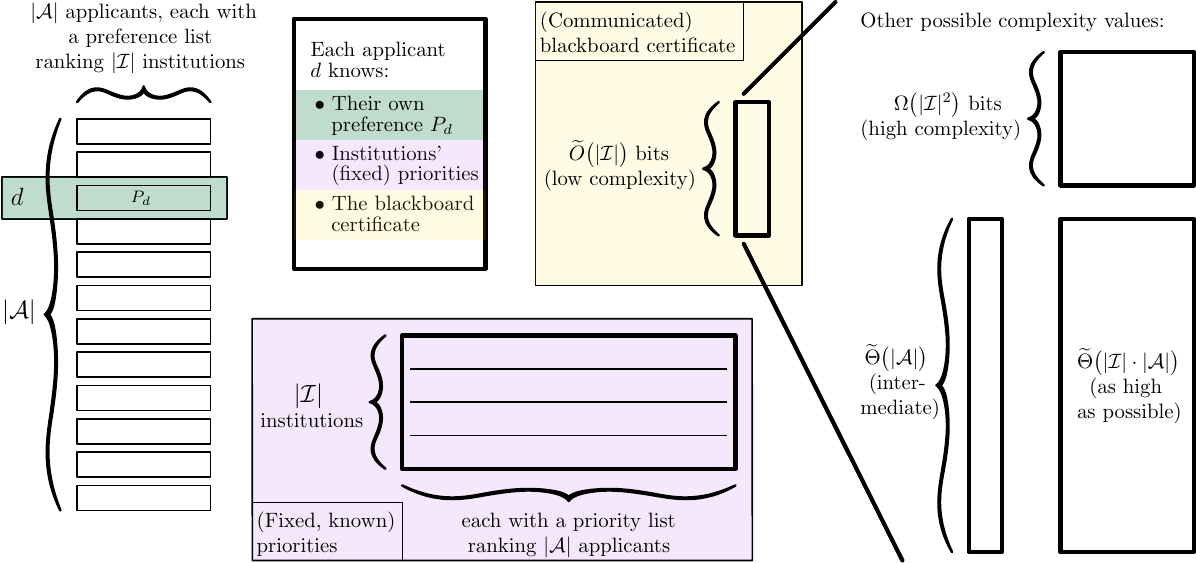}
  \end{center}
  \caption{ Illustration of the different components in our model.   }
  \vspace{0.5em}
    {\footnotesize \textbf{Notes:}
    There is a set of applicants $\Appls$ and a set of institutions $\Insts$, where $|\Appls|\gg |\Insts|$. 
    The complexity of a protocol is the maximum number of bits that need to be written on the blackboard, i.e., the logarithm of the number of certificates the protocol might require.
    We think of an $\widetilde O\bigl(|\Insts|\bigr)$-bit protocol as potentially practical / low complexity, and of an $\Omega\bigl(|\Insts|^2\bigr)$-bit protocol as impractical / high complexity.
    A complexity of $\widetilde \Theta(|\Appls|)$ can be thought of as intermediate: it may be much less or much more than $|\Insts|^2$ depending on the relative values of $|\Appls|$ and $|\Insts|$.
    \par }
  \label{fig:model-sizes-etc}
\end{figure}

While having $|\Appls|\gg|\Insts|$ is very natural, it's somewhat unlikely that real markets like those in school choice will literally have $|\Appls|$ scaling like $|\Insts|^{C}$ asymptotically; for instance, this implicitly assumes that as the number of students increases, more and more students will attend each school.
However, this stylized assumption serves to unambiguously nail down whether a given complexity measure depends on $|\Appls|$ or $|\Insts|$.
For example, if we consider balanced markets with $n = |\Appls| = |\Insts|$, then the questions regarding concurrent representation (\autoref{sec:compression-complexity}) become trivial: an $\widetilde O(n)$-bit representation is achievable simply by writing down the matching, and an $\Omega(n)$ lower bound is not hard to come up with, missing much of the structural ``action.''

\textbf{(\ref{item:assumption-known-priorities})} Second, we assume that the priorities of the institutions are fixed and part of the matching rule $f$ (e.g., $f = \APDA_Q$ with priorities $Q$). This means that each applicant knows ahead of time a $\widetilde\Theta\bigl(|\Appls|\cdot|\Insts|\bigr)$ data structure that contains every institutions' priority rankings over every student.
It is not hard to show that, if the applicants know nothing about the priorities, then the questions of representation (\autoref{sec:compression-complexity}) become trivial, and the optimal protocols for both $\TTC$ and $\DA$ require $\widetilde\Theta\bigl(|\Appls|\bigr)$ bits.\footnote{
For example, consider two institutions, $h^T$ and $h^B$, each with capacity $k$. Suppose there are $2k$ applicants, and every applicant prefers $h^T$ to $h^B$. Then in both mechanisms, the top half of the applicants will be matched to $h^T$. But then, to describe every applicant's match, at least one bit per applicant is required to tell applicants whether they are in the top half or the bottom half of $h^T$'s priority list.
}
Moreover, regarding verification (\autoref{sec:verif-complexity}), \cite{Segal07,GonczarowskiNOR19} prove a $\Omega\bigl(|\Insts|^2\bigr)$ lower bound if the priorities must be communicated.

Assuming that the priorities are common knowledge is a natural stylized version of real markets such as school choice mechanisms, where the priorities are determined by pre-committed policies which are typically publicly available information for each school. Moreover, assuming common knowledge of the priorities only makes our lower bounds stronger. But the number one reason we make this modeling assumption is that it is the most \emph{generic} way to handle the fact that some knowledge of the priorities must be present before the mechanism is run: we do not consider the institutions to be agents actively participating in the protocol or reporting their priorities, so we treat these priorities as fixed.\footnote{
  We also remark that, to keep our model and paper theoretically cohesive, all of the lower bounds in Section~\ref{sec:gtc} and~\ref{sec:type-to-menu} construct a fixed set of priorities where the lower bound holds, so these results also hold for the model where priorities are common knowledge. In contrast, our positive results for the options-effect complexity of $\DA$ in \autoref{sec:type-to-menu-DA} hold even when the priorities are not known in advance by the applicants.
}

Despite the previous paragraph, we note that all of our protocols are also possible under a different modeling assumption. This is the model where priorities are described by ``scores'' $e_h^d \in [0,1]$ and $d \succ_h d'$ if and only if $e_h^d > e_h^{d'}$, as mentioned in \autoref{sec:intro}. All of the protocols we construct can be implemented in the same complexity if each applicant $d$ starts off knowing (and trusting) her scores $\{ e_h^d \}_{h \in \Insts}$ at each institutions.
Thus, each applicant only needs to start off knowing a $\widetilde O(|\Insts|)$ data structure, not an $\Omega(|\Insts|\cdot|\Appls|)$ data structure.
We give the details in \autoref{remark:score-based-model}.

\subsection{Complexity of \texorpdfstring{$\DA$ and $\TTC$}{DA and TTC}}

Having established our general model in \autoref{sec:protocol-models}, we now define our protocols and complexity measures capturing how hard a mechanism is to represent.
In particular, we study the task of posting some blackboard certificate $C$ such that each applicant $d$ can figure out their own match using $C$ (as well as their privately known preferences $P_d$, and for the priority-based mechanisms $\TTC$ and $\DA$, also the publicly known priorities $Q$).

\begin{definition}
  \label{def:representation-protocol}
  A \emph{(concurrent) representation protocol} of a matching mechanism $f$ is a profile of functions $(\match_d)_{d \in \Appls}$, where for each $d$, $\match_d : \T_d \times \Certs \to \Insts$ for some set $\Certs$, with the following property: for every $t=(P_d)_{d\in\Appls}\in\T$, there exists $C \in \Certs$ such that if $\mu = f(t)$, then $\mu(d) = \match_d(P_d, C)$ for each $d\in\Appls$. 
  We call an element $C \in \Certs$ a \emph{certificate}, and if $\mu(d) = \match_d(P_d,C)$ for all $d\in\Appls$, we say that $\mu$ is the matching \emph{induced by} certificate $C$ in the protocol.
 
  The \emph{cost} of a representation protocol is $\log_2|\Certs|$.
  The \emph{(concurrent) representation complexity} of a mechanism $f$ is the minimum of the costs of all concurrent representation protocols for $f$. 
\end{definition}

Observe that there is always a trivial approach to constructing a concurrent representation protocol, by simply unambiguously describing the full matching in the certificate $C \in \Certs$. This amounts to writing the institution each applicant is matched to, using $|\Appls|\log|\Insts| = \widetilde O\bigl(|\Appls|\bigr)$ bits.

For stable matching mechanisms, known results imply that a large savings over this trivial solution is possible in any situation where $|\Appls|\gg|\Insts|$. In fact, the representation used is simply a restatement of the definition of stability, and thus suffices to represent any stable matching (regardless of the mechanism $f$), and uses only $\widetilde O\bigl(|\Insts|\bigr)$ bits. This is somewhat remarkable: the match of each applicant can be simultaneously conveyed to the applicants in less than a single bit per applicant, assuming applicants have complete access to the priorities.
To conveniently state and discuss this result, we state the following definition and lemma:
\begin{definition}
  \label{def:stable-budget-set}
  Given priorities $Q$ and a matching $\mu$, define an applicant $d$'s \emph{stable budget set} as: \begin{align*}
    \StabB^Q_d(\mu) = \big\{ h \in \Insts \ \big|\ 
       & \text{either $|\mu(h)| < q_h$ and $d \succ_h \emptyset$,}
    \\ & \text{or $|\mu(h)|=q_h$ and there is some $d' \in \mu(h)$ such that $d \succ_h d'$} \big\}.
  \end{align*}
\end{definition}
In particular, for one-to-one matching markets, $\StabB^Q_d(\mu)$ is the set of all $h$ such that $d \succ_h \mu(h)$, and \autoref{def:stable-budget-set} is the natural extension of this definition to many-to-one markets.
\begin{observation}
  \label{thrm:representing-DA-cutoff-lemma}
  If $\mu$ is any stable matching, then every applicant $d$ is matched to the institution they rank highest in the set $\StabB^Q_d(\mu)$.%
\footnote{
This observation is not novel. Besides simply restating the definition of stability, this observation, as well as the concurrent representation protocol constructed in \autoref{thrm:representing-DA}, is essentially equivalent to the ``market-clearing cutoffs'' characterization of stable matchings in \cite{AzevedoL16}.
}
\end{observation}
\begin{proof}
  This follows directly from the definition of a stable matching: if $d$ preferred some $h \in \StabB^Q_d(\mu)$ to her match in $\mu$, then $(d,h)$ would block $\mu$ and $\mu$ could not possibly be stable. 
\end{proof}

\begin{observation}
  \label{thrm:representing-DA}
  The concurrent representation complexity of any stable matching mechanism (including $\APDA$ and $\IPDA$) is $\widetilde\Theta\bigl(|\Insts|\bigr)$.
\end{observation}
\begin{proof}
  First, we prove the upper bound. 
  Consider any matching $\mu$ that is stable in some market with priorities $Q$ and preferences $P$. For each institution $h$ such that $|\mu(h)|=q_h$, let $d^{\mathsf{min}}_h$ denote the applicant matched to $h$ in $\mu$ with lowest priority at $h$. If $|\mu(h)|<q_h$, then define $d^{\mathsf{min}}_h = \emptyset$. 
  Now, note that we have $\StabB_d^Q(\mu) = \{ h | d \succ_h d^{\mathsf{min}}_h \}$.
  Thus, a certificate $C \in \Certs$ can simply record, for each institution $h$, the identity of $d^{\mathsf{min}}_h \in \Appls \cup \{\emptyset\}$. This requires $|\Insts|\log|\Appls| = \widetilde\Theta\bigl(|\Insts|\bigr)$ bits. Since each applicant $d$ knows her priority at each institution, she knows the set $\{ h\ |\ d \succ d^{\mathsf{min}}_h \} = \StabB^Q_d(\mu)$. By \autoref{thrm:representing-DA-cutoff-lemma}, her match in $\mu$ is then her highest-ranked institution in this set, so each $d$ can figure out her match $\mu(d)$.

  A matching lower bound is not hard to construct. Consider a market with $n+1$ institutions $h^B, h_1,\ldots,h_n$ such that the capacity of $h^B$ is $n$, and the capacity of each $h_i$ is $1$. Consider applicants $d_1,\ldots,d_n,d_1',\ldots,d_n'$, and for each $i$, let priorities be such that $d_i \succ_{h_i} d_i'$. Consider preference lists such that each $d_i'$ always prefers $h_i \succ_{d_i'} h^B$, and each $d_i$ may independently prefer $h^B \succ h_i$ or $h_i \succ h^B$.
  In any stable matching, $d_i'$ is matched to $h^B$ if and only if $d_i$ ranks $h_i \succ h^B$. Thus, the certificate $C$ must contain at least $n$ bits of information, so $|\Certs|\ge 2^n$ and the concurrent representation complexity of $\Omega(n) = \Omega\bigl(|\Insts|\bigr)$.
\end{proof}

For $\TTC$, the situation is more nuanced. Prior work \cite{LeshnoL21} has shown that it is possible to represent the matching with $O\bigl(|\Insts|^2\bigr)$ ``cutoffs'', one for each pair of institutions. This can potentially save many bits off of the trivial solution that requires $\widetilde\Theta\bigl(|\Appls|\bigr)$ bits, but only when $|\Appls| \gg |\Insts|^2$. In other words, \cite{LeshnoL21} prove an upper bound of $\widetilde O\bigl(\min(|\Appls|,|\Insts|^2)\bigr)$. However, until our work, no lower bounds were known.\footnote{
\cite{LeshnoL21} mention some impossibility results for $\TTC$, but do not formulate or prove any sort of $\Omega\bigl(|\Insts|^2\bigr)$ lower bound. For details of their arguments, see \autoref{sec:related-cutoff-TTC}.
} 

We now provide a lower bound matching the construction of \cite{LeshnoL21}; this is the main result of this section.
This construction is quite reminiscent of those concerning $\TTC$ in Sections~\ref{sec:gtc} and~\ref{sec:type-to-menu}, and we use certain similar ideas.
Intuitively and informally, our construction proceeds by constructing some applicant $d_{i,j}$ corresponding to each \emph{pair} of institutions $h_i, h_j$. 
The preferences and priorities are constructed such that $d_{i,j}$ will decide, based on their high priority at $h_i$, whether some unrelated applicant is able to gain admission to $h_j$.

\begin{theorem}
  \label{thrm:representing-TTC}
  The concurrent representation complexity of $\TTC$ is $\widetilde\Theta\bigl(\min(|\Appls|, |\Insts|^2)\bigr)$.
\end{theorem}
\begin{proof}
  An upper bound of $\widetilde O\bigl(|\Appls|\bigr)$ is trivially, by allowing the certificate posted by the protocol to uniquely define the entire matching. The $\widetilde O\bigl(|\Insts|^2\bigr)$ upper bound is constructed by \cite[Theorem 1]{LeshnoL21}. Very briefly, \cite{LeshnoL21} use the trading and matching process calculating the $\TTC$ outcome to prove that for any $\mu = \TTC_Q(P)$, there exist applicants $\bigl\{ d_{h_1}^{h_2}\bigr\}_{h_1,h_2\in\Insts}$ such that for all applicants $d$, we have $\mu(d) = \max_{P_d} \bigl\{ h_2\in\Insts ~\big|~ \text{$d \succeq_{h_1} d_{h_1}^{h_2}$ for some $h_1\in\Insts$}\bigr\}$. (Informally, this means that each applicant $d$ can use her priority at institution $h_1$ to gain admission to institution $h_2$ whenever $d \succ_{h_1} d_{h_1}^{h_2}$; see \cite{LeshnoL21} for details, which are not needed for establishing our lower bound.) Since the priorities are common knowledge, a concurrent representation protocol can thus communicate to all applicants their match by only communicating $\bigl\{ d_{h_1}^{h_2} \bigr\}_{h_1,h_2\in\Insts}$, which requires $\widetilde O\bigl(|\Insts|^2\bigr)$ bits.

  To finish this proof, it suffices to consider the case where $|\Insts| = n$ and $|\Appls| = n^2$, and prove a lower bound of $\Omega(n^2)$. To see why this case suffices, observe that if $|\Insts|^2 > |\Appls|$, then one can apply this construction with $\sqrt{|\Appls|}$ of the elements of $\Insts$ to get a lower bound of $\Omega\bigl(|\Appls|\bigr)$, and if $|\Insts|^2 < |\Appls|$, one can ignore the surplus elements of $\Appls$ and apply this construction to get a lower bound of $\Omega\bigl(|\Insts|^2\bigr)$.

  Fix $k$, where we will take $n=\Theta(k)$. Consider a matching market with institutions $h^T_1,\ldots,h^T_k$, $h^L_1,\ldots,h^L_k$, and $h^R_1,\ldots,h^R_k$, and applicants $\bigl\{d^T_{i,j}\bigr\}_{i,j\in\{1,\ldots,k\}}$, $\bigl\{d^L_{i,j}\bigr\}_{i,j\in\{1,\ldots,k\}}$, and $\bigl\{d^R_{i,j}\bigr\}_{i,j\in\{1,\ldots,k\}}$.  Mnemonically, these are the top, left, and right participants. The capacity of each of these institutions is $k$. The preferences and priorities are fixed for every participant except the top applicants $\{d^T_{i,j} \}_{i,j\in\{1,\ldots,k\}}$. These fixed preferences and priorities are:
  \newcommand{\Up}{{\mathcal{S}}}
  \begin{align*}
    h^\Up_i 
    & : d^\Up_{i, 1} \succ d^\Up_{i, 2} \succ \ldots \succ d^\Up_{i, k-1} \succ d^\Up_{i, k}
      && \text{For each $\Up \in \{ T, L, R \}$ and $i \in \{1,\ldots,k\}$}
    \\
     d^\Up_{i,j} 
    & : h^T_{j} \succ h^\Up_i
      && \text{For each $\Up \in \{ L, R \}$ and $i, j \in \{1,\ldots,k\}$}
  \end{align*}
   Now, the preferences of the $k^2$ applicants $d^T_{i,j}$ depend on $k^2$ bits $b_{i,j} \in \{L, R\}$, one for each applicant, as follows:
  \begin{align*}
    d^T_{i, j}
    & : h^{b_{i, j}}_{j}
      && \text{For each $i, j \in \{1,\ldots,k\}$}
  \end{align*}
  This family of matching markets is illustrated in \autoref{fig:representing-TTC}.
\begin{figure}[tbp]
    \includegraphics[width=\textwidth]{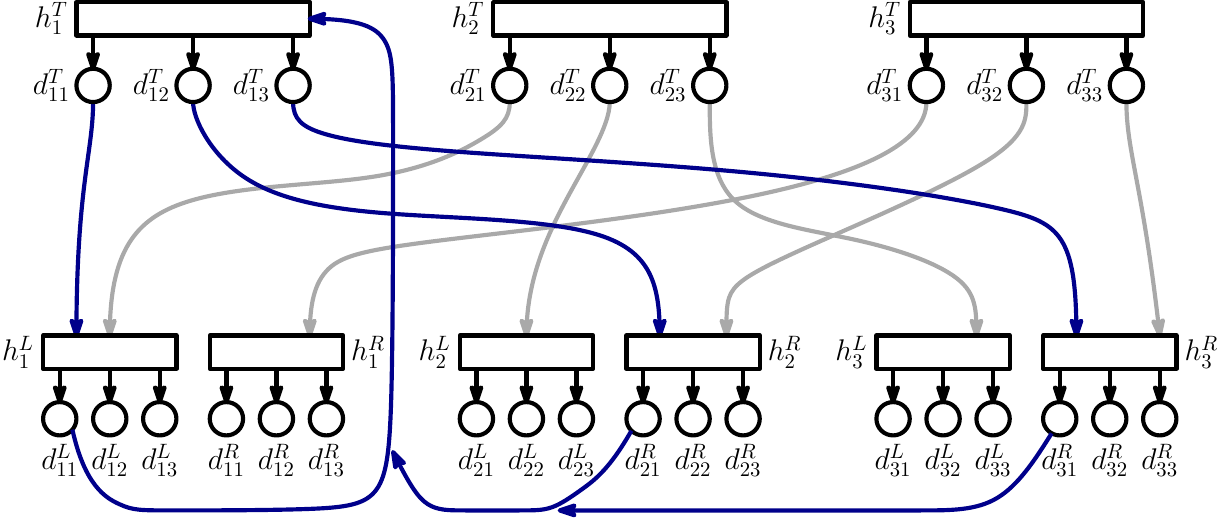}
    \caption[Representing TTC]{
    Illustration of the matching market used in the proof of \autoref{thrm:representing-TTC}. The three cycles matching applicants to $h^T_1$ are colored in dark blue.
    }
    \label{fig:representing-TTC}
\end{figure}

  We will show next that for each $i,j$, the matches of $d^{L}_{j,i}$ and $d^R_{j,i}$ are determined by the bit $b_{i,j}$. Since the preferences of these applicants are fixed (and thus cannot provide any information about $b_{i,j}$ to the function $\mathsf{match}_d$ in the concurrent representation protocol), this means that the certificate $C$ must contain the information about all bits $\vec{b}$.
  \begin{lemma}
    \label{thrm:representing-DA-LB-lemma}
    For each $i, j \in \{1,\ldots,k\}$, if $B = b_{i,j}$, then applicant $d^B_{j, i}$ is matched to $h^T_i$ (and if $\{C\} = \{L, R\} \setminus \{B\}$, then $d^C_{j,i}$ is matched to $h^C_j$).
  \end{lemma}
  First consider each $h^T_1$. First, it will point to $d^T_{1,1}$, who will point to $h^{b_{1,1}}_1$, resulting in a cycle matching $d^{b_{1,1}}_{1,1}$ to $h^T_1$. Next, a cycle with $d^T_{1,2}$ and $d^{b_{1,2}}_{2,1}$ forms, etc., until a cycle with $d^T_{1,k}$ and $d^{b_{1,k}}_{k,1}$ forms; each of these cycles match the involved (non $d^T_{i,j}$) applicant to $h^T_1$. Institution $h^T_1$'s capacity is now filled and each applicant of the form $d^{\Up}_{i, 1}$ for $\Up \in \{L, R\}$ who was not matched in this process then matches to the corresponding $h^{\Up}_i$. 
  Next, $h^T_2$ matches in a similar way to each $d^{b_{2,i}}_{i,2}$ for $i=1,\ldots,k$. 
  This continues for each top hospital until $h^T_k$ in a similar way.
  This proves \autoref{thrm:representing-DA-LB-lemma}.

  Thus, for each distinct profile of bits $\{ b_{i,j} \}_{i,j\in\{1,\ldots,k\}}$, a concurrent representation protocol must present a distinct certificate $C$ (or otherwise, the applicants $d^L_{i,j}$ could not possibly all know their matches). This proves that the concurrent representation complexity of $\TTC$ is $\Omega(n^2)$, concluding the proof.
\end{proof}

\section{Joint Verification Complexity}
\label{sec:verif-complexity}

We now turn to the \nameref{ppg:verification} Question: \questionTextVerification{}

To give exposition, we observe that the concurrent representation protocols in \autoref{sec:compression-complexity} are in some sense reminiscent of nondeterministic protocols: the mechanism designer posts a certificate $C$ that all applicants can see, and we are not concerned with how the mechanism designer found the certificate $C$.
However, the goal of these protocols is different than traditional nondeterministic protocols in computer science: the protocols of \autoref{sec:compression-complexity} only aim to describe a matching to the applicants, not to verify that the matching is correct. 
In this section, we investigate the prospect of additionally verifying the correctness of the matching in the model of \autoref{sec:protocol-models}.

To begin, we illustrate by example why our representations do not suffice to verify the matching. Consider the representation protocol of $\DA$ in \autoref{thrm:representing-DA}. Now, suppose that the mechanism designer deviates and posts cutoffs that are so high that no applicant is matched to any institution. The outcome is certainly not stable. However, no applicant will even be able to tell that some institution is under-subscribed (let alone that the matching is unstable) from only the cutoffs, the priorities, and her own preferences. 

We now enhance \autoref{def:representation-protocol} to investigate the prospect of additionally verifying the matching.

\begin{definition}
  \label{def:verification-protocol}
  A \emph{(joint) verification protocol} of a matching mechanism $f$ is a profile of pairs of functions $\big((\match_d,\check_d)\big)_{d \in \Appls}$, where $\match_d : \T_d \times \Certs \to \Insts$ and $\check_d : \T_d \times \Certs \to \{0,1\}$ for some set $\Certs$, such that:
  \begin{enumerate}
    \item (Representation) $R = (\match_d)_{d\in\Appls}$ is a representation protocol for $f$.
    \item (Completeness) For every $P = (P_d)_{d\in\Appls} \in \T$, there exists $C \in \Certs$ such that $f(P)$ is the matching induced by $C$ over $R$, and $\check_d(P_d,C)=1$ for each $d\in\Appls$.
    \item (Soundness) For every $P = (P_d)_{d\in\Appls} \in \T$ and $\mu \ne f(P)$, if $C \in \Certs$ is such that $\mu$ is the matching induced by $C$ over $R$, then $\check_d(P_d,C)=0$ for at least one $d\in\Appls$.
  \end{enumerate}
      
  The \emph{cost} of a verification protocol is $\log_2|\Certs|$.
  The \emph{verification complexity} of a mechanism $f$ is the minimum of the costs of all verification protocols for $f$. 
\end{definition}

Note that each applicant only needs to know her own match at the end of this protocol, but the applicants collectively verify the certificate in the protocol together. As in standard verification protocols, the trivial solution for a verification problem is to let the certificate $C$ record the full preference lists of all applicants, and let each $\mathsf{check}_d$ verify that the true preference list $P_d$ is correctly recorded in $C$. This trivial solution requires $\widetilde O\bigl(|\Appls|\cdot|\Insts|\bigr)$ bits (in contrast to the concurrent representation protocols of \autoref{sec:compression-complexity}, where the trivial solution requires $\widetilde O\bigl(|\Appls|\bigr)$ bits).

While \autoref{sec:compression-complexity} shows that the matching in $\DA$ (and, when $|\Appls| \gg |\Insts|^2$, in $\TTC$) can be represented with less than one bit per applicant, it turns out this is not possible for verification.
Informally, a nontrivial verification protocol needs to tell applicants information about other applicants, yielding an $\Omega\bigl(|\Appls|\bigr)$ lower bound (in contrast to the $\widetilde\Theta\bigl(|\Insts|\bigr)$- or $\widetilde\Theta\bigl(|\Insts|^2\bigr)$-bit representation protocols of \autoref{sec:compression-complexity},  which only conveyed information about the institutions), since the protocol needs to tell applicants \emph{why} they receive some match, and not just \emph{which} match they receive.
We now formalize this intuition for both $\TTC$ and $\DA$, and even for $\SD$, by reducing the verification problem to a counting problem where each agent holds one bit.

\begin{proposition}
  \label{thrm:verif-lb-both-mechs}
  The verification complexity of $\SD$ is $\Omega\bigl(|\Appls|\bigr)$.
  Thus, the verification complexity of $\TTC$ and of any stable matching mechanism is $\Omega\bigl(|\Appls|\bigr)$.
  These results hold even when $|\Insts|=O(1)$.
\end{proposition}
\begin{proof}
  Fix $k$. We will construct a market with three institutions $\Insts = \bigl\{h^T, h^C, h^B\bigr\}$ and $2k$ applicants $\Appls = \bigl\{ d^T_1,\ldots,d^T_k, d^B_1,\ldots,d^B_k \bigr\}$. The capacity of institution $h^C$ is $k$, and the capacities of each of $h^T$ and $h^B$ are $k+1$. The preferences of the applicants $d^T_1,\ldots,d^T_k$ will be parametrized by bits $\vec b = (b_1,\ldots,b_k) \in \{0,1\}^k$ for some $k$. Later, we will show that the verification complexity is be $\Omega(k) = \Omega\bigl(|\Appls|\bigr)$. The preferences and priorities are:
  \begin{align*}
    & 
    h^{\mathcal{S}} :\ \ d^T_1 \succ\ldots\succ d^T_k \succ d^B_1 \succ\ldots\succ d^B_k 
    && \text{For each $\mathcal{S} \in \{ T, C, B\}$.}
    \\ &
    d^T_j : \begin{cases}
      h^T & \qquad \text{if $b_j = 1$} \\
      h^C & \qquad \text{if $b_j = 0$}
      \end{cases}
    \\ & 
    d^B_j : h^C \succ h^B && \text{For each $j=1,\ldots,k$.}
  \end{align*}

  Since the priorities of each institution are the same, the outcome of $\SD$, $\TTC$, and all stable matching mechanisms are the same under these preferences and priorities. We now characterize this outcome.
  \begin{lemma}
    \label{thrm:verif-lb-both-mechs-outcome-characterization}
    Suppose $f : \{0,1\}^k \to \M$ is defined such that $f(b)$ denotes the result of running $\SD$, $\TTC$, or any stable matching mechanism, with the above preferences and priorities parametrized by bits $b$. 
    Then for any $b$, if we set $S = \sum_{i=1}^k b_i$, then $f(b)$ matches $d^B_1,\ldots,d^B_S$ to $h^C$ and matches $d^B_{S+1},\ldots,d^B_k$ to $h^B$.
  \end{lemma}
  We give the argument for $\SD$. Initially, applicants $d^T_1,\ldots,d^T_k$ match in order, with those $d^T_j$ where $b_j = 0$ matching to $h^C$. Since $h^C$ has capacity exactly $k$, this leaves $S$ seats open at $h^C$, which will go to the $S$ highest-priority applicants among $d^B_1,\ldots,d^B_k$. 
  This proves \autoref{thrm:verif-lb-both-mechs-outcome-characterization}.
  
  Now, consider a verification protocol $P$ for any relevant mechanism $f$ (i.e., $\TTC$, or any stable matching mechanism) with the above preferences and priorities. Since the matches of applicants $d^B_1,\ldots,d^B_k$ are determined by $S = \sum_{i=1}^k b_i$, yet the preference lists of applicants $d^B_1,\ldots,d^B_k$ are fixed and do not depend on $b_1,\ldots,b_n$, the certificate $C \in \Certs$ used by the verification protocol must unambiguously determine $S$. 
  Denote the value of $S$ corresponding to a certificate $C$ by $S(C)$.
  Now, the verification protocol must in particular verify that $S(C)$ is the value $\sum_{i=1}^k b_i$, so the verification protocol in particular suffices to define a $k$-player verification protocol for $\sum_{i=1}^k b_i$, where each player $i$ holds the single bit $b_i$.\footnote{
  Moreover, since the certificate $C$ must convey the value $S(C)$, we can ignore the fact that in \autoref{def:verification-protocol} we only require each applicant to learn their own outcome, and we can assume that in the verification protocols for bit-counting which we consider, all agents must learn the value $\sum_{i=1}^k b_i$.
  }
  Call this the \emph{bit counting problem}.

  Given the above, to prove the theorem, it suffice to show that the verification complexity of the bit counting problem is $\Omega(k)$.
  Informally, this will follow because no protocol can perform this verification without (up to lower-order terms) specifying in $C$ precisely which players $i$ have $b_i = 1$; specifying this takes $\Omega(k)$ bits.

  We now formalize this with a straightforward rectangle argument. 
  For convenience, suppose that $k$ is a multiple of $2$.
  Let $\mathcal{F} \subseteq \{0,1\}^k$ denote the set of inputs $\vec{b} = (b_1,\ldots,b_n)$ such that $\sum_{i=1}^k b_i = \nicefrac{k}{2}$. We will show that in any verification protocol for bit counting, we have $|\Certs| \ge |\mathcal{F}|$.
  To prove this, suppose for contradiction that $|\Certs| < |\mathcal{F}|$. Observe that two distinct $\vec{b}, \vec{b}' \in \mathcal{F}$ must then correspond to the same certificate $C \in \Certs$. Now, let $j \in \{1,\ldots,k\}$ be any index such that ${b}_j\ne{b}'_j$, and define $\vec{b}''$ such that ${b}''_j = {b}'_j$ and ${b}''_i = {b}_i$ for each $i\ne j$. Then, we have $\check_i(b_i'', C) = 1$ for each $i \in \{1,\ldots,k\}$, since each value of $b_i''$ is equal to either $b_i$ or $b_i'$.
  But we also have that $\mathsf{outcome}(C) = \sum_{i=1}^k b_i \ne \sum_{i=1}^k b_i''$, so the verification protocol must violate soundness. This is a contradiction, showing that $|\Certs|\ge|\mathcal{F}|$.

  Thus, the verification complexity of bit-counting is at least $\log|\mathcal{F}| = \log{\binom{k}{k/2}} = \Omega(k)$.
  This proves the joint verification complexity of $\SD$ (and hence $\DA$ and all stable matching mechanisms) is $\Omega(k) = \Omega\bigl(|\Appls|\bigr)$.
\end{proof}

\subsection{\texorpdfstring{$\TTC$}{TTC}}

We now consider upper bounds on verification complexity; unlike the lower bound in \autoref{thrm:verif-lb-both-mechs}, constructing a protocol is a very different task in $\TTC$ versus in $\DA$.
First, we show that \autoref{thrm:verif-lb-both-mechs} is tight for $\TTC$, because $\TTC$ has an easy to construct \emph{deterministic} protocol with matching communication complexity.
\begin{observation}
  \label{thrm:verif-complexity-ttc}
  The verification complexity of $\TTC$ is $\widetilde\Theta\bigl(|\Appls|\bigr)$.
\end{observation}
\begin{proof}
  The lower bound is given in \autoref{thrm:verif-lb-both-mechs}. The upper bound follows from the fact that $\TTC$ has a deterministic blackboard communication protocol. We describe this simple (possibly folklore) protocol for completeness.

  The protocol proceeds in steps, with one applicant or institution acting per step. First, the protocol selects an arbitrary institution $h_1$, which announces its top choice applicant $d_1$ and points to it. Then, $d_1$ announces her top institution $h_2$ and points to it. Then, $h_2$ announces and points, etc. This continues, with the protocol keeping track of some tentative chain $h_1\to d_1\to h_2\to d_2\to \cdots \to h_k \to d_k$, until some cycle is found, say $d_k \to h_j$ for some $j \in \{1,2,\ldots,k\}$. Every applicant in the cycle is permanently matched to the institution to which it points, and the applicants are removed from the market (and the capacities of these institutions are each reduced by $1$). Then, if $j > 1$, this process resumes, starting from $d_{j-1}$ pointing to her top ranked institution with remaining capacity.

  If each applicant had only pointed once during this protocol, our desired communication bound would have followed immediately. While applicants may point more than once, an applicant points a second time \emph{at most once for each cycle that is matched}. Since each cycle permanently matches at least one applicant, there are at most $|\Appls|$ matched cycles overall, and thus there are at most $O\bigl(|\Appls|\bigr)$ occurrences where an applicant is called to point, and thus the protocol uses $\widetilde O\bigl(|\Appls|\bigr)$ bits of communication overall.
\end{proof}

We remark that the deterministic communication protocol in the above proof works even in a model where the priorities are not prior knowledge (and must instead be communicated by the institutions).

\subsection{\texorpdfstring{$\DA$}{DA}}
\label{sec:verif-complexity-DA}

Giving an upper bound on verification complexity is far less straightforward for stable matching mechanisms than for $\TTC$. 
To begin, note that verifying that a matching $\mu$ is the outcome of $\APDA$ (or $\IPDA$) is a different and more complicated task than verifying that $\mu$ is stable.
Indeed, to verify stability using $\widetilde\Theta\bigl(|\Appls|\bigr)$ bits, the protocol simply needs to use the certificate $C$ to write down the matching itself: since the priorities of the institutions are known to all applicants, each applicant $d$ can simply check that there is no institution $h$ such that $h \succ_d \mu(h)$ and such that $h$ prefers $d$ to one of its matches. 
However, this verification protocol does not suffice to verify that the outcome is the result of $\APDA$, because if we consider any market with more than one stable matching, the above protocol will also incorrectly verify any stable matching that is not the result of $\APDA$.

Nevertheless, we prove that with only $\widetilde O\bigl(|\Appls|\bigr)$ bits, a somewhat intricate protocol can not only describe the outcome of $\APDA$ (or $\IPDA$), but also verify that the outcome is correct. Our protocol relies on classically known properties of the set of stable matchings (namely, the rotation poset / lattice properties, see \cite{GusfieldStableStructureAlgs89}) to verify that a proposed matching is the extremal (i.e., either applicant- or institution-optimal) matching within the set of stable matchings, all while using only a few bits per applicant. 

\begin{theorem}[restate=restateThrmVerifComplexityDa, name=]
  \label{thrm:verif-complexity-da}
  \label{thrm:deterministic-cc-ttc}
  The verification complexity of both $\APDA$ and $\IPDA$ is $\widetilde\Theta\big(|\Appls|\big)$.
\end{theorem}

We prove this theorem through the remainder of this section. To begin, note that the lower bound is given in \autoref{thrm:verif-lb-both-mechs}.

For the upper bound, it suffices to consider balanced markets with $n = |\Appls| = |\Insts|$, and prove that the verification complexity in this market is $\widetilde O(n)$. To see why, observe that the outcome of $\APDA$ (or $\IPDA$) in any many-to-one market with $|\Appls|\ge|\Insts|$ corresponds to the outcome in the one-to-one market where every institution $h \in \Insts$ in the original market is replaced by institutions $h^1,h^2,\ldots,h^{q_h}$ (each with priority list identical to $h$'s, and where $h$ is replaced on each applicant's preference list by $h^1\succ h^2\succ\ldots\succ h^{q_h}$). Thus, a verification protocol with cost $\widetilde O\bigl(|\Appls|\bigr)$ for many-to-one markets can be defined by running a verification protocol on the corresponding one-to-one market, and for the rest of this section we focus on one-to-one markets. Thus, for the rest of this section, we assume that $n = |\Appls| = |\Insts|$.

We now define the main concept needed for our verification protocols, the ``improvement graphs.''  This definition is closely related to notions from \cite{GusfieldStableStructureAlgs89}, the standard reference work on the lattice properties of the set of stable matchings, and we will crucially use their results in our proof of correctness below.

\begin{definition}
  \label{def:improvement-graphs}
  Consider any one-to-one matching market with priorities $Q$ and preferences $P$, and consider any matching $\mu$.

  First, define a directed graph $\ImprGrI = \ImprGrI(Q,P,\mu)$. The vertices of $\ImprGrI$ are all applicants $d \in \Appls$. The edges of $\ImprGrI$ are all those pairs $(d, d')$ such that there exists $h\in\Insts$ such that $h\ne\emptyset$ is the highest institution on $d$'s preference list such that $d \succ_h \mu(h)$, and moreover $\mu(h) = d'$. For such a $d, d', h$, we denote this edge by $d \xrightarrow{h} d'$.

  Second, define a directed graph $\ImprGrA = \ImprGrA(Q,P,\mu)$ with precisely the same definition as $\ImprGrI$, except interchanging the roles of applicants and institutions. 
\end{definition}

Observe that the vertices of each of $\ImprGrI$ and $\ImprGrA$ have out-degree at most $1$, and thus these graphs take only $\widetilde O(n)$ bits to describe.

Our key lemma will show next that a matching equals $\IPDA$ if and only if $\ImprGrI$ is acyclic.
To begin to gain intuition for this, observe that if $\mu$ is stable and $d \xrightarrow{h} d'$ is an edge in $\ImprGrI$, then $h$ must be below $\mu(d)$ on $d$'s preference list.
Informally, the lemma will follow because edges in $\ImprGrI$ represent possible ways matches could be exchanged that make the matching better for the applicants and worse for the institutions, and cycles in $\ImprGrI$ characterize such changes to the matching that could be accomplished without violating stability. 
\begin{lemma}
  \label{thrm:improvement-graphs-show-opt}
  Fix any set of preferences $P$ and priorities $Q$. For any matching $\mu$, we have $\mu = \IPDA_Q(P)$ if and only if both $\mu$ is stable with respect to $P$ and $Q$ and the graph $\ImprGrI(Q,P,\mu)$ is acyclic.
  Dually, $\mu = \APDA_Q(P)$ if and only if both $\mu$ is stable and $\ImprGrA(Q,P,\mu)$ is acyclic.
\end{lemma}
\begin{proof}
  We prove the first result, namely, that $\IPDA$ is characterized by $\ImprGrI$. The proof relies heavily on results from \cite[Section 2.5.1]{GusfieldStableStructureAlgs89}. 
  First, observe that in the language of \cite{GusfieldStableStructureAlgs89}, the edges $d\xrightarrow{h} d'$ of $\ImprGrI$ are exactly those $d,h,d'$ such that $h = s_\mu(d)$ and $d' = \mathit{next}_{\mu}(d)$.%
  \footnote{%
    \label{footnote:GI-outdegree-error}%
    We remark that \cite[Section 2.5.1, Page 87]{GusfieldStableStructureAlgs89} mentions the following: ``Note that $s_\mu(d)$ might not exist. For example, if $\mu$ is the [institution]-optimal [stable] matching, then $s_\mu(d)$ does not exist for any [applicant] $d$.'' The second sentence of this quote would seem to eliminate the need for our verification protocol to communicate any graph whatsoever, since if $\mu = \IPDA$, then this graph should be empty. However, the claim made in passing in that second sentence is false. For example, if 
    $h_1 : d_1\succ d_2$, 
    $h_2 : d_2\succ d_1$, 
    $h_a : d_1\succ d_2 \succ d_3$, and
    $d_1 : h_2 \succ h_a \succ h_1$,
    $d_2 : h_1 \succ h_a \succ h_2$,
    $d_3 : h_a$,
    then $\mu = \IPDA = \{ (d_1,h_2), (d_2,h_1), (d_3,h_a) \}$, and $s_\mu(d_1) = s_\mu(d_2) = h_a \ne \emptyset$. (Note that $s_\mu(d_3) = \emptyset$, and thus our graph $\ImprGrI$ has edges $d_1 \to d_3 \leftarrow d_2$, and is indeed acyclic.)
  }
  
  First, we show that if $\mu$ is stable and $\ImprGrI$ has a cycle, then $\mu \ne \IPDA = \IPDA_Q(P)$. Suppose $d_1 \xrightarrow{h_1} d_2 \xrightarrow{h_2} d_3 \xrightarrow{h_3} \ldots  \xrightarrow{h_{k-1}} d_k \xrightarrow{h_k} d_1$ is some cycle in $\ImprGrA$. Then, according to the definitions given in \cite[Page 88]{GusfieldStableStructureAlgs89}, we have that $\rho = [ (d_1,h_{k}), (d_2,h_1), \ldots, (d_k, h_{k-1}) ]$ is a rotation exposed in $\mu$. Define a matching $\mu'$ such that $\mu'(d_i) = h_{i}$ for each $i=1,\ldots,k$ (with indices mod $k$), and $\mu'(d) = \mu(d)$ for all $d$ not contained in the above cycle (in the language of \cite{GusfieldStableStructureAlgs89}, we have $\mu' = \mu / \rho$, the elimination of rotation $\rho$ from $\mu$). \cite[Lemma 2.5.2]{GusfieldStableStructureAlgs89} shows that $\mu' \ne \mu$ is also stable, and is preferred by all institutions to $\mu$. Thus, by the fact that $\IPDA$ is the institution-optimal stable outcome, we have $\mu \ne \IPDA$.

  Second, we show that if $\mu$ is stable and $\mu \ne \IPDA_Q(P)$, then $\ImprGrI$ must have a cycle. By \cite[Lemma 2.5.3]{GusfieldStableStructureAlgs89}, there is some rotation $\rho = [ (d_1,h_{k}), (d_2,h_1), \ldots, (d_k, h_{k-1}) ]$ that is exposed in $\mu$. By the definition of $\ImprGrI$ and the definition of a rotation from \cite{GusfieldStableStructureAlgs89}, this rotation corresponds to a cycle $d_1 \xrightarrow{h_1} d_2 \xrightarrow{h_2} d_3 \xrightarrow{h_3} \ldots  \xrightarrow{h_{k-1}} d_k \xrightarrow{h_k} d_1$ in $\ImprGrI$.%
  \footnote{
    In fact, our $\ImprGrI$ is defined in a way very close to the graph $H(\mu)$ which \cite{GusfieldStableStructureAlgs89} consider in the proof of their Lemma 2.5.3; our graph $\ImprGrI$ simply considers the entire set of applicants, instead of just those who have a different partner in $\mu$ and $\IPDA \ne \mu$, and our graph generalizes theirs to arbitrary matchings $\mu$, possibly including $\IPDA$ (which fixes the minor error mentioned in \autoref{footnote:GI-outdegree-error}).
    }
  So $\ImprGrI$ is cyclic.

  Since $\ImprGrA$ is defined precisely by interchanging the roles of applicants and institutions, we have that $\ImprGrA$ characterizes $\APDA$ in the same way.
\end{proof}

Now, we can define a verification protocol $V^\IPDA$ (resp.\ $V^\APDA$) for $\IPDA$ (resp.\ $\APDA$). 
When priorities are $Q$ (recall that the priorities $Q$ are known to each applicant before the mechanism begins\footnote{
Several steps of this protocol (both checking stability and checking the correctness of the graph $\ImprGrI$ or $\ImprGrA$) would require $\Omega(n^2)$ bits if priorities are not known to the applicants, for instance, when an applicant $d$ is required to check that for some $d_x, d_y$, and $h$, it's not the case that $d_x \succ_h d \succ_h d_y$ (cf.\ \cite{GonczarowskiNOR19}).
}) and preferences are $P$, the certificate $C$ records the matching $\mu = \IPDA_Q(P)$ (resp.\ $= \APDA_Q(P)$) and graph $G = \ImprGrI$ (resp.\ $= \ImprGrA$). By \autoref{thrm:improvement-graphs-show-opt}, it suffices for the applicants to collectively verify that the matching $\mu$ is stable, and that the graph $G$ is correctly constructed.
As already noted in \autoref{sec:verif-complexity},
it is not hard to verify that $\mu$ is stable once $\mu$ is known: each applicant $d$ can identify using $\mu$ and $Q$ which institutions $h$ are such that $d \succ_h \mu(h)$ according to $Q$, and each applicant can simply check that there are no such $h$ where $h \succ_d \mu(d)$. Thus, to complete the proof, it suffices to show the following lemma:

\begin{lemma}
  \label{lem:exist-predicates-checking-graph}
  If priorities $Q$ are fixed ahead of time, there exists a predicate $\check_d(P_d,\mu,G) \in \{0,1\}$ for each $d\in\Appls$, such that for all preferences $P$, matchings $\mu$, and graphs $G$, we have:
  \[ \bigwedge_{d \in \Appls} \check_d(P_d,\mu,G) = 1
  \quad \iff \quad
  G = \ImprGrI(Q,P,\mu).\]
  Moreover, the same claim holds (with a different predicate $\check_d(\cdot)$) replacing $\ImprGrI$ with $\ImprGrA$.
\end{lemma}
\begin{proof}
  We first prove the claim for $\ImprGrI$. This case is not too hard to see: each $d$ just needs to verify that the edge outgoing from her node in $G$ exactly corresponds to $d\xrightarrow{h} d'$, where $d'=\mu(h)$ and $h$ is $d$'s highest-ranked institution such that $d \succ_h \mu(h)$ (if it exists). Since $d$ knows the priorities, $d$ also knows the set of all $h$ such that $d \succ_h \mu(h)$, and can thus perform this verification. 

  We now prove the claim for $\ImprGrA$, which requires a somewhat more delicate predicate $\check_d(\mu,G)$. For clarity, note that the full definition of $\ImprGrA$ is as follows:
  There is a vertex for each $h \in \Insts$, and an edge $h \xrightarrow{d} h'$ whenever $h,h',d$ are such that both $\mu(d)=h'$ and $d \ne \emptyset$ is the highest-ranked applicant on $h$'s preference list such that $h \succ_d \mu(d)$. We claim that it suffices for each applicant $d$'s predicate $\check_d(\mu,G)$ to verify all of the following:
  \begin{itemize}
    \itemsep0em
    \item For each edge $h_x \xrightarrow{d} h_y$, we have $h_x \succ_d h_y = \mu(d)$.
    \item For each edge $h_x \xrightarrow{d_y} h_y$ with $d_y \ne d$: if we let $d_x = \mu(h_x)$, then whenever $d_x \succ_{h_x} d \succ_{h_x} d_y$ according to the priorities $Q$, we have $\mu(d) \succ_d h_x$.
    \item For each $h$ with out-degree zero in $G$, we have $\mu(d)\succ_d h$.
  \end{itemize}
  The first condition is necessary for each $d$ by the definition of $\ImprGrA$. 
  If the first and second condition are both true for every $d\in\Appls$, then every edge in $G$ is correctly constructed according to $\ImprGrA$, because if these conditions hold then each $d_y$ in some edge $h_x\xrightarrow{d_y} h_y$ is the highest-ranked applicant below $d_x=\mu(h_x)$ such that $h_x \succ_{d_y} \mu(d_y)$. Finally, the third condition for every $d \in \Appls$ guarantees that an institution has out-degree zero in $G$ if and only if it has out-degree zero according to $\ImprGrA$.
  This shows that $G = \ImprGrA$ if and only if each of the above three predicates is true for each $d \in \Appls$, as desired. Since the above conditions can be verified for each $d$ using only knowledge of the priorities, $\mu$, $G$, and $d$'s own preferences, this proves the lemma.
\end{proof}

We can now finish our proof:
\begin{proof}[Proof of \autoref{thrm:verif-complexity-da}]
  By \autoref{lem:exist-predicates-checking-graph}, the protocols $V^\IPDA$ and $V^\APDA$ correctly verify the outcome of $\IPDA$ and $\APDA$, respectively. Thus, they can communicate a certificate containing the matching $\mu$ and corresponding graph $G$. This certificate requires only $\widetilde O(n) = \widetilde O\bigl(|\Appls|\bigr)$ bits, and applicants check that $\mu$ is stable and that the graph is correctly constructed as detailed above.
\end{proof}

Note that unlike $\TTC$, we do not know the \emph{deterministic} blackboard communication complexity of $\DA$ in our model (where priorities are prior knowledge). While \cite{Segal07, GonczarowskiNOR19} prove a lower bound of $\Omega\bigl(|\Insts|^2\bigr)$ when the priorities must be communicated, this model is provably different than our model, because the lower bounds of \cite{Segal07, GonczarowskiNOR19} hold for the verification problem as well.
We believe that determining the deterministic communication complexity of $\DA$ in this model is an interesting and highly natural problem for future work.

\bibliographystyle{alpha}
\bibliography{MasterBib}{}

\clearpage

\appendix 

\section{Additional Results}
\label{sec:additional-results}
\label{sec:additional-outcome}

In this appendix, we include a number of additional results for completeness.
First, we recall our four main questions from \autoref{sec:intro}:
\begin{enumerate}[label=(\arabic*),ref=(\arabic*)]
  \itemsep0em
  \item \questionTextTypeToMatching{}\\ (\nameref{ppg:type-to-match} Question / \autoref{sec:gtc})
  \item \questionTextTypeToMenu{}\\ (\nameref{ppg:type-to-menu} Question / \autoref{sec:type-to-menu})
  \item \questionTextRepresentation{}\\ (\nameref{ppg:representation} Question / \autoref{sec:compression-complexity})
  \item \questionTextVerification{}\\ (\nameref{ppg:verification} Question / \autoref{sec:verif-complexity})
\end{enumerate}

\autoref{tab:appendix-results} gives an exposition of the first three supplemental questions we address here, and relates them to our four main questions.

\begin{table}[pthb]
  \caption[Results]{Summary of our supplemental results comparing $\TTC$ and $\DA$.}
  \label{tab:appendix-results}
  \begin{center}
    \begin{tabular}{C{16em}cC{10em}C{10em}}
  \toprule
    && $\TTC$ & $\DA$ \\
    \cmidrule{3-4}
    \\[-0.5em]
    \vspace{-0.3em}
    \makecell{Describing simultaneously\\all applicants' menus\\[-0.2em]{\!\!\!\footnotesize (harder version of \nameref{ppg:representation}\!\!\!}\\[-0.2em]{\footnotesize Question / \autoref{sec:bit-complexity-menus})}}
    &&  \ResultEntry{$\widetilde\Theta(n^2)$ }{ By \autoref{thrm:all-menus-TTC}.}
      &  \ResultEntry{{$\widetilde\Theta(n^2)$} }{ By \autoref{thrm:all-menus-DA}. }
    \\ \\
    \vspace{-0.3em}
    \makecell{Describing one's effect\\ on one's/another's match\\[-0.2em]{\footnotesize (easier version of \nameref{ppg:type-to-match}}\\[-0.2em]{\footnotesize Question / \autoref{sec:type-to-anothers-match})}}
    && \ResultEntry{ $\widetilde\Theta\bigl(n\bigr)$ 
      }{ By nonbossiness, \\ see \autoref{thrm:one-match-TTC} }
    & \ResultEntry{ $\widetilde\Theta\bigl(n\bigr)$ 
    }{ Corollary of other results, \\ see \autoref{thrm:one-match-DA}. }
    \\ \\
    \vspace{0in}
    \vspace{-0.3em}
    \multirowcell{2}[0pt][c]{Describing all applicants'\\ effects on one's match \\{\footnotesize (combination of \nameref{ppg:type-to-match}}\\[-0.2em]{\footnotesize \!\!\!\!\!\! and \nameref{ppg:representation} Questions / \autoref{sec:all-type-to-one-match})\!\!\!\!\!\!}}
    &&
    \multicolumn{2}{c}{
      \pdfliteral{ 1 0  0 1 0 -9 cm} 
      {$\widetilde\Theta(n^2)$}
      \pdfliteral{ 1 0  0 1 0 9 cm}
    }
    \\
    &&
    \multicolumn{2}{c}{ 
    {\footnotesize Even for $\SD$, by \autoref{thrm:all-type-to-one-match-sd}. }
    }
    \\
    \\
    \\ \bottomrule
  \end{tabular}
  \end{center}
\end{table}

We include two additional subsections in this appendix. 
In \autoref{sec:serial-dict-rot}, we study the options-effect complexity of $\SDrot$ (which turns out to be $\widetilde O(n)$, showing that $\SDrot$ could not have helped us to lower-bound this complexity for $\TTC$ as in \autoref{sec:type-to-menu-TTC}).
In \autoref{sec:additional-remarks} we make some additional observations concerning our verification protocol for $\DA$ from \autoref{sec:verif-complexity-DA}.

\subsection{All-Menus Complexity}
\label{sec:bit-complexity-menus}

In this section, we build upon \autoref{sec:compression-complexity}, and discuss a different approach to representing the process and results of strategyproof matching mechanisms. 
This approach is inspired by the representation protocols for $\DA$ and $\TTC$ in \autoref{sec:compression-complexity}. 
Both of these protocols (implicitly) communicate some set $B_d$ concurrently to each $d\in\Appls$, and each $d$ determines her match as her top-ranked institution in $B_d$. This may seem to imply that $B_d$ should be $d$'s set of available options, i.e., $d$'s menu. However, this is not the case, as we will see below.
However, it inspires one reasonable approach to representing the matching: communicate all applicants' menus, and tell each applicant she is matched to her highest-ranked institution in her menu.

To capture the number of bits required to communicate all applicants' menus, we make the following definition:

\begin{definition}
  \label{def:all-menus-rep}
  The \emph{all-menus complexity} of a matching mechanism $f$ over a set of applicants $\Appls = \{1,\ldots,n\}$ is:
  \[
  \log_2 \left| \left\{
  \bigl(\ \Menu^{f}_{1}(P_{-1}), \Menu^{f}_{2}(P_{-2}), \ldots, \Menu^{f}_{n}(P_{-n})\ \bigr) \ \middle|\ P \in \T
  \right\}\right|.
  \]
\end{definition}
Note that, in contrast to the models in \autoref{def:representation-protocol} and \autoref{def:verification-protocol}, this definition does not assume that an applicant $d$ uses information about her own report $P_d$ to figure out $\Menu_d(P_{-d})$. However, using such information could not possibly help applicant $d$ in this task, since $\Menu_d(P_{-d})$ is by definition independent of the value of $P_d$. Thus, \autoref{def:all-menus-rep} captures the complexity of representing to each applicant her own menu under the model of \autoref{sec:protocol-models} as well.
Also, recall that in our model, if $f$ is $\TTC_Q$ or $\DA_Q$ for some priorities $Q$, then we consider the priorities fixed and part of the mechanism. So, when we bound the all-menus complexity of $\TTC$ (resp.\ $\DA$), we mean the maximum over all possible $Q$ of the all-menus complexity of $\TTC_Q$ (resp.\ $\DA_Q$). 
For cohesion with \autoref{sec:compression-complexity}, we convey our results in terms of $|\Appls|$ and $|\Insts|$ (though for readability we also consider the case where $n = |\Appls| = |\Insts|$).

We now consider $\APDA$, and discuss in detail how the sets $\StabB_d\bigl(\APDA_Q(P)\bigr)$ (the stable budget sets, which are implicitly communicated by the representation protocol in \autoref{thrm:representing-DA}) and $\Menu_d^{\APDA_Q}(P_{-d})$ differ. For perhaps the most crucial high-level difference, note that $d$'s own preference $P_d$ can influence $\StabB_d\bigl(\APDA_Q(P)\bigr)$.
For a concrete example of how these sets can differ, consider the following example (taken from the related work section of \cite{GonczarowskiHT22}):
\begin{example}[\cite{GonczarowskiHT22}]
  Let institutions $h_1, h_2, h_3,h_4$ all have capacity $1$, and consider applicants $d_1,d_2,d_3,d_4$. Let the priorities and preferences be:
  \begin{align*}
      & h_1: d_1 \succ d_2
        && d_1: h_1 \succ \ldots
   \\ & h_2: d_4 \succ d_3 \succ d_2 \succ d_1
        && d_2: h_1 \succ h_2 \succ h_4 \succ \ldots
   \\ & h_3: d_3
        && d_3: h_3 \succ \ldots
   \\ & h_4: d_2 \succ d_4 
        && d_4: h_4 \succ h_2 \succ \ldots
  \end{align*}
  Then $\APDA_Q(P)$ pairs $h_i$ to $d_i$ for each $i=1,\ldots,4$. Now, for institution $h_2$, consider which applicants $d$ have $h_2 \in \StabB_d\bigl(\APDA_Q(P)\bigr)$, and which have $h_2 \in \Menu_d^{\APDA_Q}(P_{-d})$.
  First, $h_2$ is in the stable budget set of applicants $d_2$, $d_3$, and $d_4$,
    so despite $d_3$ being higher priority than $d_2$ at $h_2$, $h_2$ is \emph{not} on $d_3$'s menu.
  Second, $h_2$ is in the menu of applicants $d_1$, $d_2$, and $d_4$,
    so despite $d_1$ being lower priority than $d_2$ at $h_2$, $h_2$ \emph{is} on $d_1$'s menu.
\end{example}

More generally, for $\APDA$ the menu differs from the stable budget set in two ways. First, if $h$ is an institution who would accept a proposal from $d$, but a rejection cycle would lead to $d$ being kicked back out, then $h \in \StabB_d(\APDA) \setminus \Menu^{\APDA}_d$. Second, if $h$ is an institution such that $h$ received a proposal from $\mu(h)$ only as a result of $d$ proposing to $\mu(d)$ (and moreover, we have $\mu(h) \succ_h d \succ_h d'$, where $d'$ is the match of $h$ if $d$ submits an empty list), then it is possible that $h \in \Menu_d^{\APDA} \setminus \StabB_d(\APDA)$. 
In other words, calculating the menu of an applicant $d$ must take into account both the fact that $d$ might hypothetically propose to some $h \ne \mu(d)$, and the fact that $d$ might no longer propose to $\mu(d)$.
Our main result in this section harnesses this intuition to show that the all-menus complexity of $\DA$ is high:

\begin{theorem}
  \label{thrm:all-menus-DA}
  In a one-to-one market with $n$ applicants and $n$ institutions, the all-menus complexity of any stable matching mechanism is $\widetilde\Theta(n^2)$.
  In a many-to-one market, the all-menus complexity of any stable matching mechanism is $\Omega\bigl(|\Insts|^2\bigr)$.
\end{theorem}

\newcommand{\dAttempt}{d^{T}}
\newcommand{\hAttempt}{h^{T}}
\newcommand{\dQuestion}{d^{B}}
\newcommand{\hQuestion}{h^{B}}
\newcommand{\dRev}{d^{R}}
\newcommand{\hRev}{h^{R}}

\begin{proof}
  Let $n=3k$.  First we define the priorities $Q$.  For each $i=1,\ldots,k$, there are three institutions $\hAttempt_i, \hQuestion_i, \hRev_i$, and applicants $\dAttempt_i, \dQuestion_i,\dRev_i$. The priorities of the institutions are, for each $i=1,\ldots,k$:
  \begin{align*}
      \hAttempt_i : \dAttempt_i \succ 
        \dQuestion_1 \succ \ldots \succ \dQuestion_k
      &\qquad&
      \hQuestion_i : \dRev_i \succ
        \dAttempt_1 \succ \ldots \succ \dAttempt_k \succ \dQuestion_i
      &\qquad&
      \hRev_i : \dQuestion_i \succ \dRev_i
  \end{align*}
  The preferences of the applicants are parametrized by a family of subsets $B_i \subseteq \{ \hAttempt_1, \ldots, \hAttempt_k \}$, with one subset for each $i=1,\ldots,k$.  We define $P = P(B_1,\ldots,B_k)$ as follows:
  \begin{align*}
      \dAttempt_i : \hAttempt_i 
      &\qquad&
      \dQuestion_i : \hQuestion_i \succ B_i \succ \hRev_i
      &\qquad&
      \dRev_i : \hRev_i \succ \hQuestion_i  
  \end{align*}
  (Where the elements of $B_i$ can appear in $\dQuestion_i$'s list in any order.) These preferences and priorities are illustrated in \autoref{fig:all-menus-DA}.

  The key claim is the following:
  \begin{lemma}
    \label{lem:all-menus-da-characterization}
    Let $P = P(B_1,\ldots,B_k)$. Then $\hQuestion_j \in \Menu_{\dAttempt_i}^{DA_{Q}}\big(P_{-\dAttempt_i}\big)$ if and only if $\hAttempt_i \in B_j$.
  \end{lemma}
  To prove this lemma, consider changing $\dAttempt_i$'s list to $\bigr\{ \hQuestion_j \bigr\}$ to get a preference profile $P'$.  This will cause $\dQuestion_j$ to be rejected from $\hQuestion_j$ and start proposing to each institution in $B_j \subseteq \bigl\{\hAttempt_1,\ldots,\hAttempt_k\bigr\}$. Such an institution $\hAttempt_u$ will accept the proposal from $\dQuestion_j$ if and only if they have not received a proposal from $\dAttempt_u$.  But among $\bigl\{\dAttempt_1, \ldots, \dAttempt_k\bigr\}$, only $\dAttempt_i$ does not propose to $\hAttempt_i$ in $P'$.  So $\dQuestion_j$ will propose in the end to $\hRev_j$ if and only if $\hAttempt_i \notin B_j$.  If $\dQuestion_j$ proposes to $\hRev_j$, then $\dRev_j$ will next propose to $\hQuestion_j$, so $\hQuestion_j \notin \Menu_{\dAttempt_i}(P_{-\dAttempt_i})$.  If $\dQuestion_j$ is instead accepted by $\hAttempt_i$, then we have $\hQuestion_j \in \Menu_{\dAttempt_i}(P_{-\dAttempt_i})$, as desired. 
  This proves \autoref{lem:all-menus-da-characterization}.

  Thus, there is a distinct profile of menus $\big( \Menu_{\dAttempt_1}, \ldots, \Menu_{\dAttempt_k} \big)$ for each distinct profile of $(B_1,\ldots,B_k)$, of which there are $2^{k^2}$ , and the all-menus complexity of $\APDA$  is $\Omega(k^2) = {\Omega\bigl(|\Insts|^2\bigr)}$.
  Moreover, since by \autoref{thrm:same-menu-all-stable}, each applicant's menu is the same in all stable mechanisms, the same bound holds for any stable matching mechanism.
\end{proof}

\begin{figure}[tbp]
  \begin{minipage}[c]{0.6\textwidth}
    \includegraphics[width=\textwidth]{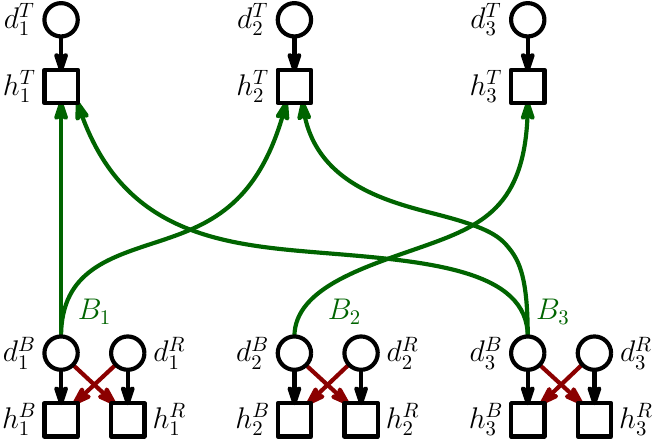}
  \end{minipage}
  \qquad
  \begin{minipage}[c]{0.35\textwidth}
    \caption[All DA Menus Together Lower Bound Construction]{
    Illustration of the set of preferences and priorities used to show that the all-menus complexity of $\DA$ is $\Omega(n^2)$ (\autoref{thrm:all-menus-DA}).
    }
    {\footnotesize \textbf{Notes:} 
    An applicant $d^T_i$ has institution $h^B_j$ on her menu if and only if $h^T_i \in B_j$, i.e. if and only if $d^B_j$ would propose to $h^T_i$ if she is rejected from $h^B_j$. 
    \par}
    \label{fig:all-menus-DA}
  \end{minipage}
\end{figure}

We remark that, if the priorities $Q$ were not known in advance by the applicants, then the above construction could have been much simpler. For example, we could simply let $h^B_i$ accept or not accept a proposal from each $d^T_j$ independently, such that $h^B_i$ is in the menu of $d^T_j$ if and only if $h^B_i$ would accept their proposal.\footnote{
  In contrast, for our lower bounds in Section~\ref{sec:gtc} and~\ref{sec:type-to-menu}, we do not know of any significantly simpler constructions that get the same bounds by allowing the priorities $Q$ to vary. 
}

It turns out that the same complexity result holds for $\TTC$, namely, the all-menu complexity of $\TTC$ is $\Omega\bigl(|\Insts|^2\bigr)$.  This is perhaps less surprising than our result for $\DA$. For instance, consider the case where $|\Appls|\gg|\Insts|^2$, and recall that in this case, our result \autoref{thrm:representing-TTC} shows that the representation complexity of $\TTC$ is $\Omega\bigl(|\Insts|^2\bigr)$. Moreover, for any strategyproof mechanism $f$, the all-menus complexity must be at least as high as the representation complexity (since for a strategyproof mechanism, any representation of the menus of each applicant suffices to describe to each applicant their own match). Thus, when $|\Appls|\gg|\Insts|^2$, we already know that the all-menus complexity of $\TTC$ is $\Omega\bigl(|\Insts|^2\bigr)$. 

Our next result shows that the all-menus complexity of $\TTC$ is $\Omega\bigl(|\Insts|^2\bigr)$ (even in the balanced case where $|\Insts|=|\Appls|$). The construction is similar to that of \autoref{thrm:all-menus-DA}, except the way that $\TTC$ works allows us to make the construction quite a bit simpler.
\begin{theorem}
  \label{thrm:all-menus-TTC}
  In a one-to-one market with $n$ applicants and $n$ institutions, the all-menus complexity of $\TTC$ is $\widetilde\Theta(n^2)$.
  In a many-to-one market, the all-menus complexity of any stable matching mechanism is $\Omega\bigl(|\Insts|^2\bigr)$.
\end{theorem}
\begin{proof}
  Let $n = 2k$.  First we define the priorities $Q$.  For each $i=1,\ldots,k$, there are institutions $\hAttempt_i, \hQuestion_i$, and applicants $\dAttempt_i, \dQuestion_i$. The priorities of the institutions are, for each $i=1,\ldots,k$:
  \begin{align*}
      \hAttempt_i : \dAttempt_i &\qquad& \hQuestion_i : \dQuestion_i
  \end{align*}
  The preferences are parametrized by a family of subsets of $B_i \subseteq \bigl\{ \hAttempt_1, \ldots, \hAttempt_k \bigr\}$ for each $i=1,\ldots,k$.  We define $P(B_1,\ldots,B_k)$ as follows: for each $i=1,\ldots,k$:
  \begin{align*}
      \dAttempt_i : \hAttempt_i &\qquad& \dQuestion_i : B_i \succ \hQuestion_i
  \end{align*}
  (Where the elements of $B_i$ can appear in $\dQuestion_i$'s list in any order.)
  The key claim is the following:
  \begin{lemma}
    \label{lem:all-menus-ttc-characterization}
    Let $P = P(B_1,\ldots,B_k)$. Then
    $\hQuestion_j \in \Menu_{\dAttempt_i}^{TTC_{Q}}\big(P\big)$ 
    if and only if $\hAttempt_i \in B_j$.
  \end{lemma}
  To prove this lemma, consider changing $\dAttempt_i$'s list to $\bigl\{ \hQuestion_j \bigr\}$ to get a preference profile $P'$. Consider now the run of $\TTC$. After all $\dAttempt_j$ for $j\ne i$
    have been matched to the corresponding $\hAttempt_j$, 
    observe that $\dQuestion_j$ points transitively to $\dAttempt_i$ if
    and only if $\hAttempt_i \in B_j$. The former is equivalent to $\hQuestion_j$
    being on $\dAttempt_i$'s menu.
    This proves \autoref{lem:all-menus-ttc-characterization}.

    Thus, there is a distinct profile of menus $\big( \Menu_{\dAttempt_1}, \ldots, \Menu_{\dAttempt_k} \big)$ for each distinct profile of $(B_1,\ldots,B_k)$, of which there are $2^{k^2}$, so the all-menus complexity of $\TTC$ is $\Omega(k^2) = \Omega\bigl(|\Insts|^2\bigr)$.
\end{proof}

\subsection{Type-to-Another's-Match Complexity}
\label{sec:type-to-anothers-match}

As exposited in \autoref{sec:intro}, the outcome-effect complexity gives a way to measure ``how much'' (or more precisely, ``in how complex of a fashion'') one applicant can affect the matching.
We measure this complexity via the number of bits it takes to represent the function from one applicant's preference $P_d$ to some other piece of data.
An easier version of this question could be to consider the complexity of one applicant's effect on a \emph{single} applicant's match, rather than the full outcome matching.
We separately consider the effect of one's report on another's match, and on one's own match, as follows:

\begin{definition}
  \label{def:type-to-one-match}
  The \emph{type-to-another's-match complexity} of a matching mechanism $f$ is
  \[ \log_2\ \max_{d_*, d_\dagger\in\Appls} \left|\left\{ 
    f_{d_\dagger}(\cdot, P_{-d_*})\ \big|\ P_{-d_*}\in\T_{-d_*}
    \right\}\right|, \]
    where $f_{d_\dagger}(\cdot, P_{-d_*}) : \T_{d_*} \to \M$ is the function mapping each $P_{d_*}\in\T_{d_*}$ to the match of applicant $d_\dagger\ne d_*$ in $f(P_{d_*}, P_{-d_*})$.

  The \emph{type-to-one's-own-match} complexity is as in the above definition, except taking $d_\dagger = d_*$.
\end{definition}

This is a special case of our other complexity measures: If one knows the entire matching, one knows any individual match; moreover, if the mechanism is strategyproof and one knows another applicant's menu, then one can derive what they will match to. Thus:
\begin{observation}
  \label{thrm:ttam-less-than-ttm}
  For any matching mechanism, the type-to-one's-own-match complexity and the type-to-another's-match complexity are at most the outcome-effect complexity.
  Additionally, for any strategyproof matching mechanism, the type-to-another's-match complexity is at most the options-effect complexity.
\end{observation}

For simplicity, in this section we consider the case where $n = |\Appls|=|\Insts|$.

\autoref{def:type-to-one-match} can be related to the classical notions of strategyproofness and nonbossiness.
For any strategyproof mechanism, an applicant is always matched to her top-ranked institution on her menu, by \autoref{thrm:TaxationPrinciple}. Thus, writing down $d_*$'s menu suffices to describe $d_*$'s match under any possible report, showing that the type-to-one's-own-match complexity is $\widetilde O(n)$. If the mechanism is additionally nonbossy, there are at most $n$ matchings $\mu$ that can result from $d_*$ submitting any preference list $P_{d_*}$; by writing down both the menu and the value of $\mu(d_\dagger)$ for each resulting matching $\mu$, one additionally knows the function $f_{d_\dagger}(\cdot,P_{d_*})$ for $d_\dagger \ne d_*$. This shows that the type-to-another's-match complexity of a strategyproof and nonbossy mechanism is $\widetilde O(n)$. 
Recalling that $\TTC$ and $\APDA$ are strategyproof, and $\TTC$ is nonbossy, we get:

\begin{observation}
  \label{thrm:one-match-TTC}
  The type-to-one's-own-match complexity of $\TTC$ and $\APDA$ is $\widetilde \Theta(n)$.
  The type-to-another's-match complexity of $\TTC$ is $\widetilde \Theta(n)$.
\end{observation}

In contrast, note that the type-to-another's-match complexity of $\APDA$ is not immediately clear from first principles: since $\APDA$ is bossy, $d_*$'s menu does not completely determine the mapping from $d_*$'s report to $d_\dagger$'s match.
Note also that for the non-strategyproof mechanism $\IPDA$, it is not immediately clear how $d_*$'s report $P_{d_*}$ determines $d_*$'s own match, let alone $d_\dagger$'s.
Nonetheless, we can harness our characterization of the options-effect complexity of $\DA$ from \autoref{sec:type-to-menu-DA} to now bound these complexity measures:

\begin{corollary}
  \label{thrm:one-match-DA}
  The type-to-one's-own-match complexity of $\IPDA$ is $\widetilde\Theta(n)$.
  The type-to-another's-match complexity of both $\APDA$ and $\IPDA$ is $\widetilde\Theta(n)$.
\end{corollary}
\begin{proof}
  First, consider the type-to-one's-own-match complexity of $\IPDA$. \autoref{lem:stabmatch-characterizes-ipda} shows that if we define $\UnrejGr$ with $S = \{ d_* \}$, then under preference $P_{d_*}$, applicant $d_*$ will match to $\stabmatch_{d_*}(P_{d_*})$ in $\IPDA(d_*: P_{d_*})$. Thus, $\UnrejGr$ with $S = \{d_*\}$ suffices to represent $f_{d_*}(\cdot,P_{-d_*})$.

  Now, the type-to-another's-match complexity of $\APDA$ follows directly from \autoref{thrm:type-to-menu-DA} and the fact that $\APDA$ is strategyproof. For $\IPDA$, consider $\UnrejGr$ with $S = \{ d_*, d_\dagger \}$, and observe that the $d_\dagger$-nodes in $\UnrejGr \setminus \chain(\stabmatch_{d_*}(P_{d_*}))$ will also correspond to $\UnrejGr$ defined with $S = \{d_\dagger\}$ (and $d_*$ submitting list $P_{d_*}$), with the same ordering according to $\trianglelefteq$.
  Thus, \autoref{lem:stabmatch-characterizes-ipda} shows that $\stabmatch_{d_\dagger}(P_{d_\dagger})$ in this restricted instance of $\UnrejGr$ suffices to give the match of $d_\dagger$ in $\IPDA(d_* : P_{d_*}, d_\dagger : P_{d_\dagger})$, as desired.
\end{proof}

\subsection{All-Type-To-One-Match Complexity}
\label{sec:all-type-to-one-match}

Our Section~\ref{sec:gtc} and~\ref{sec:type-to-menu} ask how one applicant's report can affect a mechanism, and our
Section~\ref{sec:compression-complexity} and~\ref{sec:verif-complexity} ask how some information (or verification task) can be conveyed to all applicants simultaneously.
For completeness, we next investigate a complexity measure that unites these two ideas.
Specifically, we take the ``easiest / lowest'' type of complexity measure that considers how one applicant can affect the mechanism (namely, the type-to-another's-match complexity from \autoref{sec:type-to-anothers-match}), and investigate it under the agenda where information must be conveyed about all applicants simultaneously. Our definition is:

\begin{definition}
  \label{def:all-type-to-one-match}
  The \emph{all-type-to-one-match} complexity of a matching mechanism $f$ is
  \[
  \log_2\ \max_{d_\dagger} \left| \left\{
  \bigl(\ {f}^1_{d_\dagger}(\cdot, P_{-1}), {f}^2_{d_\dagger}(\cdot, P_{-2}), \ldots, {f}^n_{d_\dagger}(\cdot,P_{-n})\ \bigr) \ \middle|\ P \in \T
  \right\}\right|,
  \]
  where for each $d\in\Appls = \{1,\ldots,n\}$, the function $f_{d_\dagger}^d(\cdot,P_{-d}) : \T_d \to \Insts$ is such that $f_{d_\dagger}^d(P_d',P_{-d})$ is the match of applicant $d_\dagger$ in $f(P_d', P_d)$.
\end{definition}

While we mostly consider this complexity measure for completeness, one interesting aspect of this measure is that it is high even for serial dictatorship, as we show in the following theorem.
In particular, a direct corollary of our next theorem is that the all-type-to-one-match complexity of $\TTC$ and $\DA$ are $\Omega(n^2)$, and furthermore, that more complicated variants of \autoref{def:all-type-to-one-match} (such as the complexity of representing every applicant's map from their type to the matching overall, i.e., combining \autoref{def:all-type-to-one-match} with \autoref{def:type-to-matching}) also yield $\Omega(n^2)$ complexity for any of these mechanisms.

\begin{theorem}
  \label{thrm:all-type-to-one-match-sd}
  In a one-to-one market with $n$ applicants and $n$ institutions, the \emph{all-type-to-one-match complexity} of serial dictatorship is $\widetilde\Theta(n^2)$.
\end{theorem}
\begin{proof}
  Fix $k$, where we will take $n = \Theta(k)$.
  Consider applicants $d^T_1,\ldots,d^T_k,d^B_1,\ldots,d^B_k,d_\dagger$, and let this order over these applicants be the (single) priority order in the serial dictatorship mechanism.
  Consider also institutions $h^T_1,\ldots,h^T_k,h^B_1,\ldots,h^B_k,h_\dagger$.
  Now, for any profiles of sets $B_1,\ldots,B_k \subseteq \bigl\{h^T_1,\ldots,h^T_k\bigr\}$, define a set of preferences $P$ such that:
  \begin{align*}
    & 
    d^T_i : h^T_i && \text{For each $i=1,\ldots,k$}
    \\ &
    d^B_i : h^B_i \succ B_i \succ h_\dagger && \text{For each $i=1,\ldots,k$}
    \\ &
    d_\dagger : h_\dagger
  \end{align*}
  The key claim is the following:
  \begin{lemma}
    \label{thrm:all-type-to-one-match-sd-lemma}
    If the preference list of $d^T_i$ is changed to $\bigl\{ h^B_j \bigr\}$, then applicant $d_\dagger$ will match to $h_\dagger$ if and only if $h^T_i \in B_j$.
  \end{lemma}
  To prove this lemma, observe that if applicant $d^T_i$ switches their preference to only rank $h^B_j$, then these two participants will permanently match. After all applicants in $\bigl\{d^T_1,\ldots,d^T_k\bigr\}$ choose, only $h^T_i$ will still be unmatched. Then, when $d^B_j$ is called to pick an institution, she will pick $h_\dagger$ if and only if $h^T_i \notin B_j$. This proves \autoref{thrm:all-type-to-one-match-sd-lemma}.

  Thus, there is a distinct profile of functions $\big(f^{d^T_1}_{d_\dagger}(\cdot, P_{-d^T_1}),\ldots,f^{d^T_k}_{d_\dagger}(\cdot, P_{-d^T_k})\big)$ for each distinct profile $(B_1,\ldots,B_k)$, of which there are $2^{k^2}$, so the all-type-to-one-match complexity of Serial Dictatorship is $\Omega(k^2) = \Omega(n^2) = \Omega\bigl(|\Insts|^2\bigr)$.
\end{proof}


\subsection{Options-Effect Complexity of \texorpdfstring{$\SDrot$}{SDrot}}
\label{sec:serial-dict-rot}

We now give an additional result concerning $\SDrot$ (as defined in \autoref{def:SDrot}). While \autoref{thrm:gtc-SDrot-LB} shows that this mechanism has high outcome-effect complexity, it turns out to have low options-effect complexity. (In this way, the complexity measures of $\SDrot$ are similar to the complexity measures of $\DA$, though interestingly, we only know how to embed $\SDrot$ into $\TTC$ but not into $\DA$.) This gives some formal sense in which our bounds on the outcome-effect complexity of $\TTC$ and the options-effect complexity of $\TTC$, which are both $\Omega(n^2)$, must hold for different reasons (or more precisely, it shows why the proof approach of \autoref{thrm:gtc-TTC-LB} will not suffice to prove \autoref{thrm:type-to-menu-TTC}).

\begin{theorem}
  \label{thrm:type-to-menu-SDrot}
  The options-effect complexity of $\SDrot$ is $\widetilde O(n)$.
\end{theorem}
\begin{proof}
  For some fixed $d_*, d_\dagger$, we bound the number of bits required to represent the function $g(P_{d_*}) = \Menu_{d_\dagger}^{\SDrot}(P_{d_*},P_{-\{d_*, d_\dagger\}})$. Recall that in $\SDrot$, applicant $d_0$ picks some institution $h^{\mathsf{rot}}_j$, and the other applicants are matched among $\{h_1,\ldots,h_n\}$ according to $\SD_{d_j,d_{j+1},\ldots,d_n}(\cdot)$.
  First, note that if we consider any $d_* \ne d_0$, the question can only be as hard as for $\SD$, and if we consider any $d_\dagger = d_i$ for $i < n$, then the applicants $d_{i'}$ for $i'>i$ cannot possibly affect $d_\dagger$'s menu.
  Thus, it is without loss of generality to take $d_* = d_0$ and $d_\dagger = d_n$.

  Now, the key lemma is the following:
  \begin{lemma}
    \label{lem:type-to-menu-SDrot-main-lemma}
    Fix a set of preferences $P$, and suppose $h$ is in $d_n$'s menu under the mechanism $\SD_{d_j, d_{j+1},\ldots,d_n}(P_j, P_{j+1},\ldots)$ for some $j$.
    Then, $h$ is in $d_n$'s menu under $\SD_{d_{j+1}, d_{j+2},\ldots,d_n}(P_{j+1}, P_{j+2},\ldots)$ as well.
  \end{lemma}
  To prove this lemma, consider the run of $\SD_{d_{j+1}, d_{j+2},\ldots,d_n}$ compared to $\SD_{d_j, d_{j+1},\ldots,d_n}$. Similarly to the proof of     \autoref{thrm:gtc-SD-UB-key-lemma-necessary},
  observe that when each applicant $d_k$ for $j < k \le n$ picks from her menu, she can have only fewer options in $\SD_{d_j, d_{j+1},\ldots,d_n}$ than in $\SD_{d_{j+1}, d_{j+2},\ldots,d_n}$. In particular, this applies to $d_n$, proving
  \autoref{lem:type-to-menu-SDrot-main-lemma}.

  Thus, to represent the function $g(\cdot)$, we claim that it suffices write down an ordered list $S_1, S_2, \ldots, S_n \subseteq \Insts$ defined as follows: For each $i>1$, the set $S_i \subseteq \Insts$ is the subset of institutions that are on $d_n$'s menu in $\SD_{d_{i+1},d_{i+2},\ldots,d_n}$ but not in $\SD_{d_{i},d_{i+1},\ldots,d_n}$. For $i=1$, set $S_1$ is the menu of $d_n$ in $\SD_{d_{1},d_{2},\ldots,d_n}$. Then, \autoref{lem:type-to-menu-SDrot-main-lemma} implies that for any $P_{d_*}$, if $j_*$ is such that applicant~$0$ picks $h^{\mathsf{rot}}_{j_*}$, we have $g(P_{d_*}) = \bigcup_{j=1}^{j_*} S_j$.
  Moreover, \autoref{lem:type-to-menu-SDrot-main-lemma} implies that each institution can appear in at most one set $S_i$, so we can represent this list in $\widetilde O(n)$ bits, as claimed.
\end{proof}

\subsection{Additional Remarks}
\label{sec:additional-remarks}

Here, we make some additional remarks relevant to representation and verification protocols in \autoref{sec:compression-complexity} and \autoref{sec:verif-complexity}.

\begin{remark}
  \label{remark:score-based-model}
  While our model (and our lower bounds) assume that the entire priority lists of each institution are known to all applicants, all of the protocols we mention or construct under the model of \autoref{sec:protocol-models} can actually be implemented with a more mild assumption on the knowledge of the priorities. 
  Assume that there are scores $e_h^d \in [0,1]$  for each $d\in\Appls$ and $h\in\Insts$, and let $d \succ_h d'$ if and only if $e_h^d > e_h^{d'}$.
  Assume that each applicant $d$ starts off knowing $\bigl\{ e_h^d \bigr\}_{h \in \Insts}$, but any other information must be communicated through the certificate $C$.
  First, observe that the protocols adapted from \cite{AzevedoL16} and \cite{LeshnoL21} in \autoref{sec:compression-complexity} work exactly as written.
  
  Now, consider the verification protocol for $\TTC$ from \autoref{thrm:verif-complexity-ttc}. This verification protocol can still post a transcript of the deterministic protocol for $\TTC$ from \autoref{thrm:deterministic-cc-ttc}, since that deterministic protocol works even when the priorities must be communicated. However, note that the institutions are not agents participating in the protocol, so the institutions cannot themselves take part in the verification of the transcript of the protocol. Thus, every time that an institution $h$ points to an applicant $d$ during the transcript, the protocol should also announce the priority score of applicant $d$ at institution $h$; and each applicant who is not yet matched in the transcript should check that they do not have higher priority at $h$. This suffices to verify the transcript of the protocol (though interestingly, it's no longer clear how to make this protocol deterministic).

  For $\DA$, an analogous trick works. When the protocol to verify $\IPDA$ posts certificate $\mu, \ImprGrI$, it should also post the priority scores of $\mu(h)$ at $h$ for each $h \in \Insts$, and applicants will still be able to verify each edge $d \xrightarrow{h} d'$ in $\ImprGrI$, as in \autoref{lem:exist-predicates-checking-graph}.
  When the protocol to verify $\APDA$ posts certificate $\mu, \ImprGrA$, it should also post, for every edge $h\xrightarrow{d}h'$, the priority scores of $d$ at both $h$ and $h'$. Then, again each applicant knows what she needs in order to perform the verification as in \autoref{lem:exist-predicates-checking-graph}.
\end{remark}

\begin{remark}
  While \autoref{thrm:verif-complexity-da} holds for both $\APDA$ and $\IPDA$, it cannot be extended to any matching mechanism. To see why, consider a market with applicants $d^0_1,\ldots,d^0_n, d_1^1,\ldots,d_n^1$ and institutions $h_1^0,\ldots,h_n^0, h_1^1,\ldots,h_n^1$, where each applicant and institution has a full-length preference list. Suppose each institution $h_i^j$ for $i\in\{1,\ldots,n\}$ and $j \in \{0,1\}$ ranks $d_i^j$ first and ranks $d_i^{1-j}$ last, 
  with any fixed order between them. Now consider the set of preferences where each applicant $d_i^j$ for $i\in\{1,\ldots,n\}$ and $j \in \{0,1\}$ ranks $h_i^{1-j}$ first and ranks $d_i^{j}$ last, where any possible ordering over the other institutions appears between first and last. Observe that, for any profile of preferences in this class, the matchings $\mu_0 = \bigl\{ (d_i^j, h_i^j) \bigr\}_{i \in \{1,\ldots,n\}, j \in \{0,1\}}$ and $\mu_1 = \bigl\{ (d_i^j, h_i^{1-j}) \bigr\}_{i \in \{1,\ldots,n\}, j \in \{0,1\}}$ are both stable. Thus, consider a stable matching mechanism $f$ that (on inputs in this class) outputs either $\mu_0$ or $\mu_1$, depending on some arbitrary function $g : \T \to \{0,1\}$ of the preference lists of the $2n$ applicants, and suppose that the function $g$ has verification complexity $\Omega(n^2)$ (such a function can be constructed by standard techniques using a counting argument). Then, the verification complexity of $f$ will be at least the verification complexity of $g$, i.e., $\Omega(n^2)$.
\end{remark}

\begin{remark}
\label{remark:unbalanced-bounds}
We now remark on our complexity measures from \autoref{sec:gtc} and \autoref{sec:type-to-menu} in the case of many-to-one markets with $|\Appls|\gg|\Insts|$ (considered in \autoref{sec:compression-complexity} and \autoref{sec:verif-complexity}).
In this regime, some of our constructions leave a gap between a lower bound of $\Omega\bigl(|\Insts|^2\bigr)$ and the trivial upper bound of $\widetilde O\bigl(|\Insts|\cdot|\Appls|\bigr)$.\footnote{
  None of the bounds we present in \autoref{sec:gtc}, \autoref{sec:bit-complexity-menus}, or \autoref{sec:all-type-to-one-match} are tight in the $|\Appls|\gg|\Insts|$ case.
  For our options-effect lower bounds in \autoref{sec:type-to-menu}, one can show that our results are already tight: $\widetilde O(|\Insts|^2)$ suffices for $\TTC$ by nonbossiness, and $\widetilde O\bigl(|\Insts\bigr|)$ suffices for $\DA$ because the un-rejections graph $\UnrejGr$ in \autoref{sec:type-to-menu-DA} has at most one node for each element of $\{d_*,d_\dagger\}\times \Insts$ (and our proofs in that section hold as written for many-to-one markets). 
}
Still, each such bound of $\Omega\bigl(|\Insts|^2\bigr)$ gives a qualitative negative result, since we think of $\Omega\bigl(|\Insts|^2\bigr)$ as large/complex. 

We also remark that almost none of our results in this paper are tight in terms of the exact $\log$ factors. While these $\log$ factors are small (corresponding to the need to index a single applicant/institution), deriving bounds that are exactly asymptotically tight may require very different constructions that would allow the order of applicants' preference lists to vary much more dramatically than our constructions do. 
While this might possibly lead to interesting technical challenges, we believe that our results carry the main economic and complexity insights for each of the questions that we ask.
\end{remark}

\section{Additional Preliminaries}
\label{sec:additional-prelims}
\paragraph{Many-to-one matching markets.}

When we study many-to-one matching rules (particularly relevant in \autoref{sec:compression-complexity}), each $h \in \Insts$ is equipped with a fixed capacity $q_h \ge 1$, and we require that all matchings $\mu\in \M$ satisfy $\bigl|\mu(h)\bigr| \le q_h$ for each $h\in\Insts$.
To define each of the mechanisms we consider ($\TTC$, $\APDA$, $\IPDA$, and $\SD$) in many-to-one markets, one can use the following standard trick: For each institution $h_i \in \Insts$ with capacity $q_h$, define $q_h$ distinct institutions $h_i^1, h_i^2,\ldots, h_i^{q_h}$, each with capacity~$1$ and a priority list identical to $h_i$'s list in $Q$, and replace each $h_i\in\Insts$ on each preference list of each applicant $d\in\Appls$ with $h_i^1\succ h_i^2\succ\ldots\succ h_i^{q_h}$. Then, each mechanism is defined as the outcome in this corresponding one-to-one market.
In economics nomenclature, such preferences are called responsive preferences. 

\paragraph{ }

We also occasionally consider the classically studied property of nonbossiness:
\begin{definition}
  \label{def:nonbossy}
  A mechanism $f$ is \emph{nonbossy} if, for all $d\in\Appls$, all $P_d, P_d' \in \T_d$, and all $P_{-d}\in\T_{-d}$, we have the following implication:
  \[ f(P_d, P_{-d}) \ne f(P_d', P_{-d})
     \quad\implies\quad
     f_d(P_d, P_{-d}) \ne f_d(P_d', P_{-d}).
  \]
  That is, if changing $d$'s report changes some applicant's match, then it in particular changes $d$'s own match.
\end{definition}
$\TTC$ and $\IPDA$ are nonbossy, but $\APDA$ is bossy.

We now give some well-known properties of $\TTC$ and $\DA$, which we use throughout our paper. 

\paragraph{Properties of top trading cycles.}

\begin{lemma}[Follows from \cite{ShapleyS74}; \cite{RothP77}]\label{thrm:TTC-order-independent}
  The matching output by the $\TTC$ algorithm in \autoref{def:TTC} is independent of the order in which trading cycles are chosen and matched.
\end{lemma}

\begin{lemma}[\cite{ShapleyS74, Roth82-TTC}]
  \label{claim:TTC-strategyproof}
  $\TTC$ is strategyproof and nonbossy.
\end{lemma}

\paragraph{Properties of stable matching mechanisms.}

\begin{lemma}[\cite{GaleS62}]
  \label{claimDaStable}
  The output of $\APDA$ (or $\IPDA$) is a stable matching.
\end{lemma}

\begin{corollary}[\cite{DubinsMachiavelliGaleShapley81}]
  \label{thrm:da-indep-execution}
  The matching output by the $\APDA$ algorithm in \autoref{def:DA} is independent of the order in which applicants are selected to propose.
\end{corollary}

\begin{corollary}[\cite{GaleS62, McVitieW71}]
  \label{thrm:apda-apps-at-best}
  In the matching output by $\APDA$, every applicant is matched to her favorite stable partner.
  Moreover, each $h\in \Insts$ is paired to her worst stable match in $\Insts$.
  The same holds in reverse for $\IPDA$.
\end{corollary}

We also need the following classical characterization relating the set of matched agents in each stable outcome.

\begin{theorem}[Lone Wolf / Rural Hospitals Theorem, \cite{RothRuralHospital86}]
  \label{thrm:rural-hospitals}
  The set of unmatched agents is the same in every stable matching.
  Moreover, in a many-to-one stable matching market, if there is some stable matching where an institution receives fewer matches than its capacity, then that institution receives precisely the same set of matches in every stable outcome.
\end{theorem}

For $\DA$ and all stable matching mechanisms, we make the following observation about the menu:
\begin{lemma}\label{thrm:same-menu-all-stable}
  If $f$ and $g$ are any two stable matching mechanisms (with respect to priorities $Q$), then $\Menu_d^f(P_{-d}) = \Menu_d^g(P_{-d})$ for all $d, P_{-d}$.
\end{lemma}
\begin{proof}
  Consider any $h \in \Menu_d^f(P_{-d})$, and let $\mu = f(P_d, P_{-d})$ for any $P_d$ such that $\mu(d)=h$. Then $\mu$ will also be stable under preferences $P'$ that are the same as $P = (P_d,P_{-d})$, except that $d$ ranks only $\{h\}$, and thus by \autoref{thrm:rural-hospitals}, we must have $\mu'(d) = h = \mu''(d)$, where $\mu' = f(P')$ and $\mu'' = g(P')$. Thus, $h \in \Menu_d^g(P_{-d})$, and by symmetry $\Menu_d^f(P_{-d}) = \Menu_d^g(P_{-d})$.
\end{proof}
By this lemma, when we give bounds related to the complexity of the menu in stable matching mechanisms, the distinction between $\APDA$ and $\IPDA$ is not important.

We have the following strategyproofness result for $\APDA$:
\begin{theorem}[\cite{Roth82-DA,DubinsMachiavelliGaleShapley81}]\label{thrm:sp-apda}
  $\APDA$ is strategyproof (for the applicants).
\end{theorem}

Note that $\IPDA$ is not strategyproofness (and thus, while we sometimes discuss the menu in $\IPDA$, an applicant is not matched in $\IPDA$ to her top choice from her menu). However, while $\APDA$ is bossy, we observe that $\IPDA$ is not:

\begin{proposition}\label{thrm:nonbossy-ipda}
  $\IPDA$ is nonbossy.
\end{proposition}
\begin{proof}
  Consider any $Q, P$ where $\mu = \IPDA_Q(P)$, and consider any $\mu(d)=h$. Let $P'$ be identical to $P$, except that $d$ truncates her preference list after $h$ (i.e., all institutions strictly below $h$ are marked unacceptable). 
  By \autoref{thrm:apda-apps-at-best}, no stable partner of $d$ was below $h$ on $d$'s list, so the set of stable matchings is identical under $P$ and $P'$, and in particular $\mu = \IPDA_Q(P')$. Now let $P''$ be identical to $P'$, except $d$ truncates her preference list \emph{above} $h$ (i.e., now $d$'s list consists only of $\{h\}$). 
  Since $d$ never receives a proposal from any institution above $h$ during the run of $\IPDA_Q(P')$, this run of $\DA$ is identical to $\IPDA_Q(P'') = \mu$.
  But, $P''$ is identical for any initial value of $P_d$ such that $\mu(d)=h$, so $\mu$ must be identical for any such $P_d$, which finishes the proof.
\end{proof}

\section{Relation to Other Frameworks and Results}
\label{sec:related-simplicity}

\subsection{Verification of \texorpdfstring{$\DA$}{DA}}
\label{sec:related-cutoff-DA}

This section discusses how the results of \cite{AzevedoL16} and \cite{Segal07} relate to the results of \autoref{sec:compression-complexity} through \autoref{sec:bit-complexity-menus} for $\DA$.

\cite{Segal07} is concerned with the verification of social-choice functions and social-choice correspondences. Social-choice functions are the natural generalization of matching rules to scenarios other than matching, such as voting or auctions. Social-choice correspondences are generalizations of social-choice functions, where there can be multiple valid outcomes that correspond to the same inputs of the players. For example, ``the $\APDA$ outcome'' is a social-choice function, while ``a stable outcome'' is a social-choice correspondence. 
\cite{Segal07} shows that, for a large class of social-choice correspondences including the stable matching correspondence, the verification problem (which differs from \autoref{def:verification-protocol} only in that every agent should know the full outcome of the protocol) reduces to the task of communicating ``minimally informative prices.''

When \cite{AzevedoL16} relate their work---in particular the cutoff representation of the matching of $\DA$, which we discussed in the proof of \autoref{thrm:representing-DA}---to the work of \cite{Segal07}, they mention that their cutoffs coincide with \cite{Segal07}'s ``minimally informative prices.'' Since \cite{Segal07} shows that these prices verify that a matching is stable, this may seem to imply that the $\widetilde\Theta\bigl(|\Insts|\bigr)$-bit vector of cutoffs should suffice to verify that a matching is stable, contradicting our \autoref{thrm:verif-lb-both-mechs}, which says that $\Omega\bigl(|\Appls|\bigr)$ bits are needed (even to verify that a matching is stable). However, there is no actual contradiction: the model and characterization in \cite{Segal07} implicitly assume that each agent in the verification protocol knows the complete matching, ruling out the $\widetilde\Theta\bigl(|\Insts|\bigr)$-bit representation protocol implicitly discussed in \cite{AzevedoL16} and formalized by our \autoref{thrm:representing-DA}. Indeed, as we discussed in \autoref{sec:verif-complexity}, simply writing down the matching suffices to verify in our model that a matching is stable using $\widetilde\Theta\bigl(|\Appls|\bigr)$ bits (though verifying that the matching is the outcome of $\APDA$ or $\IPDA$, as we do in \autoref{thrm:verif-complexity-da}, is more involved). 

There is also a strong conceptual relation between our results in \autoref{sec:verif-complexity} for $\DA$ and the results of \cite[Section 7.5]{Segal07} and \cite{GonczarowskiNOR19}. Namely, all of these results bound the complexity of verifying stable matchings. However, there is a large technical difference between our model and prior work: we treat the priorities as fixed and known by all applicants. This difference turns out to have dramatic implications: both \cite{Segal07} and \cite{GonczarowskiNOR19} achieve $\Omega(n^2)$ lower bounds (in different models) for the problem of verifying a stable matching where $n = |\Appls| = |\Insts|$; we get an upper bound of $\widetilde O(n)$ (\autoref{thrm:verif-complexity-da}). 

There is also a conceptual connection between our work and \cite{HakimovR23}, who construct protocols achieving variants of both concurrent representation (as per our \autoref{sec:compression-complexity}, which they term ``verification'') and joint verification (as per theoretical computer science notions and our \autoref{sec:verif-complexity}, which they term ``transparency''). Their positive results for concurrent representation match the protocols constructed by \cite{AzevedoL16, LeshnoL21} in terms of communication. Their main positive result constructs a certain type of $\widetilde O(n^2)$-bit interactive verification protocol for $\DA$, under the additional assumption that the mechanism designer can never deviate in a way that might leave a seat unfilled at some demanded school (formally, they assume any mechanism to which the designer might deviate is non-wasteful; interaction is needed to make use of their non-wastefulness assumption).
Due to this additional assumption in \cite{HakimovR23}, our protocol in \autoref{thrm:verif-complexity-da} achieves a strictly stronger type of verification, as well as lower communication cost (which they do not attempt to minimize).
The highest-level distinction between \cite{HakimovR23} and our work is that \cite{HakimovR23} are focused on constructing (potentially practically-useful) protocols, but do not attempt to minimize communication complexity or provide any lower bounds. In contrast, our paper provides both new protocols and impossibility results getting tight bounds on the communication complexity of different protocols.
Within an auction environment, \cite{Woodward20} studies a related set of questions conceptually similar to \cite{HakimovR23} (and to our \autoref{sec:compression-complexity}), while factoring in the incentives of the auctioneer.

\subsection{Cutoff Structure of \texorpdfstring{$\TTC$}{TTC}}
\label{sec:related-cutoff-TTC}

This section discusses how the results of \cite{LeshnoL21} relate to our results for $\TTC$, especially in \autoref{sec:compression-complexity}. 

As mentioned in the proof of \autoref{thrm:representing-TTC}, \cite{LeshnoL21} prove that the outcome of $\TTC$ can be described in terms of $|\Insts|^2$ ``cutoffs,'' where in the language of our paper, a cutoff is simply an index on an institution's priority list. Based on examples and the intuition that their description provides, they state that ``the assignment cannot [in general] be described by fewer than $\frac 1 2 n^2$ cutoffs.'' However, they do not formalize what is precisely meant by this claim.

\cite{LeshnoL21} briefly discuss one formal result in the direction of an impossibility result. We attempt to capture the essence of their approach as follows:
\begin{proposition}[{Adapted from \cite[Footnote 8]{LeshnoL21}}]
  \label{prop:LL-TTC-LB-footnote}
  In $\TTC$, for any $n$, there exists a market with fixed priorities $Q$, institution $h$, and applicants $U = \{ d_1, \ldots, d_n \} \subseteq \Appls$, where $n = \Theta\bigl(|\Insts|\bigr) = \Theta\bigl(|\Appls|\bigr)$, with the following property:
  Let $P_U^*$ denote preferences of applicants in $U$ such that each $d \in U$ ranks only $\{h\}$. Then, for any $S \subseteq U$, there exist priorities $P_{-U}$ of applicants outside $U$ such that if $\mu = \TTC_Q(P_U^*, P_{-U})$, then for each $d \in S$ we have $\mu(d)=h$, and for each $d\in U\setminus S$, we have $\mu(d)\ne h$.

  On the other hand, in $\APDA$, suppose there is any $Q, h,$ and set $U = \{d_1,\ldots,d_k\}\subseteq \Appls$ with the following property: 
  For any $S \subseteq U$, there exist priorities $P_{-U}$ such that if $\mu = \APDA_Q(P_U^*, P_{-U})$, we have $\mu(d)=h$ for each $d\in S$ and $\mu(d)\ne h$ for each $d\notin S$.
  Then we have $k = |U| \le 1$.
\end{proposition}
\begin{proof}
  For $\TTC$, consider a market with applicants $d_i, d_i'$ for $i=1,\ldots,n$ and institutions $h_i$ for $i=0,1,\ldots,n$. Now, let institution $h_0$ have capacity $n$ and priority list $d_1'\succ\ldots\succ d_n' \succ d_1 \succ d_n$, and for each $i=1,\ldots,n$ let $h_i$ have capacity $1$ and have priority list $d_i \succ d_i' \succ L$, where $L$ is arbitrary. Let each $d_i$ submit list $h_0 \succ h_i$, and consider the $2^n$ preference profiles where each $d_i'$ might submit list $h_0 \succ h_i$, or might submit list $h_i \succ h_0$. Then, by the way $\TTC$ matches applicants, $d_i$ will match to $h_0$ if and only if $d_i'$ submits list $h_i \succ h_0$, as desired.

  For $\APDA$, consider any fixed $Q$, $h$, and $\{d_1,d_2\}\subseteq U$ where (without loss of generality) $d_1\succ_h^{Q_h} d_2$. For any $P$ such that $d_1$ and $d_2$ both rank only $h$, if $\mu$ is stable and $\mu(d_1) = h$, then we must have $\mu(d_2) = h$, by stability.
  This concludes the proof (and in fact shows that the claim holds for any stable matching mechanism, not just $\APDA$).
\end{proof}

While \autoref{prop:LL-TTC-LB-footnote} provides an interesting formal sense in which $\TTC$ (but not $\DA$) is complex, this approach only considers admission to one institution $h$ at a time, and correspondingly to our understanding might only conceivably show an $\Omega(n)$ lower bound and not the desired $\Omega(n^2)$ lower bound that we formalize in \autoref{sec:compression-complexity}.

\begin{remark}
  In contrast to \autoref{prop:LL-TTC-LB-footnote}, \cite[Footnote 8]{LeshnoL21} from which it is adapted is phrased in terms of separately defined definitions of budget sets for each of $\TTC$ and $\APDA$. For $\TTC$, the budget sets are defined in \cite{LeshnoL21}; see \cite{LeshnoL21} for the precise definition (whose details are not required to follow this discussion). While neither \cite{LeshnoL21} nor the most closely related work \cite{AzevedoL16} seem to explicitly define the budget set in $\APDA$, we confirmed with the authors of \cite{LeshnoL21} (private communication, October 2022) that they intended the budget set of $d$ in $\APDA$ to be the sets that we denote as $\StabB_d\bigl(\APDA_Q(P)\bigr)$ (see \autoref{def:stable-budget-set}).
  Note that neither of these sets equals the menu in the corresponding mechanism. This was observed for $\TTC$ in \cite[Footnote 10]{LeshnoL21} and for $\APDA$ in the beginning of our \autoref{sec:bit-complexity-menus} (and in \cite{GonczarowskiHT22}).
  Our \autoref{thrm:all-menus-DA} directly implies that the result in \cite[Footnote 8]{LeshnoL21} for $\DA$ would no longer hold if the budget set $\StabB_d$ were replaced with the menu $\Menu_d^\APDA$.
  We rephrased the insights of \cite[Footnote 8]{LeshnoL21} to produce \autoref{prop:LL-TTC-LB-footnote}, which replaces the notion of an institution $h$ being in a budget set of applicant $d$ with the question of whether $d$ matches to $h$ when $d$ only ranks $h$. We did this in order to compare the mechanisms $\DA$ and $\TTC$ through a single complexity lens---which is perhaps more consistent with the rest of the current paper---rather than using budget sets, which despite their appealing properties, to our knowledge must be defined separately for each of the two mechanisms.
\end{remark}

We also remark that \cite{LeshnoL21} occasionally use the word ``verification,'' and speak of agents ``verifying their match.'' However (like \cite{HakimovR23} but unlike \cite{Segal07}), to our best understanding they seem to use this term in an informal sense of checking that an assignment is accurate, rather than as any sort of verification according to formal computer-science notions.

\end{document}